\documentclass[11pt,english]{article}
\usepackage[sc]{mathpazo}
\usepackage{geometry}
\usepackage{color}
\usepackage[dvipsnames]{xcolor}
\definecolor{wxycolor}{RGB}{230,120,0}
\usepackage{array}
\usepackage{bbding}
\usepackage{multirow}
\usepackage[fleqn]{amsmath}
\usepackage{autobreak}
\usepackage{amssymb}
\usepackage{amsthm}
\usepackage{graphicx}
\usepackage{natbib}
\usepackage{lscape}
\usepackage{comment}
\usepackage{bm}
\usepackage[title]{appendix}
\usepackage{longtable}
\usepackage{mathtools}
\usepackage{chngcntr}
\usepackage{authblk}
\usepackage{enumerate}
\usepackage{subcaption} 
\usepackage{float} 
\makeatletter
\allowdisplaybreaks

\usepackage[colorlinks,
            linkcolor=blue,
            anchorcolor=blue,
            citecolor=blue
            ]{hyperref}

\usepackage{babel}
\usepackage{enumitem}

\newtheorem{Lem}{Lemma}
\newtheorem{prop}{Proposition}

\newtheorem{assum}{Assumption}

\begin{document}

\nocite{*}
\title{
Sharing economy in the era of full automation: Evidence from autonomous vehicle on-demand mobility services 
}
\author{Authors' names blinded for peer review}
\author{Xiaoyan Wang\textsuperscript{a}\hspace{1em}  Kenan Zhang\textsuperscript{a,}\footnote{Corresponding author. E-mail address: \textcolor{blue}{kenan.zhang@epfl.ch}.}\hspace{1em}Yaochen Ma\textsuperscript{b}}

\date{}
\maketitle 

\begin{abstract}

The digital age has facilitated the sharing of underutilized assets. This paper focuses on privately owned autonomous vehicles (AVs), a unique class of robots that can move independently and provide transportation services. When not in personal use, private AV owners can lease their vehicles to a platform that operates an on-demand mobility service (MoD). We refer to this service as AV crowdsourcing, and develop a time-expanded network flow model that captures temporal and spatial heterogeneity in AV usage of both owners and passengers while preserving analytical tractability. 
We analyze the conditions under which AV crowdsourcing reduces MoD operating costs and identify their key factors, namely, the complementarity of the mobility pattern between AV owners and MoD passengers, the slack time reserved by vehicle owners, 
and the vehicle repositioning distance.  
A case study of Chicago further reveals substantial spatiotemporal heterogeneity in optimal prices and service quality. The results demonstrate how centralized dispatching can simultaneously fulfill the high demand in downtown areas while maintaining relatively high service quality in peripheral regions. 
Our findings provide insights into how supply heterogeneity and market conditions jointly shape the performance of AV crowdsourcing systems that leverage the underutilized private robotic assets. 

\par
\hfill\break%
\end{abstract}

\section{Introduction}

On-demand sharing economy has experienced rapid growth in recent years and has reshaped a variety of industries. Enabled by digital technologies, it improves the utilization of underused resources, generates new economic opportunities, and expands access to services. On-demand sharing services typically rely on dedicated platforms to coordinate and manage service operations. These operations are inherently spatiotemporal, as shared resources are allocated across locations and their usage must be scheduled to avoid conflicts. A prominent application domain is on-demand mobility (MoD) services, which have been expanded rapidly worldwide over the past decade. In 2024, Uber operated in more than 15,000 cities worldwide \citep{Uber2024AnnualReport}, while Uber, Lyft, and DiDi collectively completed more than 28 billion trips \citep{HALL2026104221}.

Advances in artificial intelligence have contributed to the development of autonomous vehicles (AVs), which are widely regarded as a promising future mobility solution~\citep{hu2025multimodal,LIN2026104762}, owing to their potential to enhance traffic efficiency, improve safety and reduce human labor~\citep{Othman2022}. 
As a particular application, shared autonomous vehicles, also known as robo-taxis, are being deployed in several countries. 
For example, Baidu Apollo Go has deployed robo-taxis in several Chinese cities, including Wuhan and Beijing~\citep{Yu2026Robotaxis}. Waymo reported over 170 million cumulative autonomous miles across its fleet operating in multiple U.S. cities, including the San Francisco Bay Area and Phoenix, as of December 2025~\citep{hawkins2026waymo170miles}.
The autonomous driving technology provides an additional opportunity for sharing private vehicles to improve both vehicle utilization (which is currently around 10\%~\citep{ZHANG2025103305}) and enhance the MoD service quality with additional vehicle supply. 
\citet{WANG2021103362} refers to the resulting business model, where the ride-hailing platforms leverage idle private AVs to provide MoD services, as AV crowdsourcing.
As illustrated in Figure \ref{fig_robotaxi}, private AVs switch between owner usage and MoD service, while their repositioning and trip assignments across time and space are coordinated by the MoD platform. 
\begin{figure}[H]
    \centering
    \includegraphics[width=0.8\linewidth]{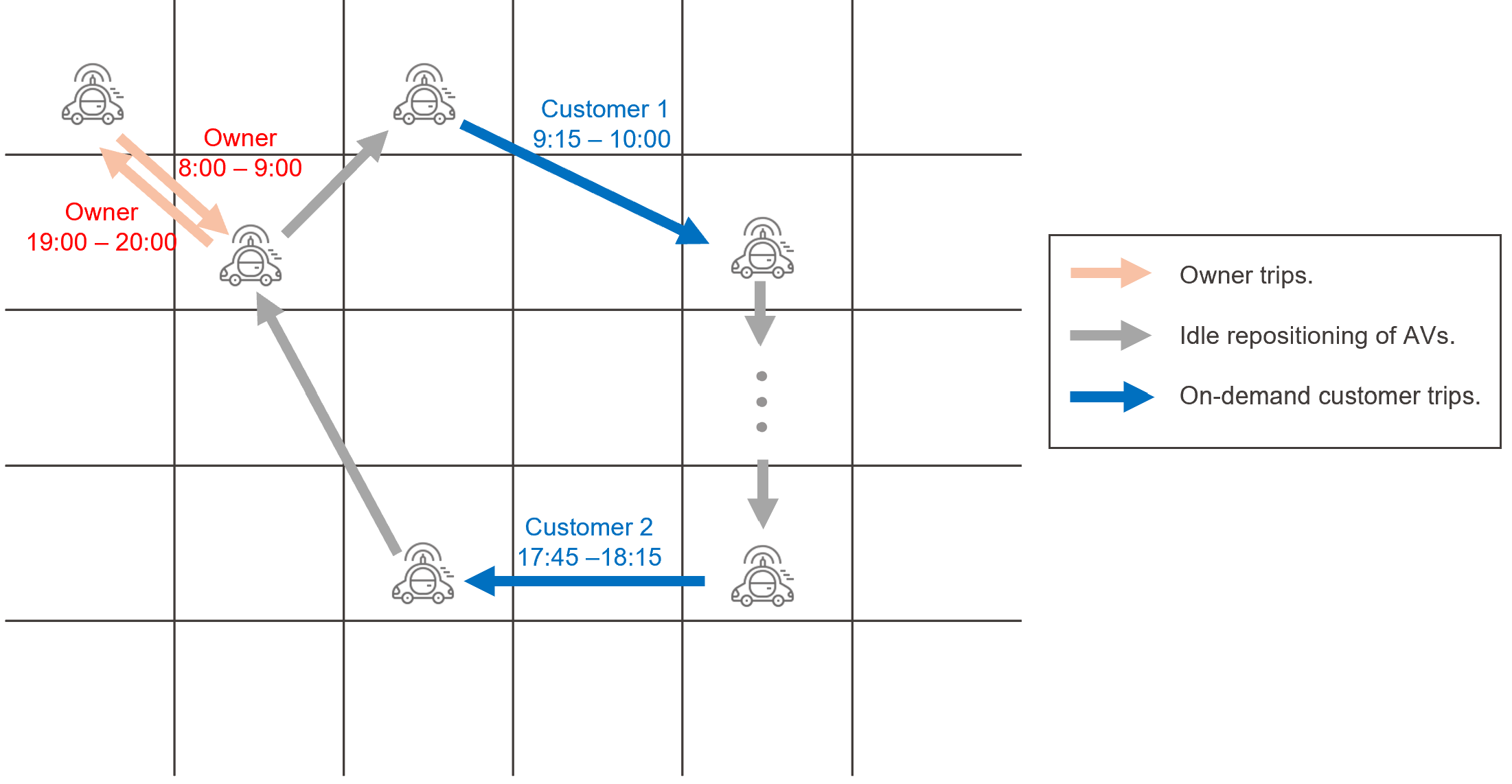}
    \caption{Shared use of private autonomous vehicles}
    \label{fig_robotaxi}
\end{figure}

As a foreseeable example of on-demand sharing of robotic assets, AV crowdsourcing services warrant careful investigation prior to practical deployment. Several key questions remain to be addressed: (1) Under what conditions can AV crowdsourcing services benefit the platform? (2) What key metrics should the platform consider when operating the service? (3) How does centralized dispatching affect service quality? 
Answering these questions involves modeling the 
the spatiotemporal heterogeneity of private AV availability, which distinguishes this work from previous studies. 
Overall, our main contributions are as follows:
\begin{itemize}
    \item We develop a time-expanded network flow model to jointly optimize not only AV crowdsourcing payment strategies but also the operations for both crowdsourced and platform-owned AVs. Under the assumption of a uniform distribution of AV owners' opportunity costs, the resulting optimization problem can be reformulated as a convex quadratic program, which guarantees the existence of a global optimum and may admit multiple optimal strategies.
    \item Building on this framework, we analyze the feasibility conditions under which the MoD platform benefits from AV crowdsourcing, and investigate strategies for selecting participating AV owners across heterogeneous regions.
    \item Through numerical experiments on a two-node network, we identify several key metrics that have significant impacts on service scale. The service scale increases as the travel time gap between AV owners and customers decreases, the reliability buffer for owners' private use is reduced, the customer loss penalties increase, and inter-regional travel times decrease. 
    \item We validate the analytical results through a case study constructed using the open-source data of Chicago and generate additional insights into the real practice of AV crowdsourcing.
    Results reveal that demand and owner availability exhibit significant spatial heterogeneity across regions. Although empty vehicles tend to concentrate in high-demand areas, some peripheral areas may experience good service quality due to lower traveler competition.  
\end{itemize}

The remainder of this paper is organized as follows. Section \ref{sec_liter} reviews the related studies on sharing economy and crowdsourcing, with a particular focus on MoD services. Section \ref{sec_model} presents the network flow model that serves as the foundation of the service capacity planning problem formulated in Section \ref{sec_property}, along with the analytical properties. Section \ref{sec_num} details the numerical experiments on a stylized network and examines the effects of various factors on platform strategy and service performance. The design and main findings of the Chicago case study are summarized in Section \ref{sec_chicago}. 
Finally, Section \ref{sec_concl} concludes the paper with a discussion of implications beyond AV crowdsourcing and provides several directions for future research.

\section{Related works}\label{sec_liter}

The sharing economy has been transforming the industry and market, and has raised unprecedented research questions.  
For instance, \citet{doi:10.1111/poms.13491} and \citet{doi:10.1111/poms.13883} studied how the sharing economy reshapes market competition by changing consumers' decisions between buying and renting, thereby affecting manufacturers’ profits. \citet{doi:10.1177/10591478251331407} compared two distinct sharing strategies, one is peer-to-peer (P2P) sharing and the other is business-to-consumer (B2C) sharing. While P2P sharing creates a scale effect that promotes market expansion, B2C sharing offers greater control over pricing, which facilitates the implementation of price discrimination. Other studies focus on the pricing and service design problems in the sharing economy. 
For example, \citet{doi:10.1177/10591478251326390} examines pricing and service assortment decisions for service providers offering both group and individual services under heterogeneous customer waiting behavior.

Previous studies have examined the financial feasibility of sharing economy in various industry fields, such as the rental channels in the fashion industry~\citep{doi:10.1177/10591478261444827}, and the subscription sharing of general digital services~\citep{doi:10.1177/10591478251404236}.
Some scholars have extended economic analysis of the sharing economy to account for characteristics that vary across temporal and spatial dimensions. 
For example, \citet{doi:10.1111/poms.12672} and \citet{doi:10.1177/10591478251319683} studied the resource sharing subject to usage time constraints. \citet{doi:10.1177/00222437241255057} studied the spatiotemporal diffusion of a car-sharing platform and finds that growth is driven by asymmetric spatial network effects shaped by both consumer mobility and supplier distribution. However, few studies incorporate both spatial and temporal operational constraints while maintaining analytical tractability. Such considerations are particularly important in transportation services, where demand and supply are inherently heterogeneous across both space and time~\citep{LUO2024103936, XU2025101533,doi:10.1177/10591478251349724}. 
For example, the operation of on-demand mobility (delivery) services face the challenge of dispatching drivers (couriers) to pick up passengers (orders) in real-time~\citep{9907878,chen2024real,JIANG2026104505}.  
Compared to these services, AV crowdsourcing has distinct operational constraints: private AV can only serve MoD trips when not used by AV owners. Different from drivers and couriers in other on-demand services, AV owners are much more heterogeneous in their activity patterns, which translates into diverse vehicle availability. 
\citet{mourad2019owning} conducted a simulation study of AV crowdsourcing integrating these constraints, whereas our study presents an analytical modeling framework that jointly examines the AV rental pricing strategies and the MoD operational strategies.

MoD services with AVs or a mixed fleet of AVs and human-driven vehicles (HVs) have been studied extensively in the literature~\citep[e.g.,][]{Siddiq2021,10609802,doi:10.1287/trsc.2022.1188}, though most studies assume the AVs are owned by the platform. 
Differently, AV crowdsourcing considers that AVs are owned by individuals but centrally operated by the platform during their rental periods. 
Several recent studies have explored this special business model. \citet{WANG2021103362} studied AV crowdsourcing in a multimodal transportation market and investigated pricing strategies under profit maximization and social welfare maximization objectives. \citet{Siddiq2021} studied whether a human-driven ride-hailing platform can improve profitability by introducing crowdsourced AVs or purchased AVs. \citet{WANG2024104732} studied the market structure of AV crowdsourcing platforms.
\citet{ZHANG2025103305} analyzed the potential of a win-win-win outcome of AV crowdsourcing for AV owners, MoD passengers, and the platform.  
However, these studies all utilized aggregate models for tractability, neglecting the spatiotemporal characteristics in AV crowdsourcing, particularly the vehicle availability. \citet{LI2025103871} established a bi-level programming framework to optimize fixed-fleet sizing and pricing for crowdsourced AVs, captures the interactions between platform strategy and owners' activity-travel scheduling patterns. Different from their approaches, we adopt a time-expanded network flow model that preserves a degree of analytical tractability. Similar fluid-based approaches have also been adopted to analyze service capacity and staff planning problem~\citep{doi:10.1287/mnsc.1100.1203,doi:10.1287/opre.2019.1916}.

\section{Model} \label{sec_model}

We simultaneously optimize the platform's rental pricing strategies and AV dispatching strategies within a day. The set of regions is denoted as $\mathcal{J}$ of size $J=|\mathcal{J}|$. AV owners are assumed to rent out their private AVs for multiple units of intervals $\Delta_s$ (e.g., $\Delta_s=1$ hr). Accordingly, we use $h=1,\dots, H$ as the index of AV rental periods. For each rental period, the platform sets region-specific rental prices $P\in\mathbb{R}_{+}^{H\times J}$, with each element $P_{hj}$ representing the payment made to an AV owner who rents out the private vehicle originating from region $j$ during period $h$. On the demand side, the service horizon is divided into finer time slots of length $\Delta_t$ (e.g., $\Delta_t=15$ min). Similarly, we introduce $t=1,\dots, T$ to denote each service period. To make the two time scales compatible, we assume $\Delta_s=s\Delta_t$ with a constant integer multiplier $s$ (e.g., $\Delta_s=4\Delta_t$ with $s=4$).
Besides, we assume the travel times between regions are constant and collectively denoted by matrix $\pi$, where $\pi_{ij}$ is the travel time from region $i$ to region $j$ as multiples of $\Delta_t$.

\subsection{Vehicle supply}

AV owners are categorized into a finite number of classes, denoted by $\Theta$, according to their personal activity patterns during the day, where each class $\theta\in\Theta$ contains a total number of $M^\theta$ AV owners. 
We use an incidence matrix $\Omega^\theta\in \{0,1\}^{H\times J}$ to represent the spatiotemporal vehicle availability of AV owners in class $\theta$, where $\Omega^\theta_{hj} = 1$ if the vehicle can be rented out from region $j$ for period $h$, and 0, otherwise.

We assume AV owners perceive heterogeneous opportunity costs $\varepsilon$ that follow a known distribution with a positive finite support. To focus on the analysis of heterogeneous travel patterns, we assume the distribution of opportunity costs is homogeneous among AV owner classes. In the context of AV crowdsourcing, the opportunity cost can be interpreted as the minimum acceptable compensation. Hence, AV owners would only rent out their AVs if the spatiotemporal rental price $P_{hj}$ exceeds their opportunity costs, subject to availability. Accordingly, the spatiotemporal private AV supply of class $\theta$ from region $j$ over time period $h$ is given by 
\begin{align}\label{eq:private-av-supply}
    N^\theta_{hj} = M^\theta \mathbb{P}(P_{hj}\geq \varepsilon)\Omega^\theta_{hj}.
\end{align}
The resulting crowdsourced supply matrix is denoted by $N^\theta \in \mathbb{R}^{H\times J}$.

To convert the private AV supply from the rental periods to service periods, we introduce another incidence matrix $\Gamma\in \{0,1\}^{T\times H}$. The private AV supply of class $\theta$ measured by the service periods, denoted by $\tilde{N}^\theta\in \mathbb{R}^{T\times J}$, is thus given by 
\begin{align}
    \tilde{N}^\theta = \Gamma N^\theta.
\end{align}
We further define $Z^\theta \in \mathbb{R}^{T\times J}$ as the difference between private AV supply between consecutive time periods, and assume all rentals start and finish within the same day. Hence, the net private AV supply of class $\theta$ is given by
\begin{align}
    Z^\theta_{1\cdot} &= \tilde{N}^\theta_{1\cdot},\label{eq:net-private-av-supply-1}\\
    Z^\theta_{t\cdot} &= \tilde{N}^\theta_{t\cdot} - \tilde{N}^\theta_{t-1\cdot}, & t > 1,\label{eq:net-private-av-supply-t}
\end{align}
where $Z^\theta_{t\cdot}$ denotes the $t$-th row in matrix $Z^\theta$, and the same notation applies to other variables hereafter. 
In this study, we assume that all private AVs are rented and returned in the same region. Thus, a positive (negative) value in $Z^\theta$ indicates the number of newly rented (returned) private AVs of each region and time period. 

Besides renting private AVs, the platform also owns a fleet of $N^\nu$ AVs that operate over the entire service region and time horizon, where $\nu$ distinguishes these platform AVs from privates AVs and collectively yields the set of vehicle classes $\mathcal{K} = \Theta \cup \{\nu\}$. In this study, we consider $N^\nu$ to be exogenous to focus on the operational pricing and dispatching problems. Also note that $N^\nu$ is a scalar rather than a matrix as per the private AV supply $N^\theta$.

\subsection{Vehicle dispatching strategies}

The vehicle dispatching is expressed as a series of matrices $Y^k \in \mathbb{R}_+^{T\times J\times J}$, each of which indicates the flows of AVs in class $k$. Specifically, each element $Y^k_{tij}$ denotes the vehicle flow from region $i$ to region $j$ during time period $t$. 
Since travel times between regions are not uniform, we construct the arrival matrix $V^k$ in line with the departure matrix $Y^k$ as follows:
\begin{align}\label{eq:arrival-flow}
    & V^k_{\tau(t)\cdot\cdot} = Y^k_{t\cdot\cdot}, & \forall k\in \mathcal{K},\; \forall t, 
\end{align}
where the mapping of time period $\tau: \mathbb{R}^{J\times J}\to \mathbb{R}^{J\times J}$ is defined as 
\begin{align}\label{eq:travel_time_mapping}
    \tau_{ij}(t) = 
    \begin{cases}
        t+\pi_{ij}, & \text{if}\ t+\pi_{ij}\le T, \\
        t+\pi_{ij}-T, & \text{otherwise}.
    \end{cases}
\end{align}
It thus allows us to consider periodic operations of AVs without specifying the starting and ending times. The inverse mapping $\tau^{-1}$ can be easily derived as per Eq.~\eqref{eq:travel_time_mapping}.

We are now ready to present the constraints on the vehicle dispatching flows $Y$. 
The first condition shared among all vehicle classes is that the total dispatching flow from each region must not exceeds the total available vehicles. Let $E^k\in\mathbb{R}_+^{T\times J}$ denote the empty vehicle flow of class $k$ at the beginning of each time period in each region. The first constraint is thus given by 
\begin{align}
    &  Y^k_{t\cdot\cdot}\mathbf{1}  \leq E^k_{t\cdot}, & \forall k\in \mathcal{K}, \; \forall t,\label{eq:feasible-dispatch-flow}
\end{align}
where $\mathbf{1}$ denotes an all-ones vector of dimension $J$.

The second constraint expresses the flow dynamics of private AVs in each region between two consecutive time periods:
\begin{align}
    & E^\theta_{t\cdot} =  E^\theta_{t-1\cdot} + (V^\theta_{t-1\cdot\cdot})^\top \mathbf{1}  - Y^\theta_{t-1\cdot\cdot} \mathbf{1}  + Z^\theta_{t\cdot}, & \forall \theta\in\Theta, \; \forall t> 1. \label{eq:empty-flow-dynamics-private-av}
\end{align}

Since all private AV rentals complete within the same day, we also need to introduce the following constraints on the empty vehicles at the first and last time periods:
\begin{align}
    & E^\theta_{1\cdot} = N^\theta_{1\cdot}, & \forall \theta\in\Theta,\label{eq:private-av-conservation-1}\\
    & E^\theta_{T\cdot} + (V^\theta_{T\cdot\cdot})^\top \mathbf{1}  - Y^\theta_{T\cdot\cdot} \mathbf{1} = N^\theta_{T\cdot}, & \forall \theta\in\Theta.\label{eq:private-av-conservation-T}
\end{align}
It also enforces the zero vehicle dispatching at the end of service:
\begin{align}
    & Y^\theta_{tij} = 0, & \forall t  > T - \pi_{ij},\; \forall i,j\in \mathcal{J}. \label{eq:private-av-settle}
\end{align}

Similar constraints are defined for platform AVs, while a single fleet conservation is sufficient given that the platform AV supply remains the same over time:
\begin{align}
     & E^\nu_{t\cdot} = E^\nu_{t-1\cdot} + (V^\nu_{t-1\cdot\cdot})^\top \mathbf{1}  - Y^\nu_{t-1\cdot\cdot} \mathbf{1}, & \forall t,\label{eq:empty-flow-dynamics-platform-av}\\
    & E^\nu_{T\cdot} \mathbf{1}  + \mathbf{1}^\top U^\nu_T \mathbf{1}= N^\nu,\label{eq:platform-av-fleet-size}
\end{align}  
where $U^\nu_T\in\mathbb{R}_+^{J\times J}$ denotes the en-route platform AVs during the last time period $T$ with each element computed as 
\begin{align}\label{eq:en-route-platform-AV}
    &U^\nu_{Tij} = \sum_{t'=\tau^{-1}(T)}^{T-1} Y^v_{t'ij}, & \forall i,j\in \mathcal{J}.
\end{align}
Also note that, with the setting of periodic operations of platform AVs, here we let $t-1=T$ for $t=1$ for notation simplicity.

\section{Service capacity planning}\label{sec_property}

Based on the aforementioned settings, we formulate the platform's service capacity planning as an optimization problem that aims to minimize the total supply cost, which comprises vehicle rental expenses and operational costs. The passenger demand and trip revenue are assumed exogenous in this study, whereas a penalty is introduced to the objective function to capture the loss of trip fare revenue when customer demand is not fully served. Let $D\in \mathbb{R}_+^{T\times J\times J}$ denote the demand across regions in time period $t$. Then, the demand loss $W\in \mathbb{R}_+^{T\times J\times J}$ is given by 
\begin{align}
    & W_{t\cdot\cdot} =\max\left\{0, D_{t\cdot\cdot}-\sum_{k\in\mathcal{K}} Y^k_{t\cdot\cdot}\right\}, & \forall t.\label{eq:demand-loss}
\end{align}
Accordingly, the platform's optimization problem is formulated as follows: 
\begin{subequations}\label{eq:planning-v0}
\begin{align}
    \min_{P, Y} \quad & \sum_{\theta\in \Theta} \langle P, N^\theta\rangle_F + c_0 \sum_{k\in\mathcal{K}} \langle Y^k, \Pi \rangle_F + \langle c_\rho, W \rangle_F,\\
    s.t. \quad & \text{Constraints \eqref{eq:private-av-supply}-\eqref{eq:en-route-platform-AV},} \\
    & P, Y\ge 0
\end{align}
\end{subequations}
where $Y=(Y^\theta, Y^\nu)_{\forall \theta\in \Theta}$ denotes the full set of dispatching flows, $\langle \cdot,\cdot\rangle_F$ denotes the Frobenius inner product (i.e., sum of all element values of the element-wise product of two matrices),  $c_0\in\mathbb{R}_+$ represents the unit AV operation cost assumed equal between platform and private AVs, $\Pi\in \mathbb{R}_+^{T\times J\times J}$ stacks $T$ times of travel time matrix $\pi$, and $c_\rho\in \mathbb{R}_+^{T\times J\times J}$ expresses the penalty of unit demand loss by time and OD pair (e.g., trip fare).

\subsection{Reformation as quadratic program}

Under a mild assumption on the opportunity cost stated below, the original capacity planning problem \eqref{eq:planning-v0} can be reformulated as a convex quadratic program. 
\begin{assum}\label{ass:linear-opportunity-cost}
    The opportunity cost of AV owners follows a uniform distribution with support $[\underline{\varepsilon}, \overline{\varepsilon}]$. 
\end{assum}

Due Assumption~\ref{ass:linear-opportunity-cost}, we can, without loss of generality, restrict the feasible range of the rental price to $[\underline{\varepsilon}, \overline{\varepsilon}]$. Any price $P_{hj} < \underline{\varepsilon}$ results in zero supply, which is equivalent to the strategy $P_{hj} = \underline{\varepsilon}$, while any $P_{hj} > \overline{\varepsilon}$ is strictly dominated by $P_{hj} = \overline{\varepsilon}$ as it incurs higher costs for the same maximum supply. 
The assumption also reduces the private AV supply Eq.~\eqref{eq:private-av-supply} to a linear function of $P$, which is also commonly applied in the literature~\citep{baldick2000linear,Siddiq2021}:
\begin{align}\label{eq:linear-private-av-supply}
    N^\theta_{hj} = M^\theta \Omega^\theta_{hj} \frac{P_{hj} - \underline{\varepsilon}}{\overline{\varepsilon}-\underline{\varepsilon}}.
\end{align}

The following proposition formally states the reformulation of \eqref{eq:planning-v0} into a quadratic program. 
Due to its particular sparse structure, the corresponding quadratic program is well-suited for large-scale instances and can be efficiently solved using commercial solvers such as Gurobi. 
\begin{prop} \label{prop1}
    Under Assumption~\ref{ass:linear-opportunity-cost}, Problem~\eqref{eq:planning-v0} is reformulated as a convex quadratic program:
    \begin{subequations}\label{eq:planning-qp}
    \begin{align}
        \min_x \quad & x^\top Q x+ c^\top x,\\
        s.t. \quad & A^\text{eq}x = b^\text{eq},\\
        & A^\text{neq}x \leq b^\text{neq},
    \end{align}
    \end{subequations}
    where $x=\text{col}(P,Y,N,Z,V,E,W)$ denotes the full vector of decision variable; the matrices $A^\text{eq},A^\text{neq}$ and $Q$, and the vectors $b^\text{eq},b^\text{neq} $ and $c$ are provided in Appendix~\ref{app:proof_prop1}.
\end{prop}

\subsection{Service feasibility}
Using the above framework, we first derive the conditions under which the platform benefits from AV crowdsourcing. In particular, we examine whether deploying crowdsourced AVs in some time periods can reduce the platform’s overall cost. As a baseline, we first characterize the platform’s optimal strategies and the total service cost in the absence of crowdsourced AVs. Let $\tilde{x} = \text{col}(\tilde{P}, \tilde{Y}, 0, 0, \tilde{V}, \tilde{E}, \tilde{W})$ be the optimal solution to Problem \eqref{eq:planning-qp} without crowdsourced AVs (i.e., by setting the rental price to its minimum value $\underline{\varepsilon}$).
Accordingly, we have $N^\theta = 0$ for all $\theta \in \Theta$ and the service is solely served by platform AVs. 

We define the set of trip indices that are not fully served in this case as:
\begin{align}
    \mathcal{R}_0 = \left\{ (t,i,j) \mid D_{tij} -  \tilde{Y}^\nu_{tij} >0 \right\}.
\end{align}

For a given request $r = (t,i,j) \in \mathcal{R}_0$, we define $\Theta_r$ as the set of eligible AV owner classes, i.e., the vehicle is available for completing the pickup, service, and return trips. 
To simplify the analysis, we only consider renting vehicles in the same region as the trip origin or destination and denote the corresponding eligible classes as 
\begin{align}
    \Theta_r^{(i)} &= \left\{ \theta \in \Theta \,\middle|\,
        \Omega^\theta_{hi} = 1,
        \;\forall\,  \lceil t/s \rceil \le h \le  \lceil (t + \pi_{ij} +\pi_{ji}- 1) /s \rceil\right\}, \\
    \Theta_r^{(j)} &= \left\{ \theta \in \Theta \,\middle|\,
        \Omega^\theta_{hj} = 1,
        \; \forall\, \lceil (t - \pi_{ji})/s \rceil  \le h \le \lceil (t + \pi_{ij} - 1)/s \rceil
        \right\}\\
        & \text{if } t> \pi_{ji}, \text{ otherwise } \emptyset\nonumber.
\end{align}
Thus, $\Theta_r = \Theta_r^{(i)} \cup \Theta_r^{(j)}$.  Note that this setting essentially yields a more conservative service feasibility condition.

Let $\Theta_0 = \bigcup_{r \in \mathcal{R}_0} \Theta_r$ be the union of
all eligible AV owner classes for the unserved demand. Proposition
\ref{prop_feasible} gives the sufficient conditions under which the
platform reduces its operational cost through adopting crowdsourced AVs.

\begin{prop}\label{prop_feasible}
    Suppose $\mathcal{R}_0 \neq \emptyset$ and $\Theta_0 \neq \emptyset$.  If there exists  
    requests $r = (t,i,j) \in \mathcal{R}_0$ and an owner class $\theta \in \Theta_r$ such that the 
    penalty of demand loss associated with trip $r$, denoted as $c_{\rho, r}$, satisfies:
    \begin{align}\label{eq:prop_feasible}
        c_{\rho, r} > \left( 2 + \frac{\pi_{ij} + \pi_{ji}}{s} \right)
        \underline{\varepsilon} + c_0 (\pi_{ij} + \pi_{ji}),
    \end{align}
    then the optimal solution to Problem \eqref{eq:planning-qp} satisfies $\sum_{\theta\in\Theta} N^{\theta*}>\mathbf{0}$.
\end{prop}

\begin{proof}
    Since $x \ge \lfloor x \rfloor$ for any $x>0$, inequality
    \eqref{eq:prop_feasible} implies
    \begin{align}\label{eq:proof_prop_feasible_penalty_bound}
        c_{\rho,r} > \left(2 + \Bigl\lfloor \frac{\pi_{ij} + \pi_{ji}}{s} \Bigr\rfloor \right) \underline{\varepsilon} + c_0 (\pi_{ij} + \pi_{ji}).
    \end{align}
    A crowdsourced AV of class $\theta \in \Theta_r^{(i)}$ and $\Theta_r^{(j)}$ corresponds to different rental periods. In the former case, the number of rental periods is computed as 
    \begin{align}\label{eq:proof_prop_feasible_n_i}
        n_r^{(i)} = \Big\lceil \frac{t+\pi_{ij}+\pi_{ji}-1}{s}\Big\rceil - \Big\lceil\frac{t}{s}\Big\rceil + 1;
    \end{align}
    and in the latter case, it is given by 
    \begin{align}\label{eq:proof_prop_feasible_n_j}
        n_r^{(j)} = \Big\lceil \frac{t+\pi_{ij}-1}{s}\Big\rceil - \Big\lceil\frac{t-\pi_{ji}}{s}\Big\rceil + 1.
    \end{align}
    We first unify Eqs.~\eqref{eq:proof_prop_feasible_n_i} and \eqref{eq:proof_prop_feasible_n_j} by defining the rental starting time 
    \begin{align}
        t_\theta=\begin{cases}
            t, & \theta\in \Theta_r^{(i)},\\
            t-\pi_{ji}, & \theta \in \Theta_r^{(j)}, 
        \end{cases}
    \end{align}
    thus the general number of rental periods is 
    \begin{align}\label{eq:proof_prop_feasible_n}
        n_r = \Big\lceil \frac{t_\theta+\pi_{ij} + \pi_{ji}-1}{s}\Big\rceil - \Big\lceil\frac{t_\theta}{s}\Big\rceil + 1. 
    \end{align}
    Suppose $t_\theta - 1 = \beta_0 + \beta_1 s$ with $\beta_0,\beta_1\in\mathbb{Z}_+, \beta_0\leq s-1$. Then, we have
    \begin{align}
        \Big\lceil \frac{t_\theta+\pi_{ij} + \pi_{ji}-1}{s}\Big\rceil = \beta_1 + \Big\lceil \frac{\beta_0+\pi_{ij} + \pi_{ji}}{s}\Big\rceil,\quad 
        \Big\lceil\frac{t_\theta}{s}\Big\rceil = \beta_1 + 1.
    \end{align}
    Plugging the above into Eq.~\eqref{eq:proof_prop_feasible_n} yields
    \begin{align}\label{eq:proof_prop_feasible_n_bound}
        n_r = \Big\lceil \frac{\beta_0 + \pi_{ij} + \pi_{ji}}{s}\Big\rceil \leq \Big\lceil \frac{s-1 + \pi_{ij} + \pi_{ji}}{s}\Big\rceil &= 1+ \Big\lceil \frac{\pi_{ij} + \pi_{ji} -1}{s}\Big\rceil \\
        & \leq 2 + \Big\lfloor \frac{\pi_{ij} + \pi_{ji} -1}{s}\Big\rfloor \leq 2 + \Big\lfloor \frac{\pi_{ij} + \pi_{ji}}{s}\Big\rfloor. \nonumber
    \end{align}
    Combining Eqs.~\eqref{eq:proof_prop_feasible_penalty_bound} and \eqref{eq:proof_prop_feasible_n_bound}, we have 
    \begin{align}
        c_{\rho,r} > n_r\underline{\varepsilon} + c_0(\pi_{ij} + \pi_{ji}). 
    \end{align}
    Let $\delta_1>0$ denote the demand associated with request $r\in\mathcal{R}_0$, and let $\delta_2$ be defined by
    \begin{align*}
        & c_{\rho,r} = n_r\left[\underline{\varepsilon} + \frac{\delta_2}{M^\theta}(\overline{\varepsilon}-\underline{\varepsilon})\right] + c_0(\pi_{ij} + \pi_{ji}).
    \end{align*}
    Choose $0<\delta <\min\left\{\delta_1,\delta_2\right\} $, then 
    \begin{align}\label{eq:proof_prop_feasible_final}
        c_{\rho,r} > n_r\left[\underline{\varepsilon} + \frac{\delta}{M^\theta}(\overline{\varepsilon}-\underline{\varepsilon})\right] + c_0(\pi_{ij} + \pi_{ji}),
    \end{align}
    which essentially implies that the platform can rent AVs of class $\theta$ for $n_r$ rental periods at a rental rate $\underline{\varepsilon} + \frac{\delta}{M^\theta}(\overline{\varepsilon}-\underline{\varepsilon})$ to cover 
    a small amount $\delta$
    unserved trip $r$ at a lower total cost than $c_{\rho,r}$. Therefore, the platform's optimal capacity planning must contain crowdsourced AVs (i.e., $\sum_{\theta\in\Theta} N^{\theta*}>\mathbf{0}$), and this completes the proof. 
\end{proof}
  
To explore the generality of Proposition~\ref{prop_feasible}, we define a rental payment threshold as 
\begin{align}
    \epsilon = \frac{[c_{\rho,r} - c_0(\pi_{ij} + \pi_{ji})]s}{2s + \pi_{ij} + \pi_{ji}}
\end{align}
thus Condition \eqref{eq:prop_feasible} holds if and only if $\epsilon > \underline{\varepsilon}$. 

Figure~\ref{fig_wp} shows the distribution of $\epsilon$ computed on Chicago MoD trip data during the period from 18 to 24 November 2024~\citep{chicago_tnp_trips_2023}. Specifically, we set $c_{\rho,r}$ as the total trip fare, time-scale raio $s=4$, and test on three levels of operational cost (including energy cost and maintenance cost, etc) $c_0 \in \{0.1, 0.6, 1.1\}$ U.S. dollars per time slot $\Delta_t = 15$ minutes according to \cite{ayetor2023comparing}. 
We assume that the AV purchase cost can be covered by the owner with the lowest opportunity cost, which is approximated as $\underline{\varepsilon}=1$ (\$/hr), indicated by the dashed line in Figure~\ref{fig_wp}. It can be seen that the majority of rental threshold payments are significantly higher than the lowest opportunity cost, suggesting that condition (\ref{eq:prop_feasible}) is likely to hold in real practice.
\begin{figure}
    \centering
    \includegraphics[width=0.95\linewidth]{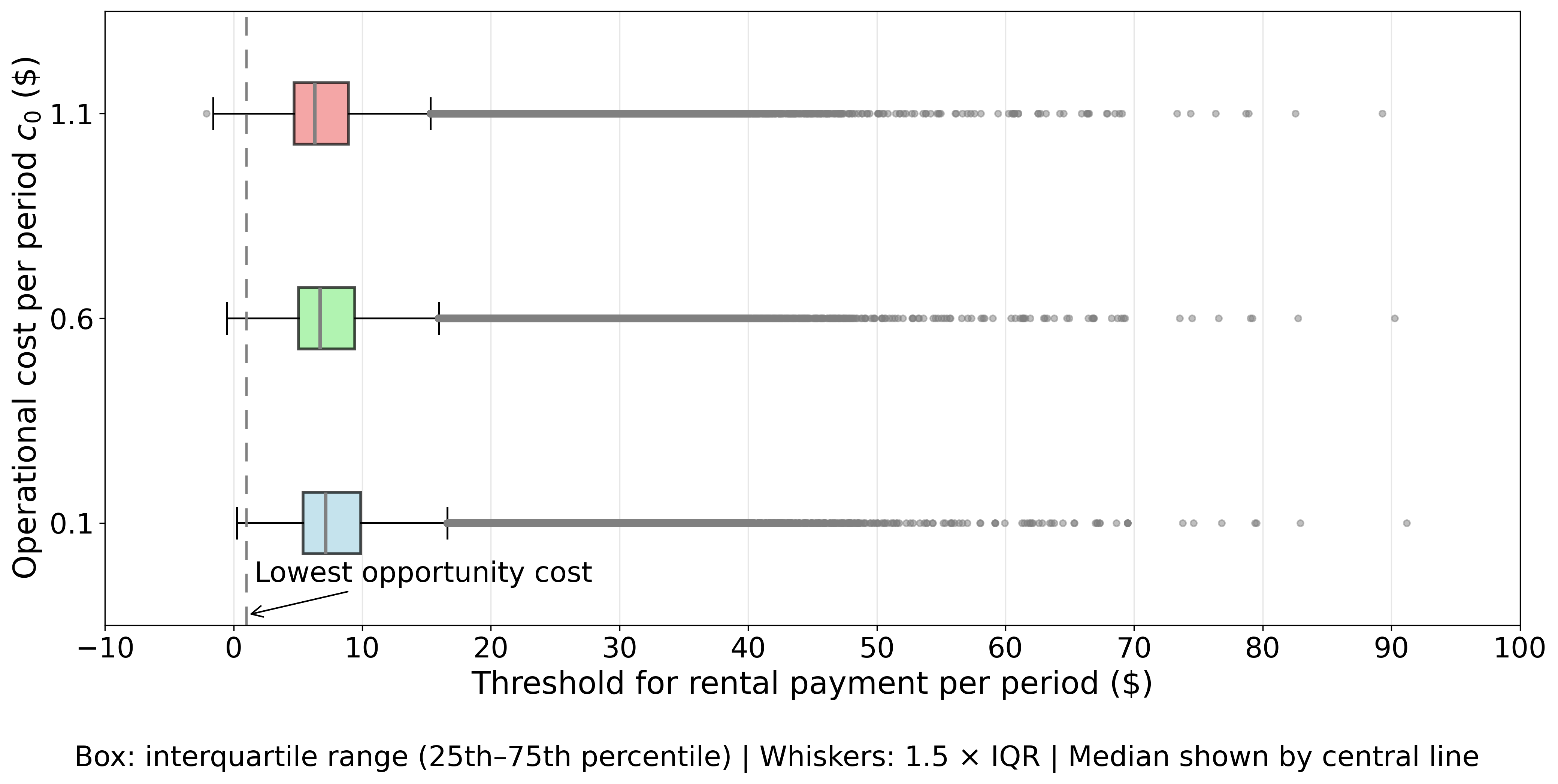}
    \caption{Illustration of rental payment threshold.} 
    \label{fig_wp}
\end{figure}

\subsection{Benefit of spatial heterogeneity in supply}

This section proceeds to investigate the benefit of spatially diverse crowdsourced AV supply, a factor that cannot be evaluated in aggregate models and thus remains underexplored in the literature. 
Since the demand for MoD services is often spatially imbalanced, e.g., from residential areas to the central business district (CBD), the platform must make strategic decisions on where to crowdsource private AVs. For a single directional travel demand, sourcing AVs from the origin and destination areas should be symmetric, as both cases induce comparable deadheading vehicle time---each vehicle either serves a MoD trip then returns vacantly or first repositions to the MoD trip origin. 
However, our analysis shows that this seemingly intuitive result does not hold in general, mainly due to the different discretization of rental periods and operational time steps.

\begin{figure}[h]
    \centering
    \includegraphics[width=0.9\linewidth]{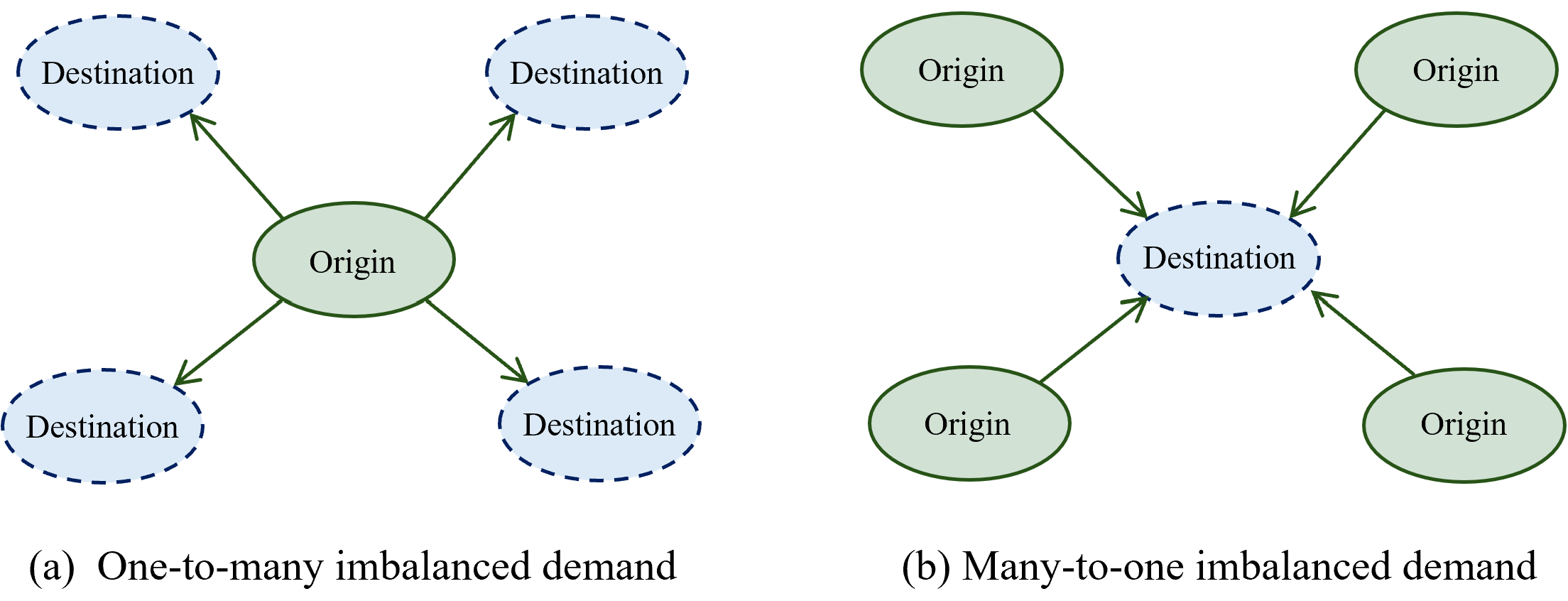}
        \caption{Directional travel demand patterns.}
    \label{fig:imbalanced_demand}
\end{figure}

Consider the directional travel demand (e.g., commuting trips) shown in Figure~\ref{fig:imbalanced_demand}. 
The set of regions $\mathcal{J}$ is partitioned into two subsets, namely the origin set $\mathcal{J}_o$ and the destination set $\mathcal{J}_d$, according to whether each region serves as the MoD trip origin or destination. 
We then define two types of demand patterns to be examined: 
\begin{itemize}
    \item One-to-many demand pattern: All trips originate from a single origin region and distribute across one or more destination regions, i.e., $|\mathcal{J}_o|=1$.
    \item Many-to-one demand pattern: All trips from one or more regions end in a single destination region, i.e., $|\mathcal{J}_d|=1$. 
\end{itemize}

Without loss of generality, we assume the platform only operates crowdsourced AVs and the inter-region travel times are symmetric, i.e., $\pi_{ij}=\pi_{ji}$ for any $i,j\in\mathcal{J}$. Besides, the travel times satisfy the triangle inequality, i.e.,
\begin{align}\label{eq:tri_ineq}
    & \pi_{ij}+\pi_{jl} \ge \pi_{il}, & \forall i,j,l\in\mathcal{R}.
\end{align}
For notation simplicity, we define the longest and shortest travel times between regions as $\overline{\pi}=\max_{i,j\in\mathcal{J}} \pi_{ij}$, and $\underline{\pi}=\min_{i,j\in\mathcal{J}} \pi_{ij}$, respectively, and denote the set of rental starting periods as $\mathcal{B} =\{s, 2s, \dots, Hs\}$ and thus the time interval between two consecutive rental starting times is exactly $\Delta_s = s\Delta_t$.

The following proposition proves that, under certain conditions, the \emph{single-sided crowdsourcing strategy}, i.e., only sourcing AVs from either origin or destination zones, is dominated by the \emph{double-sided strategy}, i.e., sourcing AVs from both types of zones. The underlying rationale is illustrated in Figure~\ref{fig:lemma_cross} and rooted in the different time discretization of rental and service periods. To serve the MoD trip 2, the platform under the single-sided strategy needs to source an AV for two rental periods $h-1$ and $h$, whereas only one rental period $h$ is needed under the double-sided strategy. Yet, such a situation occurs when certain conditions hold, which are summarized below. 

\begin{Lem}\label{lemma_cross}
    Suppose there exists a positive potential private AV supply in every region, and the maximum and minimum inter-region travel times satisfy $2\overline{\pi}\Delta_t < \Delta_s$ and $\underline{\pi}>0$. Further, the optimal single-sided rental strategy leads to either of the following outcomes:
    \begin{enumerate}[label=\textbf{Case \Alph*.}, leftmargin=*, align=left ] 
        \item The platform only crowdsources vehicles from the origin regions, and there exist some crowdsourced AVs serving trips from region $i$ to $j$ that depart in the interval $(b-\pi_{ij}-\delta, b-\pi_{ij}]$ for some $b\in\mathcal{B}$ and $\delta >0$. 

        \item The platform only crowdsources vehicles from the destination regions, and there exist some crowdsourced AVs serving trips from region $i$ to $j$ that depart in the interval $[b, b+\delta)$ for some $b\in\mathcal{B}$ and $\delta >0$. 
    \end{enumerate}
    Then, the optimal double-sided crowdsourcing strategy dominates with a strictly reduced total rental time and non-increasing vehicle operation cost. 
\end{Lem}

\begin{figure}[h]
    \centering
    \includegraphics[width=0.9\linewidth]{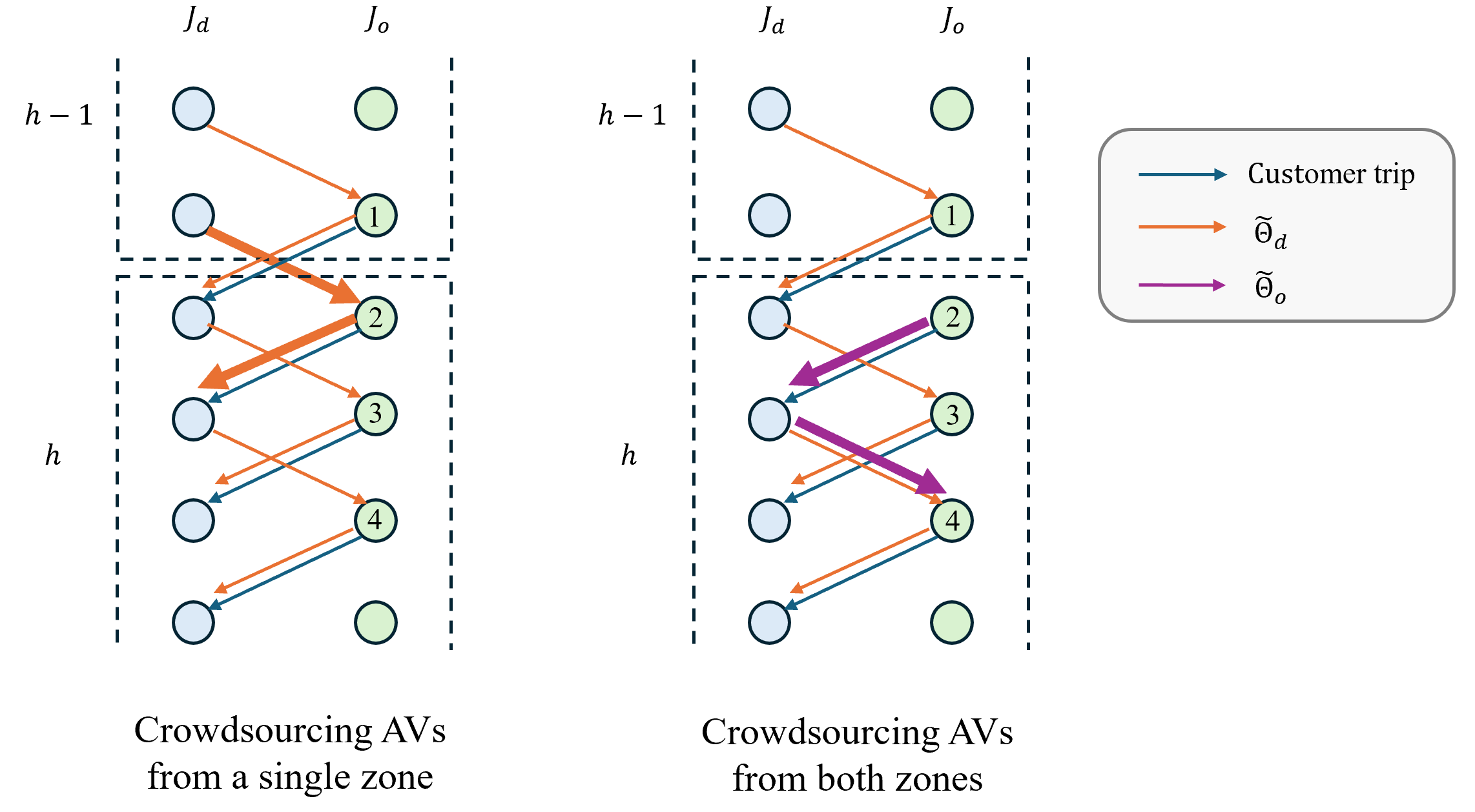}
    \caption{Illustration of Lemma\ref{lemma_cross}: serving MoD trip 2 saves one rental period}
    \label{fig:lemma_cross}
\end{figure}

\begin{proof}
    Without loss of generality, we focus on the one-to-many demand scenario, since the many-to-one scenario can be analyzed in a similar manner and yields the same results. In the following, we prove the two cases separately.

    \begin{enumerate}[label=\textbf{Case \Alph*.}, leftmargin=*, align=left ] 
        \item Suppose the platform only sources AVs from the origin region $i\in \mathcal{J}_o$. Let $0<\delta\le \min(s-2\pi_{ij}, \pi_{ij})$, where $s-2\pi_{ij}> 0$ due to the condition $2\bar{\pi}\Delta_t<\Delta_s$. Then, to serve a trip starting in the interval $t \in (b - \pi_{ij}-\delta, b - \pi_{ij}]$ for some $b\in \mathcal{B}$, the AV is at least sourced during the interval $[t, t+2\pi_{ij}]$, which covers two rental periods as $t<b$ and $t+2\pi_{ij}>b$. On the other hand, if the platform sources an AV from region $j\in\mathcal{J}_d$, the minimum sourcing interval becomes $[t-\pi_{ij}, t+\pi_{ij}]$ with both ends falling in the same rental period as $t-\pi_{ij}>b-1$ and $t+\pi_{ij}<b$. Meanwhile, both crowdsourcing strategies lead to the same vehicle travel distance $2\pi_{ij}$ and thus the same operation cost $2\pi_{ij} c_0$. It then concludes that the double-sided strategy dominates with fewer rental periods and non-increasing operation cost. 

        \item Suppose the platform only sources AVs from destination regions. The result can be similarly proved by setting $0<\delta\leq \min(\underline{\pi}, s-2\bar{\pi})$. In this case, the trips of interest with departure time $t\in[b, b+\delta)$ for some $b\in\mathcal{B}$ is served by some AV sourced from a destination region $j'\in \mathcal{J}_d$. If $j'\neq j$, the minimum sourcing interval is $[t-\pi_{ij'}, t+\pi_{ij}+\pi_{jj'}]$; otherwise $[t-\pi_{ij}, t+\pi_{ij}]$. In both cases, it crosses the boundary $b$ and yields two rental periods. In contrast, if the platform sources AVs from region $i$, the minimum sourcing interval reduces to $[t, t+2\pi_{ij}]$, which is covered by a single rental period. The dominance of operation cost can be easily proved by the triangle inequality Eq.~\eqref{eq:tri_ineq}. 
    \end{enumerate}

\end{proof}

\begin{prop}
     \label{prop_both_zone}
     Suppose Assumption~\ref{ass:linear-opportunity-cost} and the same conditions as in Lemma~\ref{lemma_cross} hold. Then, the optimal double-sided crowdsourcing strategy strictly dominates the optimal single-sided strategy with a lower total cost.
\end{prop}

\begin{proof}
    Suppose the platform initially sources AVs exclusively from region $i$. By Lemma \ref{lemma_cross}, there exist time windows where replacing a vehicle sourced from region $i$ with one from region $j$ serves the same trip while reducing the total vehicle rental hours, without increasing the vehicle operating cost. We focus on the change in total payment cost, denoted by $\delta$. 
    The original rental cost is $\sum_h \mathring{N}_{i,h} \mathring{P}_{i,h} $. 
    Following the strategy in the analysis of  Lemma \ref{lemma_cross}, denote $\hat{\mathcal{H}}$ as the set of rental periods that the platform decreases $\epsilon$ vehicles. Since Lemma \ref{lemma_cross} shows that the total vehicle hours decreases, the new rental cost is less than
    \begin{align*}
        & \sum_{h\in\hat{\mathcal{H}}} \left[(\mathring{N}_{i,h}- \epsilon)(\mathring{P}_{i,h}- \Delta P_{i,h}) + \epsilon(\underline{\varepsilon} + \Delta P_{j,h}) \right]+\sum_{h\notin\hat{\mathcal{H}}} \mathring{N}_{i,h}\mathring{P}_{i,h}
    \end{align*}
    Then, let $\mathring{P}_{i,h}$ denote the original equilibrium price and $\Delta P_{i,h}$ denote  the reduction in price, and the new payment in region $j$ (starting from a baseline $\underline{\varepsilon}$) represented as $ (\underline{\varepsilon} + \Delta P_{j,h})$. We have 
    \begin{align*}
        & \delta \le \sum_{h\in\hat{\mathcal{H}}} \left[(\mathring{N}_{i,h}- \epsilon)(\mathring{P}_{i,h}- \Delta P_{i,h}) + \epsilon(\underline{\varepsilon} + \Delta P_{j,h}) - \mathring{N}_{i,h}\mathring{P}_{i,h}\right] \\
        & = \sum_{h\in\hat{\mathcal{H}}} \left[-\mathring{N}_{i,h}\Delta P_{i,h} - \epsilon (\mathring{P}_{i,h}- \Delta P_{i,h} - \underline{\varepsilon} - \Delta P_{j,h})\right]\\
        &  \overset{(a)}{=} \sum_{h\in\hat{\mathcal{H}}}\left[-\mathring{N}_{i,h}\left( \frac{\overline{M}_j}{\overline{M}_i} \Delta P_{j,h} \right) - \overline{M}_j \Delta P_{j,h} (\mathring{P}_{i,h}- \frac{\overline{M}_j}{\overline{M}_i} \Delta P_{j,h} - \underline{\varepsilon} - \Delta P_{j,h}) \right]\\
        & = \sum_{h\in\hat{\mathcal{H}}}\left\{-\overline{M}_j \Delta P_{j,h} \left[ \frac{\mathring{N}_{i,h}}{\overline{M}_i} + (\mathring{P}_{i,h}- \underline{\varepsilon}) - \Delta P_{j,h} \left( 1 + \frac{\overline{M}_j}{\overline{M}_i} \right) \right] \right\}
    \end{align*}
    where the equality (a) is because both the decrease in vehicle supply in region $i$ and the increase in vehicle supply in region $j$ are equal to $\epsilon$, and Assumption \ref{ass:linear-opportunity-cost} implies that $\overline{M}_i \Delta P_{i,h} = \overline{M}_j \Delta P_{j,h}$.

    As $\epsilon \to 0^+$, it implies $\Delta P_{j,h} \to 0^+$. In the limit, the term inside the square brackets simplifies to $\frac{N_{i,h}^{*}}{\overline{M}_i} + (\mathring{P}_{i,h}- \underline{\varepsilon})$. Since $\mathring{N}_{i,h}> 0$, $\overline{M}_i > 0$, and the equilibrium price $\mathring{P}_{i,h}$ is greater than the baseline opportunity cost $\underline{\varepsilon}$, the term in the brackets is strictly positive. Therefore, $\delta < 0$, proving that the total payment cost decreases when sourcing from both zones. The case where the platform initially sources AVs exclusively from region $j$ follows by symmetry.
\end{proof}

\section{Sensitivity analyses of influential factors} \label{sec_num}

This section proceeds to explore key factors that affect the performance of the AV crowdsourcing service through numerical experiments on a stylized network. As shown in Figure~\ref{fig:region}, the network consists of a residential area and a commercial area.
Accordingly, we construct trips of both AV owners and MoD customers following the daily commuting pattern, shown in Figure~\ref{fig:shift_concept}. Specifically, all MoD trips are from the residential area to the commercial area during the morning peak from 7 am to 10 am, while the opposite direction applies for the evening peak from 17 pm to 20 pm.
The travel pattern of AV owners closely mimics that of MoD customers but shifts either forward or backward (see two examples in Figure~\ref{fig:shift_concept}). 

All experiments presented in this section adopt the rental interval $\Delta_s = 1$ hr and service interval $\Delta_t = 15$ min, consistent with the Chicago case study in Section~\ref{sec_chicago}. 
The service horizon spans the entire 24 hours of a day, and the default travel times are set as $\Delta_t$ for intra-region trips and $3\Delta_t$ for inter-region trips.  
Following Assumption~\ref{ass:linear-opportunity-cost}, we assign a uniform distribution of AV owners' opportunity cost with lower bound $\underline{\varepsilon} = 1$ and upper bound $\overline{\varepsilon} = 25$  (\$/hr). The operation cost is set as $c_0 = 0.1$ (\$/hr) following \cite{WU201552}. Finally, the fleet size of platform-owned AVs is assumed to be $N^\nu = 100$.

\begin{figure}[h]
    \centering 
    \includegraphics[width=0.5\linewidth]{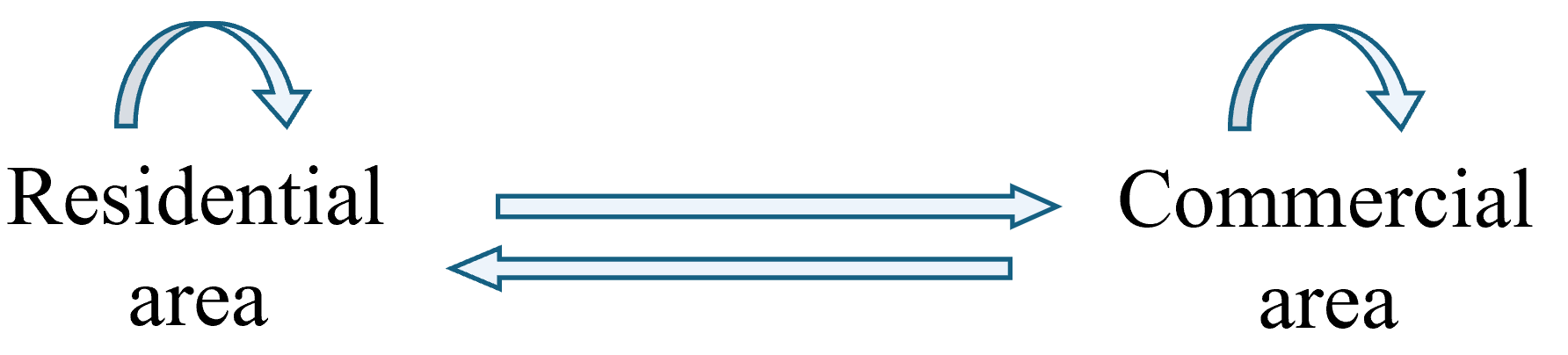}
    \caption{Stylized two-region network. }
    \label{fig:region}
\end{figure}
\begin{figure}[h]
    \centering
    \includegraphics[width=0.9\linewidth]{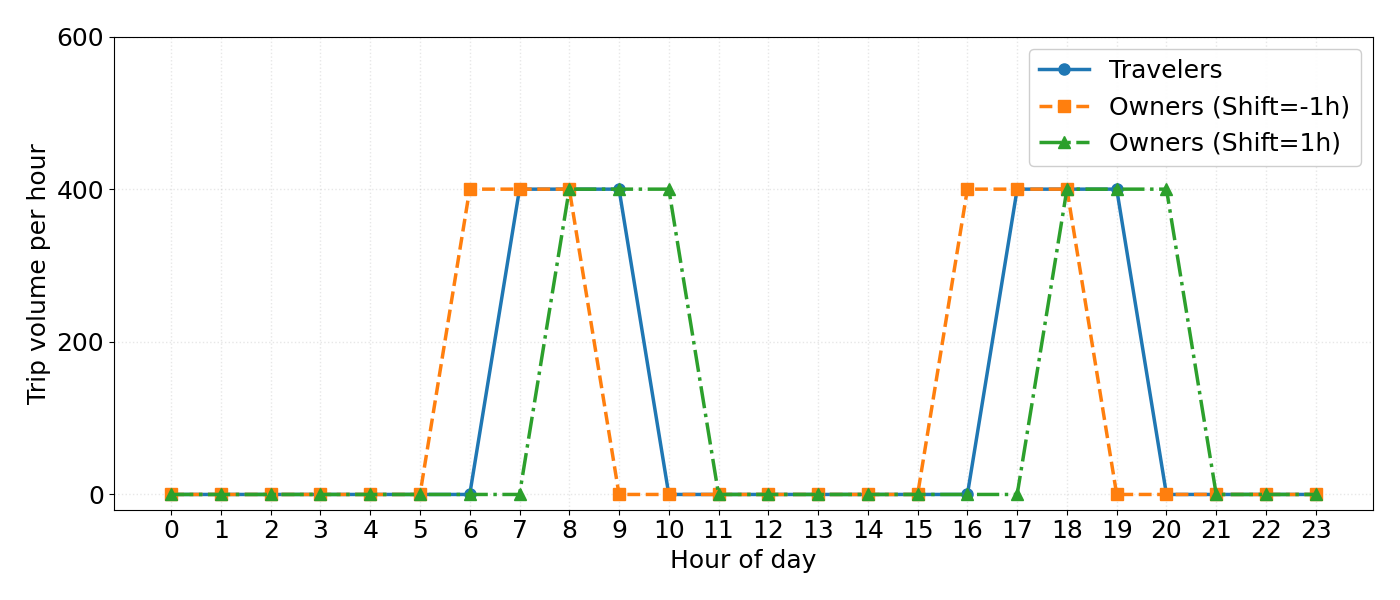}
    \caption{An example of travel time shift between owners and customers.}
    \label{fig:shift_concept}
\end{figure}

In what follows, we examine the joint effect of the activity pattern shift between AV owners and MoD customers, the primary factor, and each of the three other factors on both optimal platform strategies and service performance metrics. The goal is not to deliver practical insights, which is, however, the main purpose of Section~\ref{sec_chicago}, but instead to quantify the sensitivity of the AV crowdsourcing service towards these factors.

\subsection{Demand loss penalty}
We first investigate the service outcomes against the shift in activity patterns and the penalty of loss customer. 
As shown in Figure~\ref{fig:performance_metrics_shift}, the service achieves a higher efficiency with higher demand fulfillment and lower total service cost when the absolute gap between activity patterns expands. The direction of shift, however, does not make a significant difference, as all the results are mostly symmetric. These findings are generally expected because a significant supply-demand imbalance is likely to arise when passengers and AV owners simultaneously require vehicle access.

As expected, the higher the penalty, the more aggressive the private AV acquisition, which is reflected in both rental hours and price as per Figures~\ref{fig:performance_metrics_shift} (b) and (c). After the absolute travel time shift exceeds 1 hour, rental hours and served demand increase as the magnitude of the shift grows. This increasing trend, followed by a gradual slowdown, is more noticeable under the high-penalty scenario.
Noticeably, under the high-penalty scenario, rental prices exhibit substantial oscillations, driven by the trade-off between demand-loss penalties and supply acquisition costs. Overall, results in Figure~\ref{fig:performance_metrics_shift} suggest that given the considerable lost customer penalty, the success of AV crowdsourcing largely depends on the difference in activity patterns between private AV owners and passengers. 
\begin{figure}[!b]
    \centering
    \begin{subfigure}{0.45\linewidth}
        \centering
        \includegraphics[width=\linewidth]{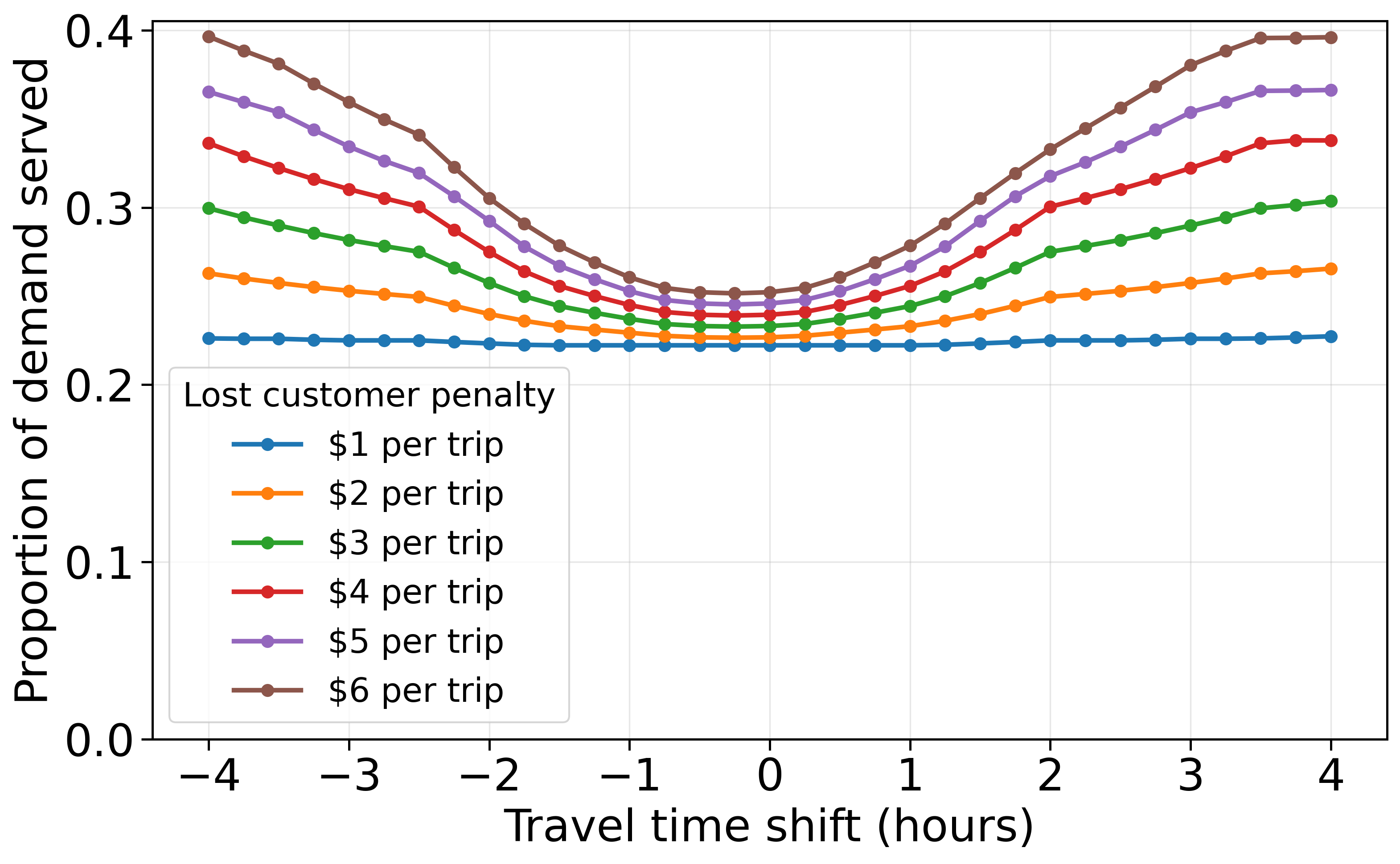}
        \caption{Proportion of demand served}
        \label{fig:demand_served}
    \end{subfigure}
    \hfill 
    \begin{subfigure}{0.45\linewidth}
        \centering
        \includegraphics[width=\linewidth]{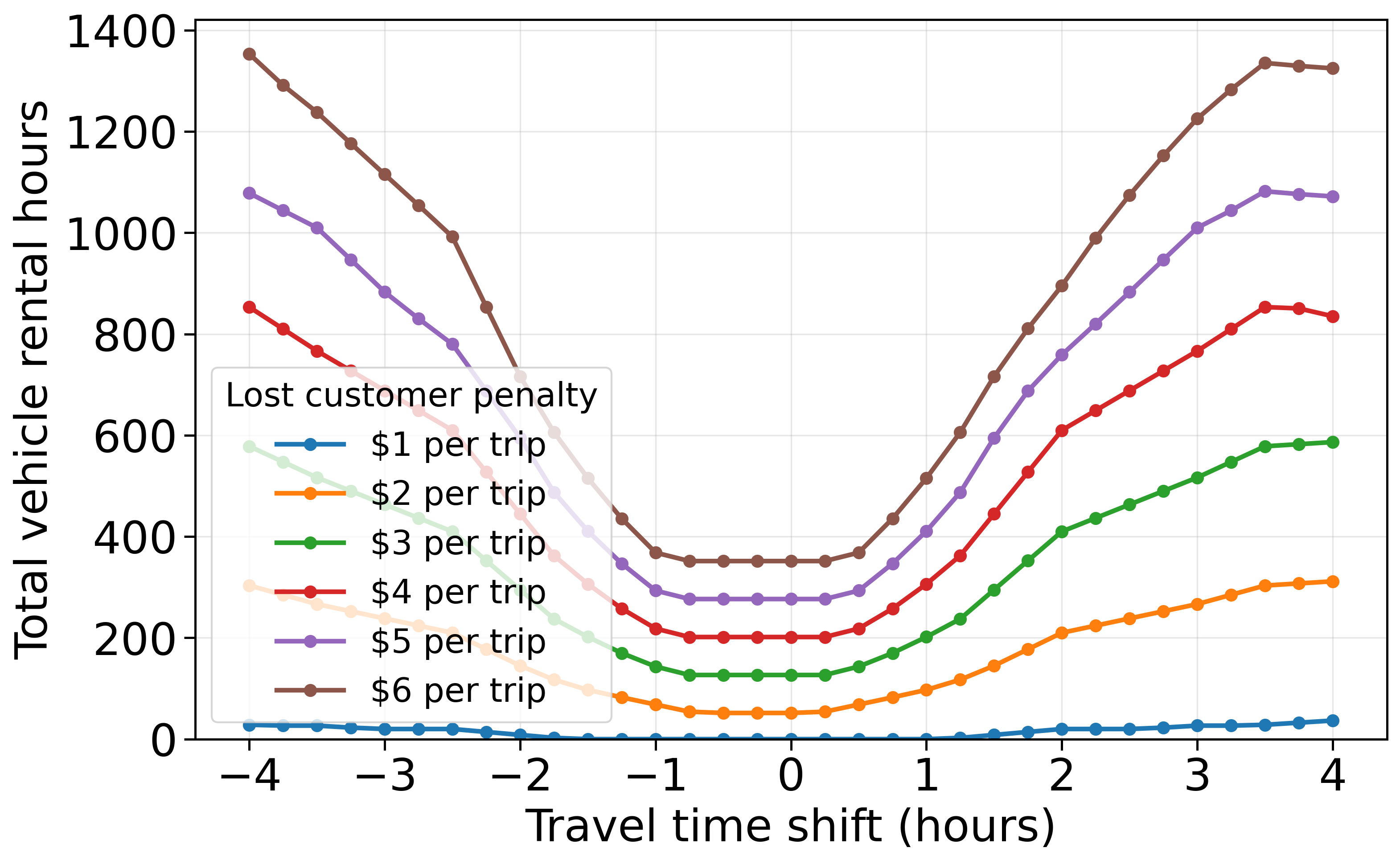}
        \caption{Total vehicle rental hours}
        \label{fig:rental_hours}
    \end{subfigure}
    \begin{subfigure}{0.45\linewidth}
        \centering
        \includegraphics[width=\linewidth]{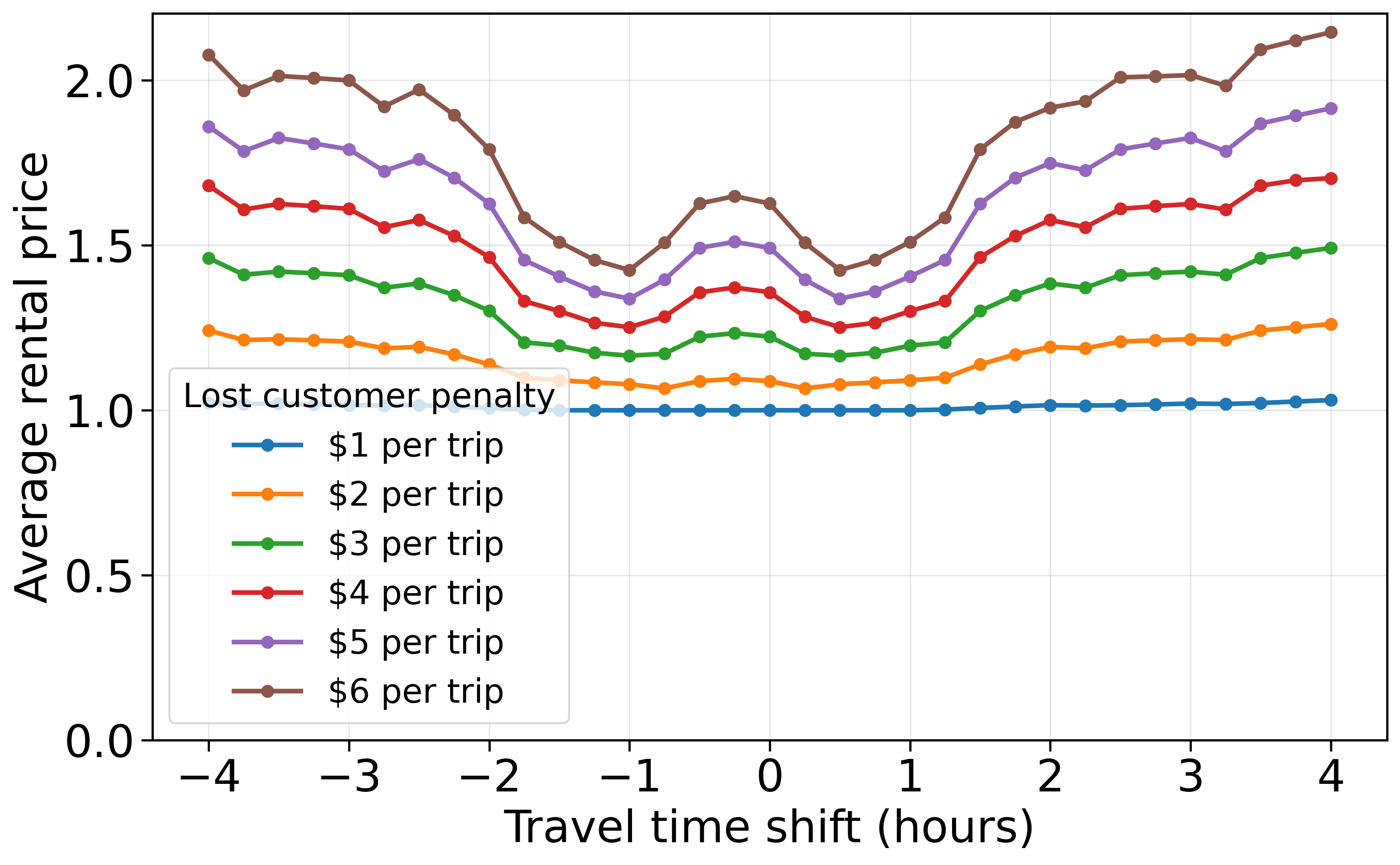}
        \caption{Average rental price}
        \label{fig:avg_price}
    \end{subfigure}
    \hfill
    \begin{subfigure}{0.45\linewidth}
        \centering
        \includegraphics[width=\linewidth]{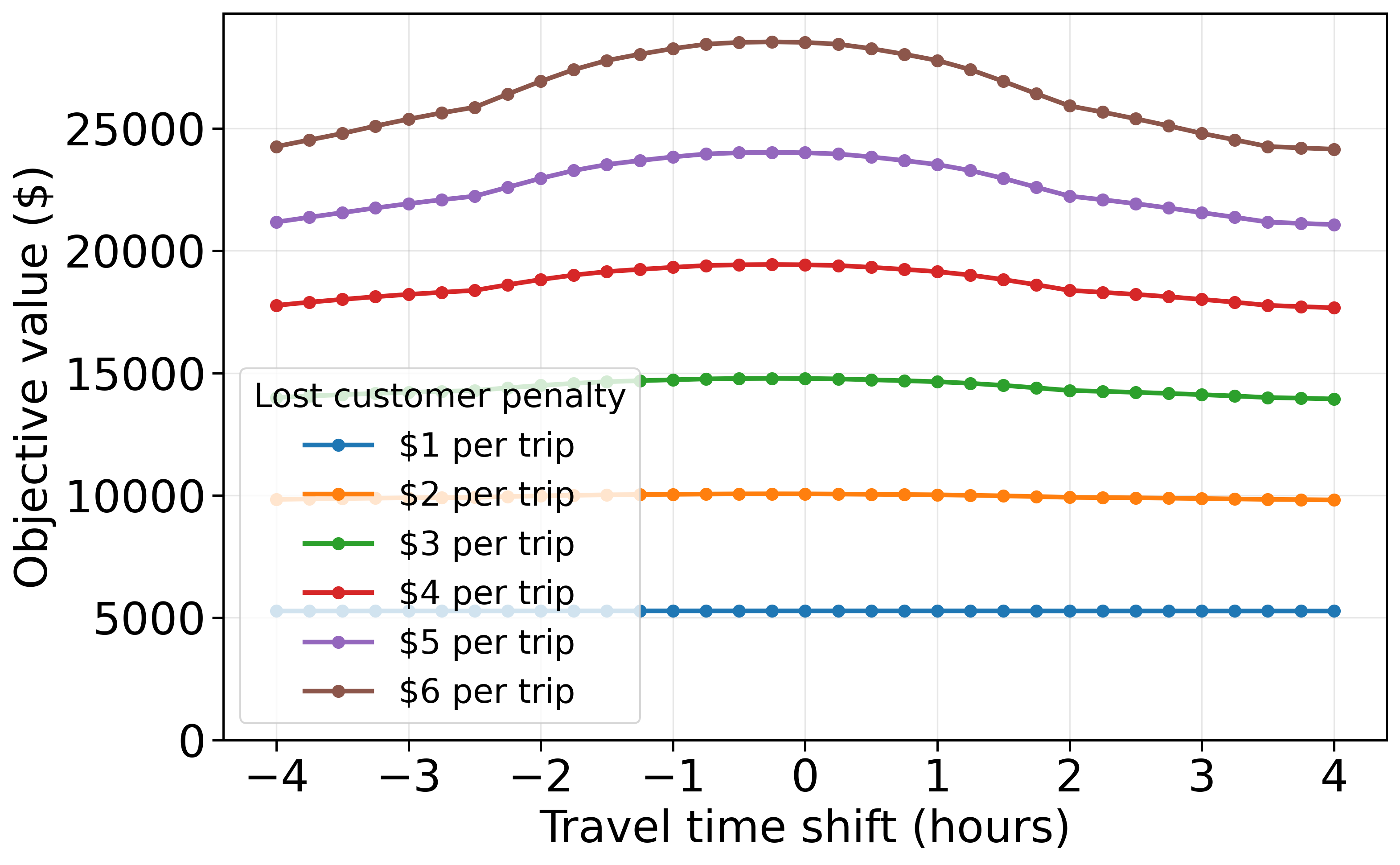}
        \caption{Objective value}
        \label{fig:obj_value}
    \end{subfigure}
    \caption{Results under different activity pattern shifts and demand loss penalties.} 
    \label{fig:performance_metrics_shift} 
\end{figure}

%
\subsection{Reliability buffer}
To guarantee vehicle access for AV owners when needed, some buffer time shall be introduced before and after the AV owner's trip, as illustrated in Figure~\ref{fig:buffer_concept}. Consequently, the available period for crowdsourcing becomes shorter. The length of buffer time also reflects the AV owner's risk preference: a larger buffer indicates the owner is highly risk-averse and favors access to his/her own vehicle over the potential revenue of vehicle sharing. 

\begin{figure}[h]
    \centering
    \includegraphics[width=0.65\linewidth]{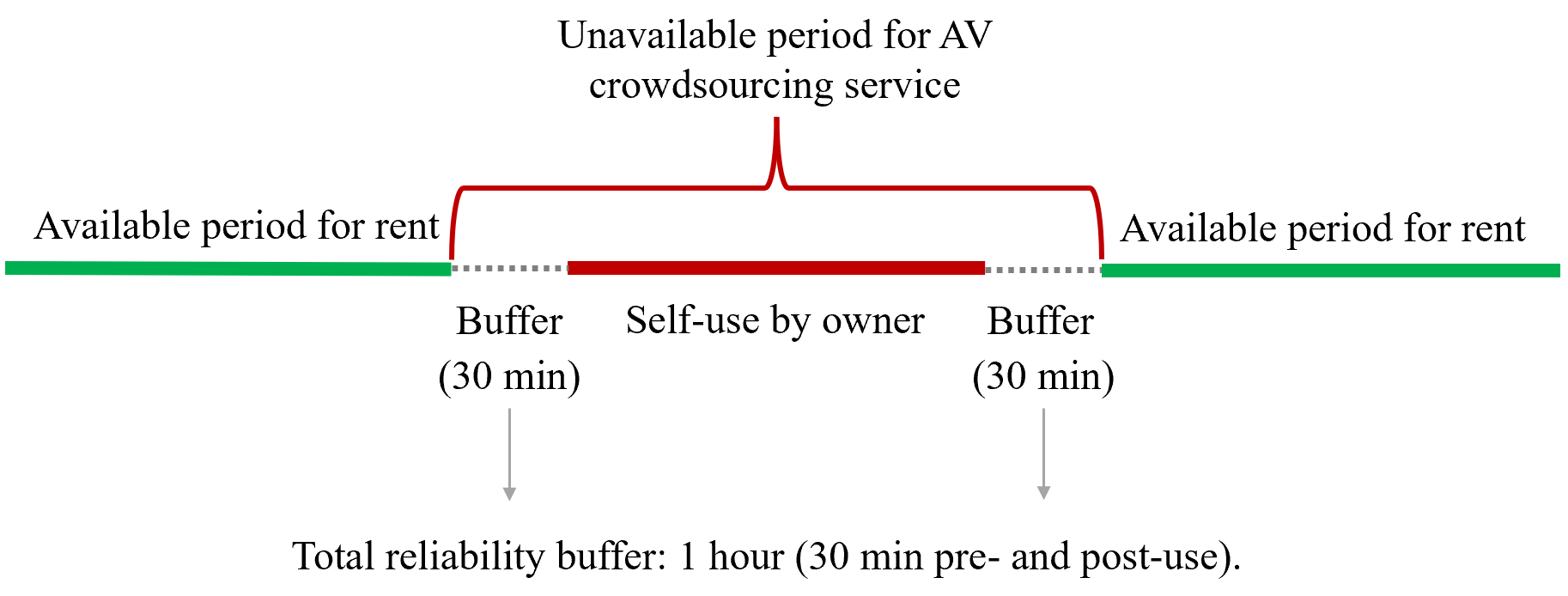}
    \caption{Illustration of one-hour reliability buffer}
    \label{fig:buffer_concept}
\end{figure}

\begin{figure}[!h]
    \centering
    \begin{subfigure}{0.45\linewidth}
        \centering
        \includegraphics[width=\linewidth]{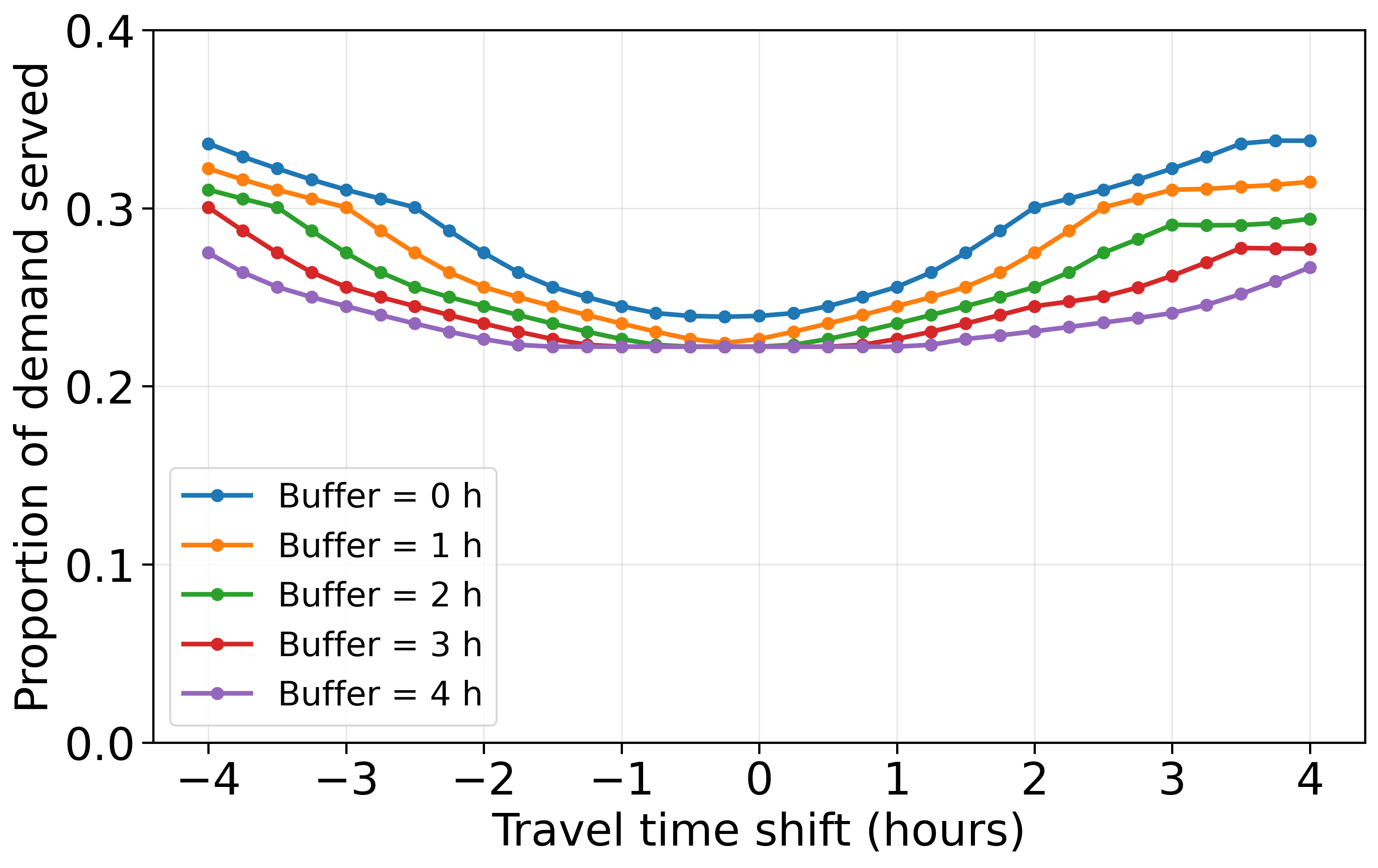}
        \caption{Proportion of demand served}
        \label{fig:demand_served}
    \end{subfigure}
    \hfill 
    \begin{subfigure}{0.45\linewidth}
        \centering
        \includegraphics[width=\linewidth]{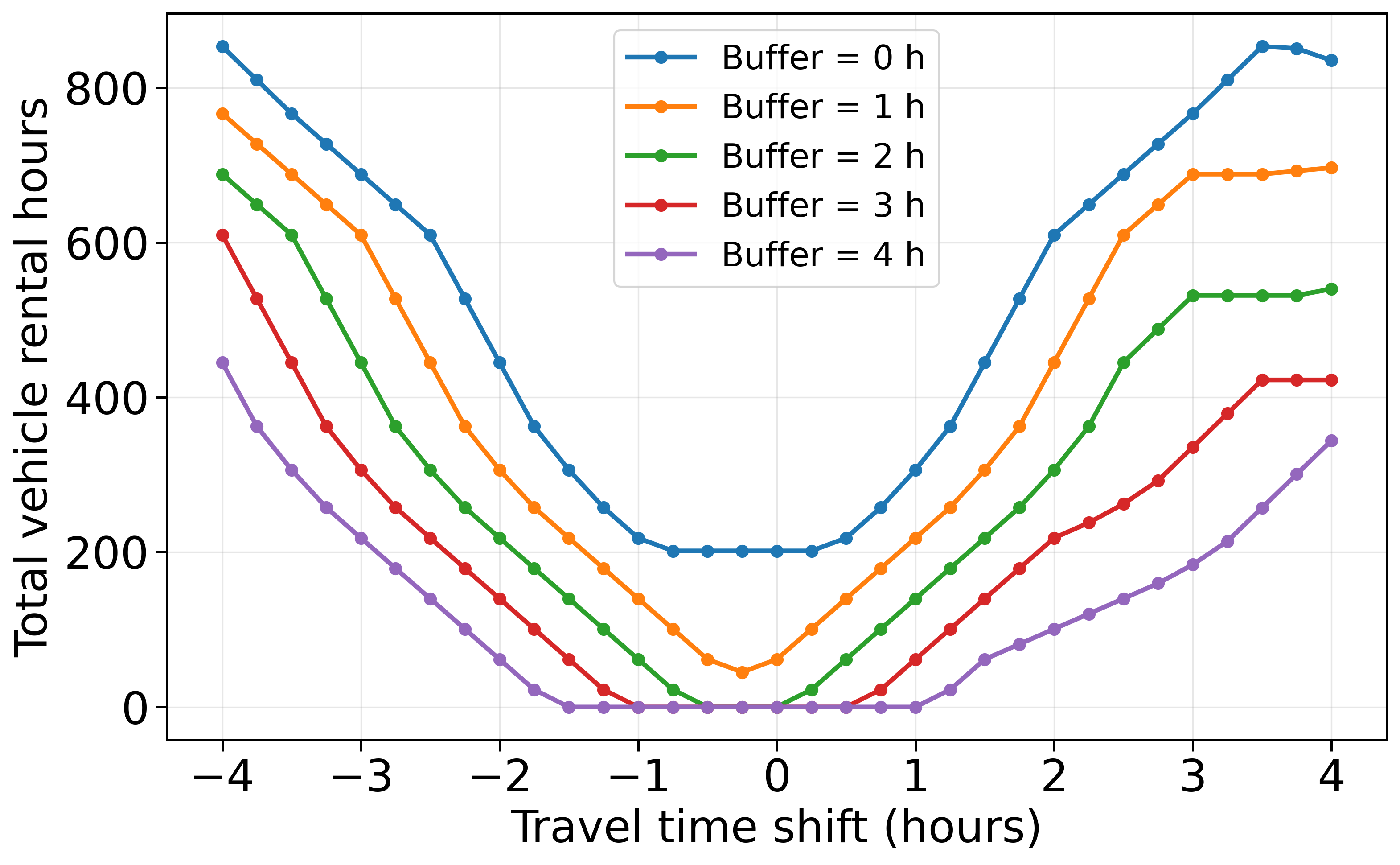}
        \caption{Total vehicle rental hours}
        \label{fig:rental_hours}
    \end{subfigure}

    \begin{subfigure}{0.45\linewidth}
        \centering
        \includegraphics[width=\linewidth]{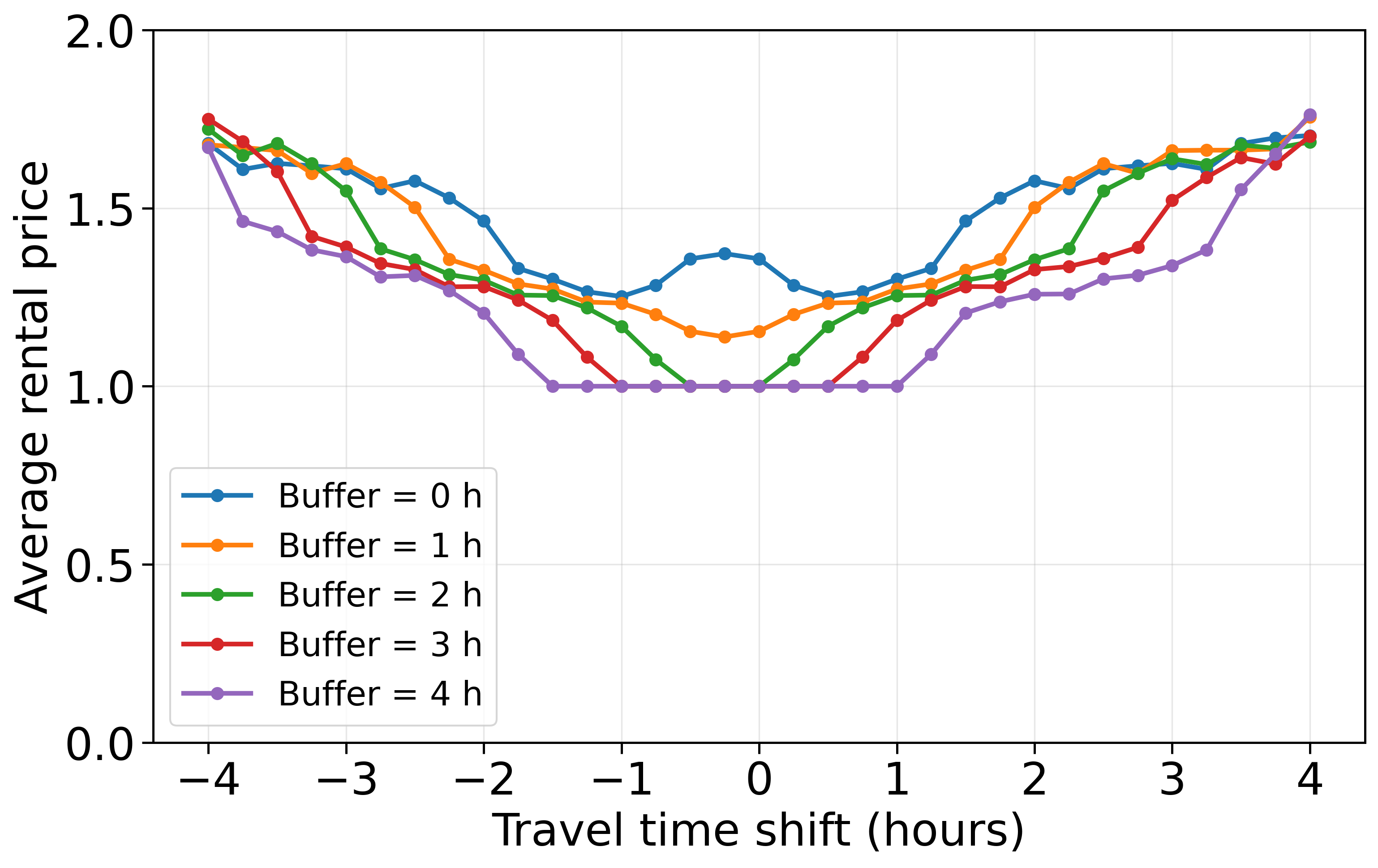}
        \caption{Average rental price}
        \label{fig:avg_price}
    \end{subfigure}
    \hfill
    \begin{subfigure}{0.45\linewidth}
        \centering
        \includegraphics[width=\linewidth]{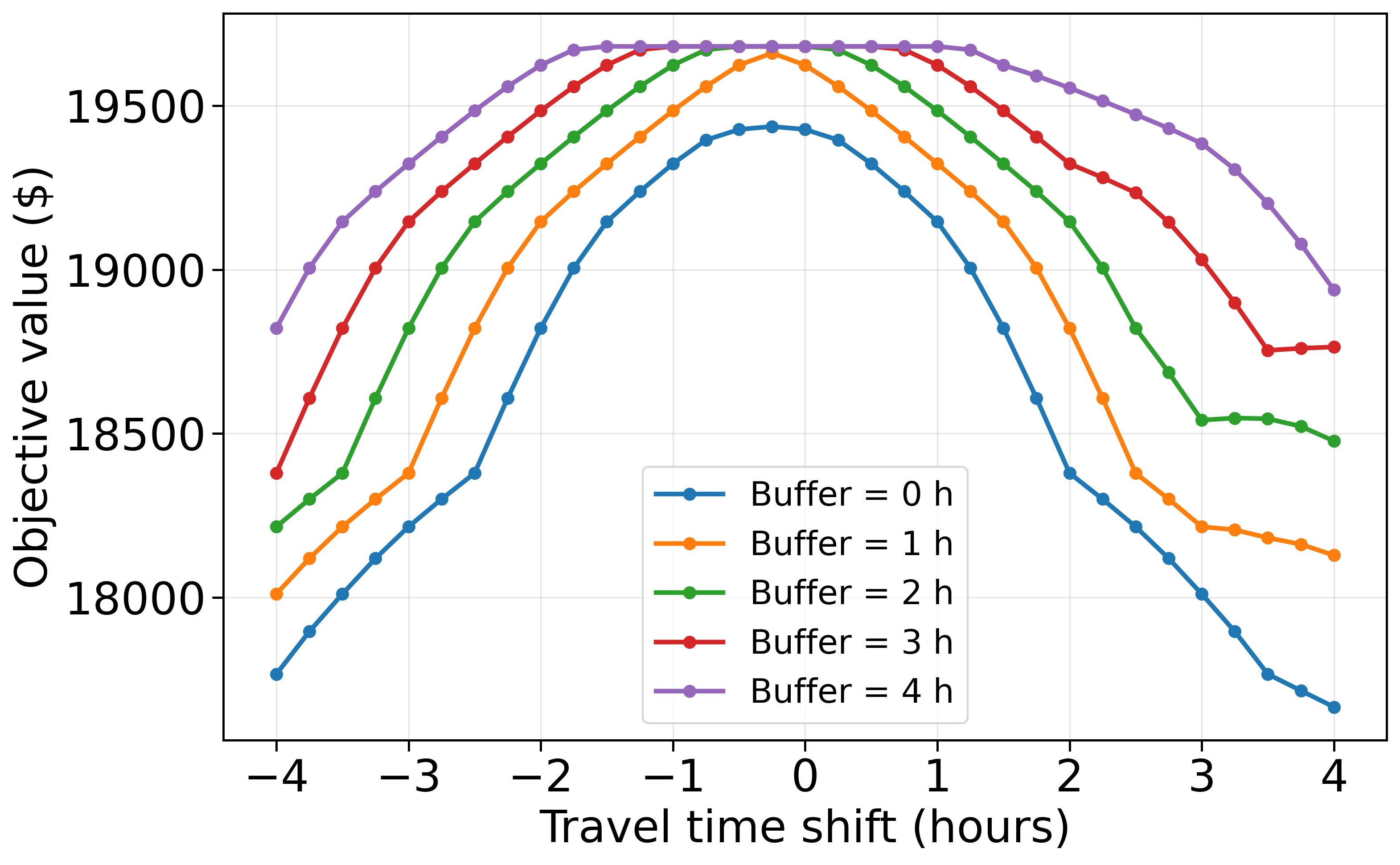}
        \caption{Objective value}
        \label{fig:obj_value}
    \end{subfigure}
    \caption{Results under different activity pattern shifts and reliability buffers.} 
    \label{fig:performance_metrics_buffer} 
\end{figure}

As shown in Figure \ref{fig:performance_metrics_buffer}, a larger buffer compromises the service performance, but the negative impact can be offset by the shift in activity patterns. For instance, with a four-hour buffer, the service rate under a four-hour activity pattern shift recovers to the level observed with zero buffer and one-hour shift. However, when the buffer is considerably large, the platform no longer sources AVs if the activity patterns of AV owners and MoD customers are not sufficiently different. 
For example, with a four-hour buffer, the platform does not implement crowdsourcing when the shift is below one hour. Moreover, with the same level of rental price, much fewer vehicle hours are obtained with longer buffers, as shown in Figures~\ref{fig:performance_metrics_buffer}(b) and (c). These results indicate that the reliability buffer, together with the activity pattern shift, plays an important role in determining the feasibility of AV crowdsourcing services. It also plays a critical role in the design and operation of such services, which requires a deeper analysis considering the heterogeneity among AV owners.

\subsection{Inter-region travel time} \label{sec_time}

\begin{figure}[!h]
    \centering
    \begin{subfigure}{0.45\linewidth}
        \centering
        \includegraphics[width=\linewidth]{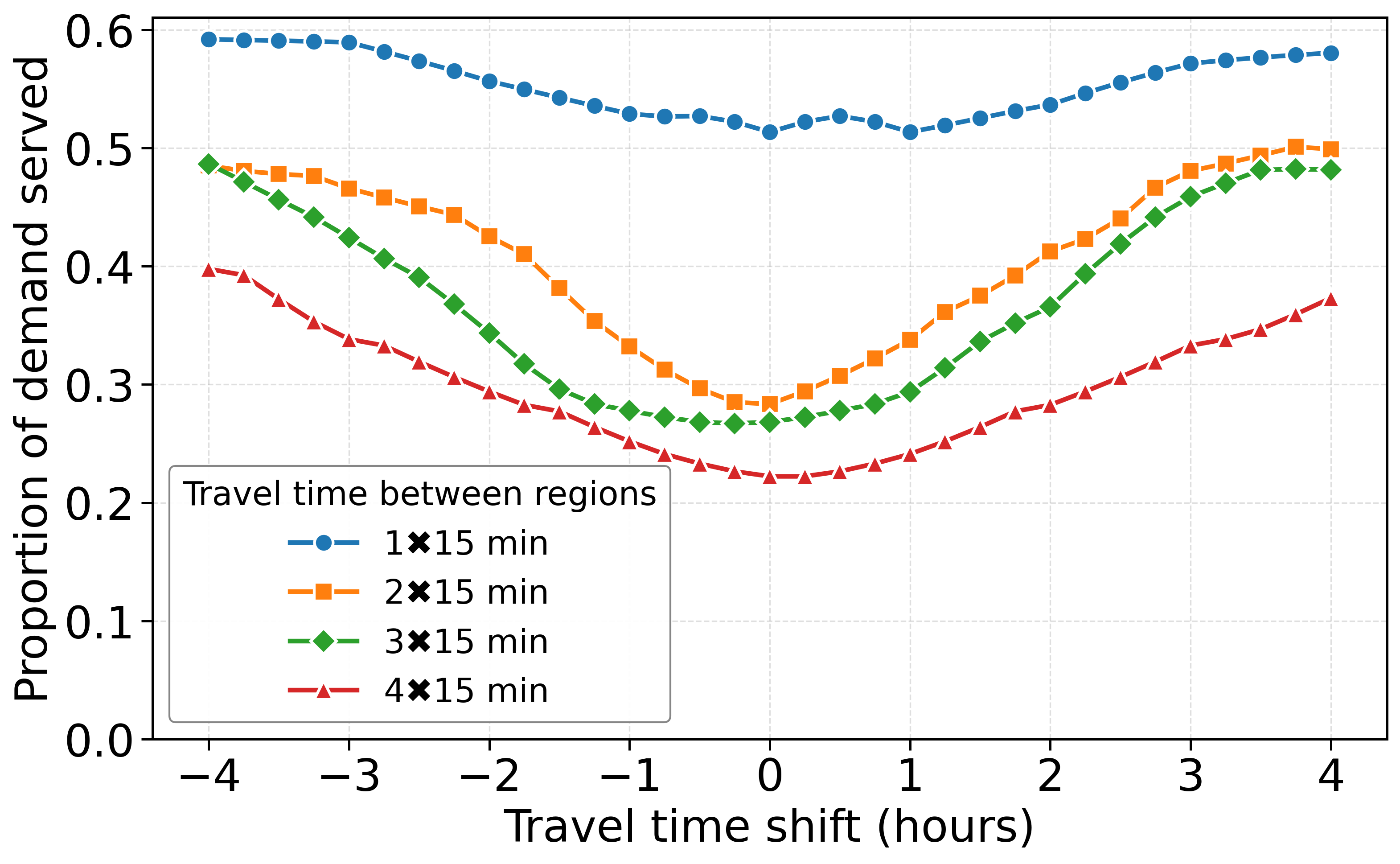}
        \caption{Proportion of demand served}
        \label{subfig:demand_rate}
    \end{subfigure}
    \hfill
    \begin{subfigure}{0.45\linewidth}
        \centering
        \includegraphics[width=\linewidth]{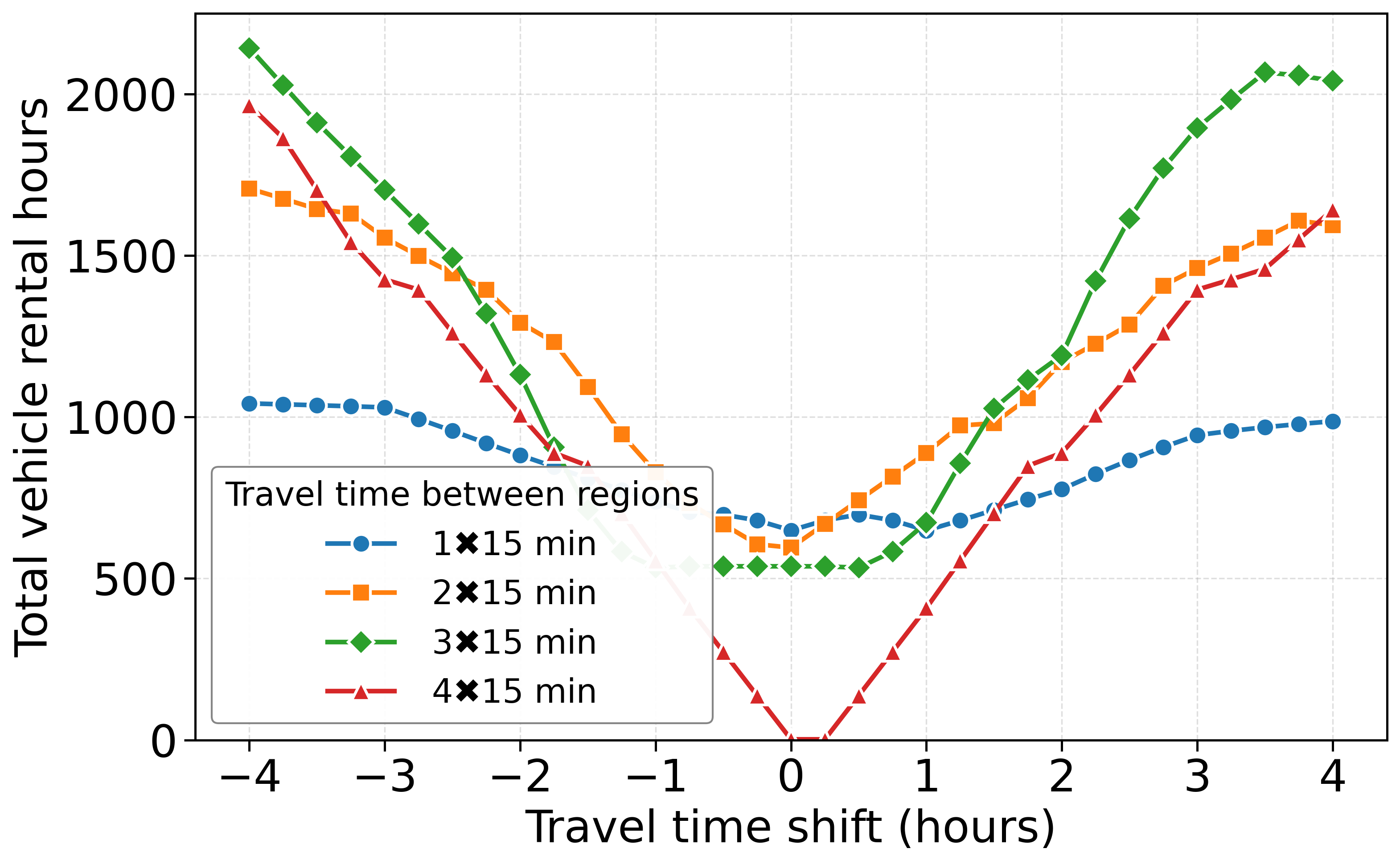}
        \caption{Total vehicle rental hours}
        \label{subfig:supply}
    \end{subfigure}
    \begin{subfigure}{0.45\linewidth}
        \centering
        \includegraphics[width=\linewidth]{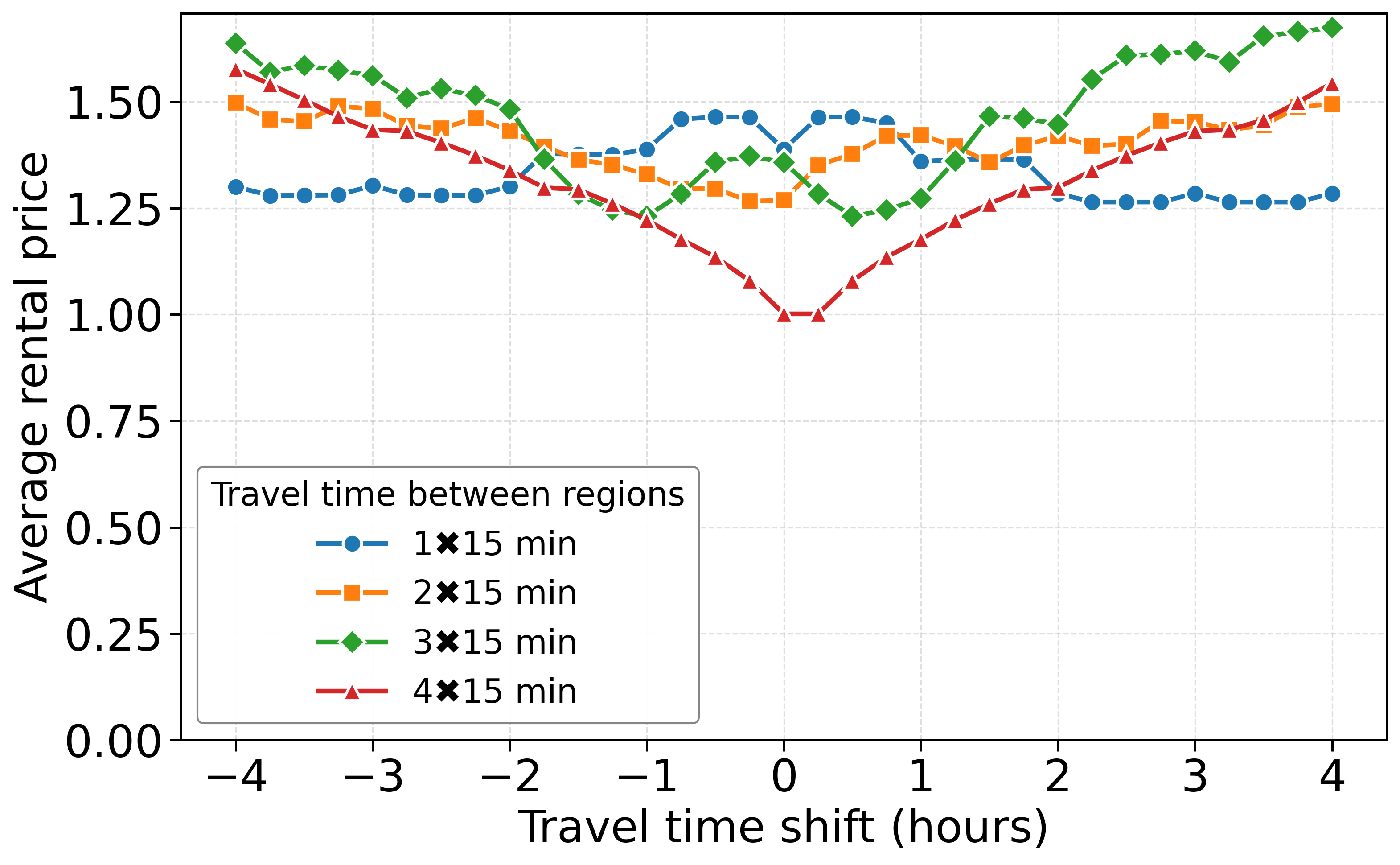}
        \caption{Average rental price}
        \label{subfig:price}
    \end{subfigure}
    \hfill
    \begin{subfigure}{0.45\linewidth}
        \centering
        \includegraphics[width=\linewidth]{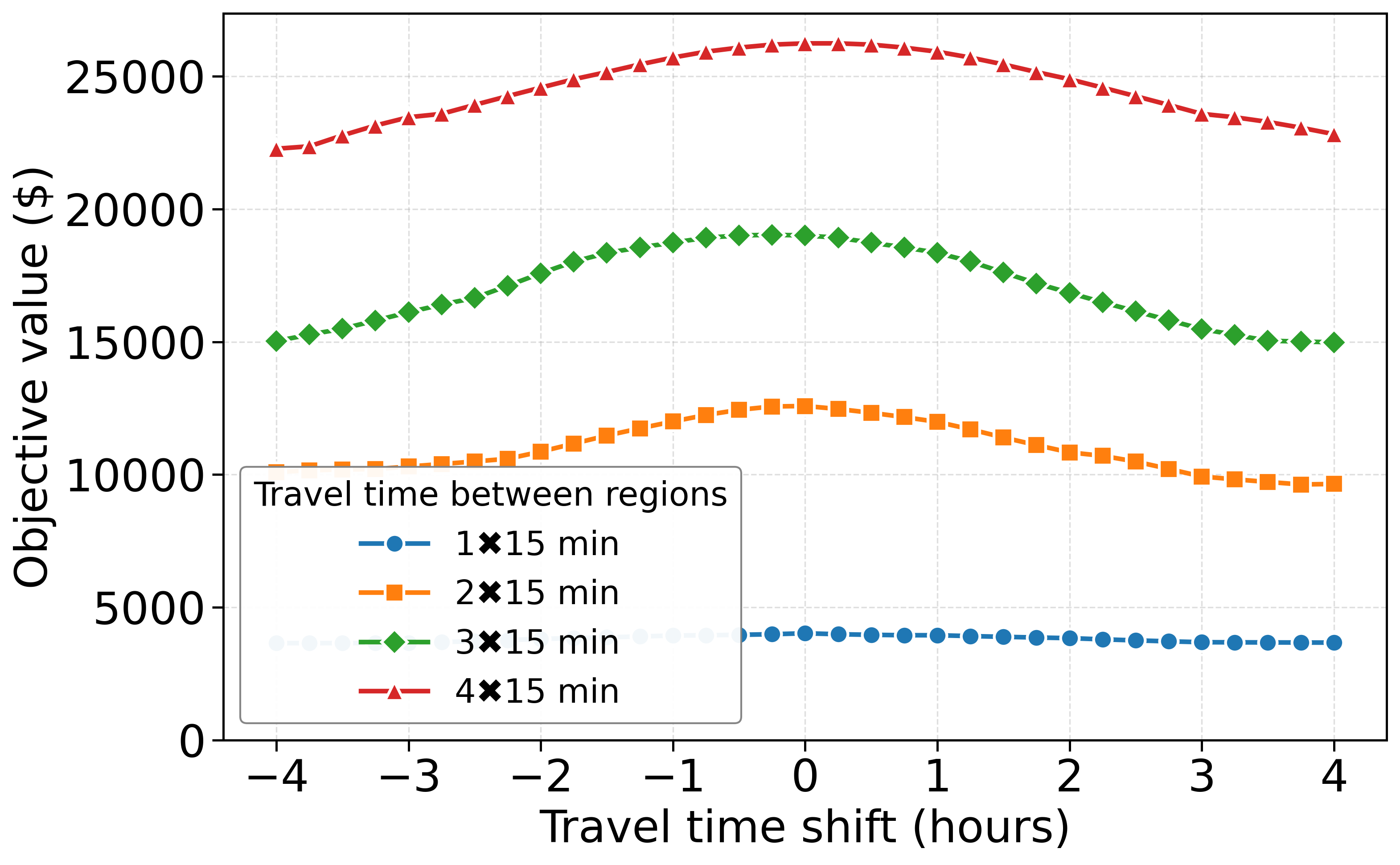}
        \caption{Objective value}
        \label{subfig:objective}
    \end{subfigure}
    \begin{subfigure}{0.45\linewidth}
        \centering
        \includegraphics[width=\linewidth]{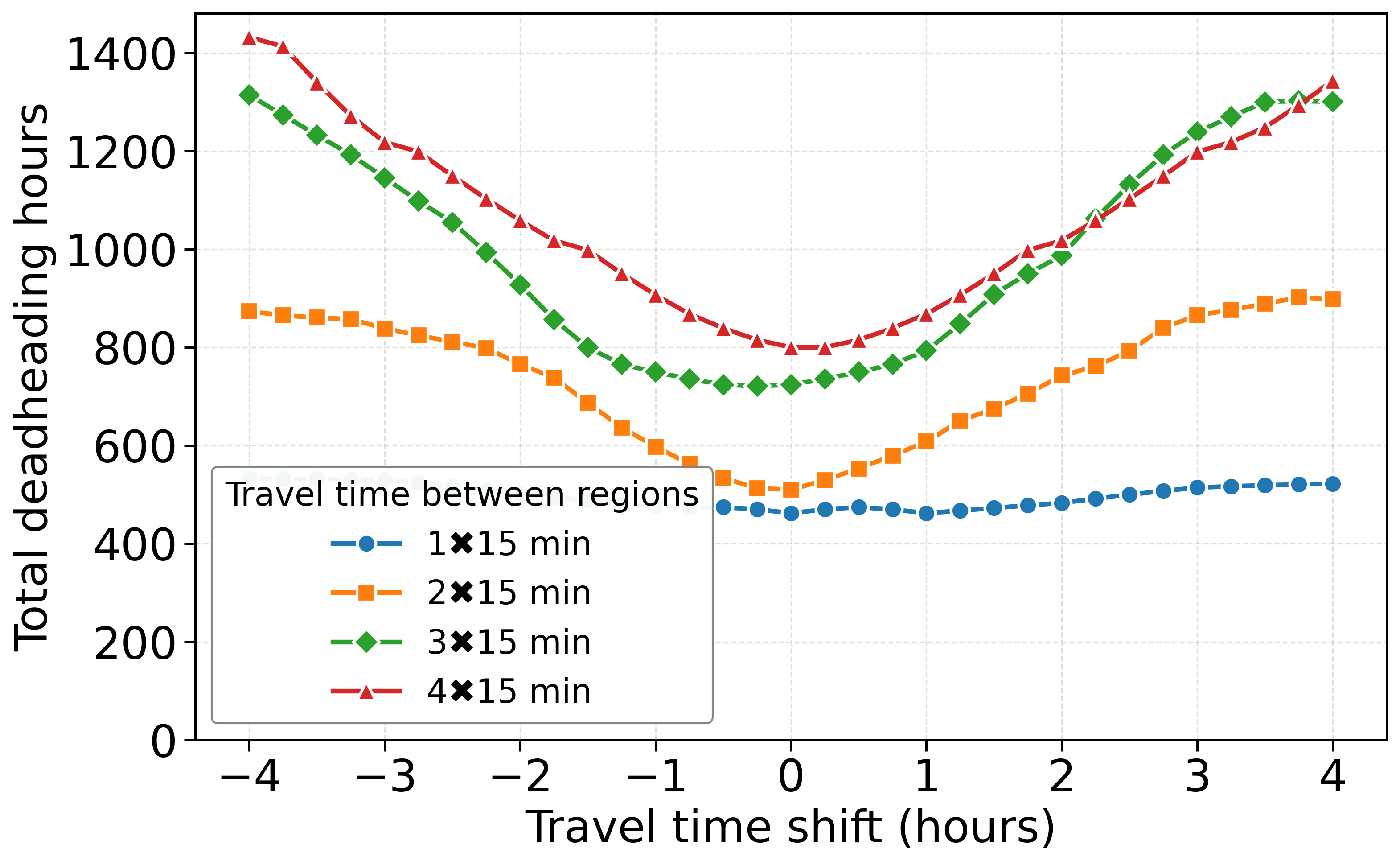}
        \caption{Total deadheading hours}
        \label{subfig:deadheading}
    \end{subfigure}
    \caption{Results under different travel pattern shifts and inter-region travel times.} 
    \label{fig:sensitivity_analysis_travel_time}
\end{figure}

The inter-region travel time is particularly important when crowdsourced AVs are required to be returned to the same locations. If there are no passenger trips in the same direction as the vehicle return trips and the inter-zone travel time is long, a large number of deadheading hours is expected. 
This result is indeed reflected in Figure~\ref{fig:sensitivity_analysis_travel_time} (e), especially when the activity pattern shift is large. In this case, more private AVs are available and also rented to fulfill the majority of passenger demand; see Figures~\ref{fig:sensitivity_analysis_travel_time} (a) and (b). In contrast, when inter-region travel time is sufficiently short, the benefit of activity pattern shifts diminishes because the platform can easily coordinate the movements of private AVs between regions. Furthermore, a slightly lower rental price is observed in Figure~\ref{fig:sensitivity_analysis_travel_time} (c), indicating an oversupply of private AVs. 

Figures~\ref{fig:sensitivity_analysis_travel_time}(b) and (d) indicate that, when the inter-region travel time is extremely high (e.g., one hour) and the activity pattern shift is relatively small,  the platform does not crowdsource any AVs, meanwhile incurring the highest cost. This indicates that the inter-region travel time, together with the activity pattern shift between owners and customers, is a determining factor affecting service feasibility. With shorter inter-region travel time, the service is more likely to be feasible. 


\section{Case study of Chicago} \label{sec_chicago}

To further validate the operational feasibility of AV crowdsourcing services, we conduct a case study of Chicago based on the 2019 household travel survey by the Chicago Metropolitan Agency of Planning~\citep{CMAP_TravelSurvey2019}, and the open-access ride-hailing data~\citep{chicago_tnp_trips_2023}.

\subsection{Model setups}\label{sec_setup}
\begin{figure}[h]
    \centering
    \begin{subfigure}[b]{0.45\linewidth}
        \centering
        \includegraphics[width=\linewidth]{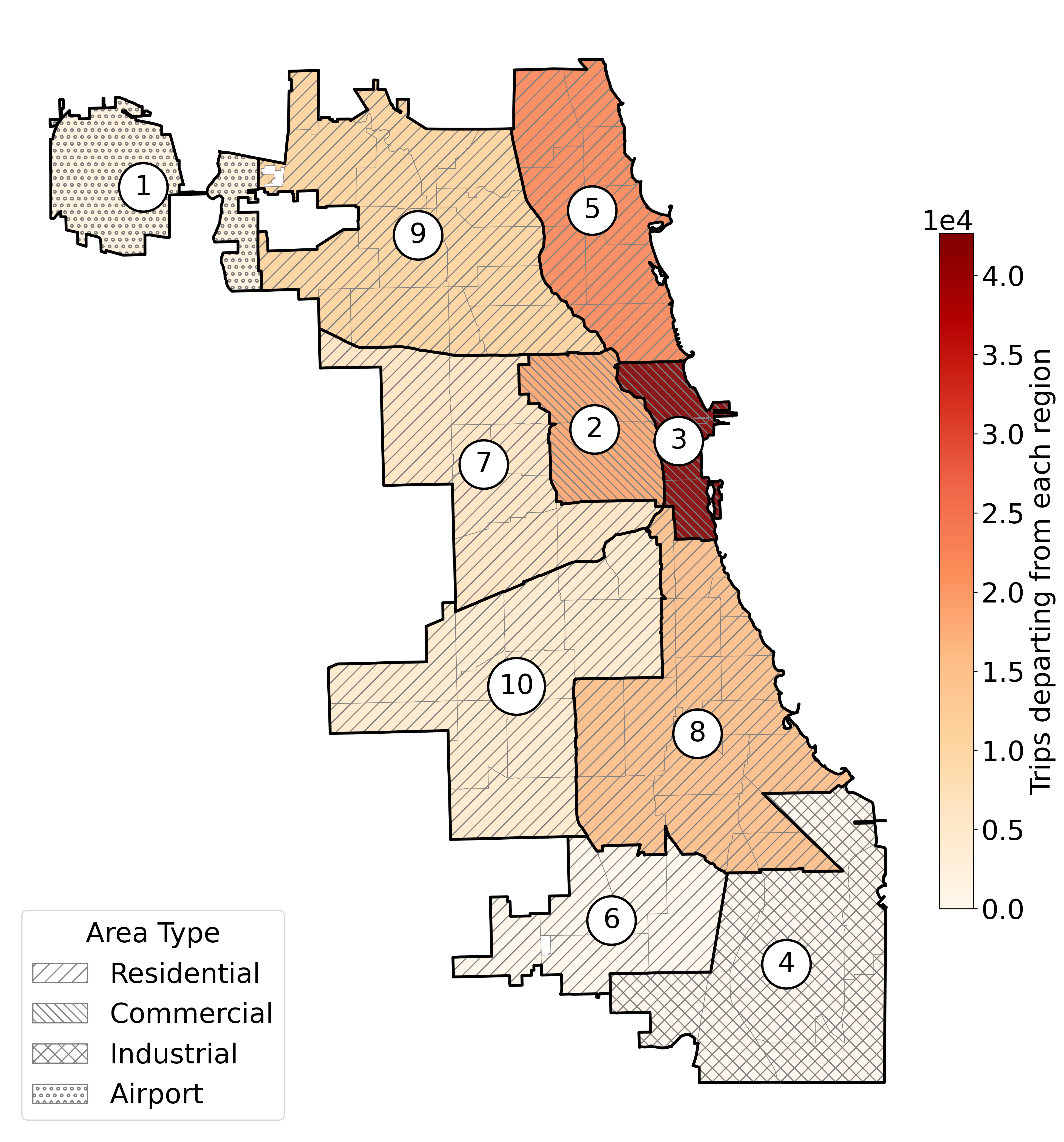}
        \caption{Potential travelers' origin distribution}
    \end{subfigure}
    \begin{subfigure}[b]{0.45\linewidth}
        \centering
        \includegraphics[width=\linewidth]{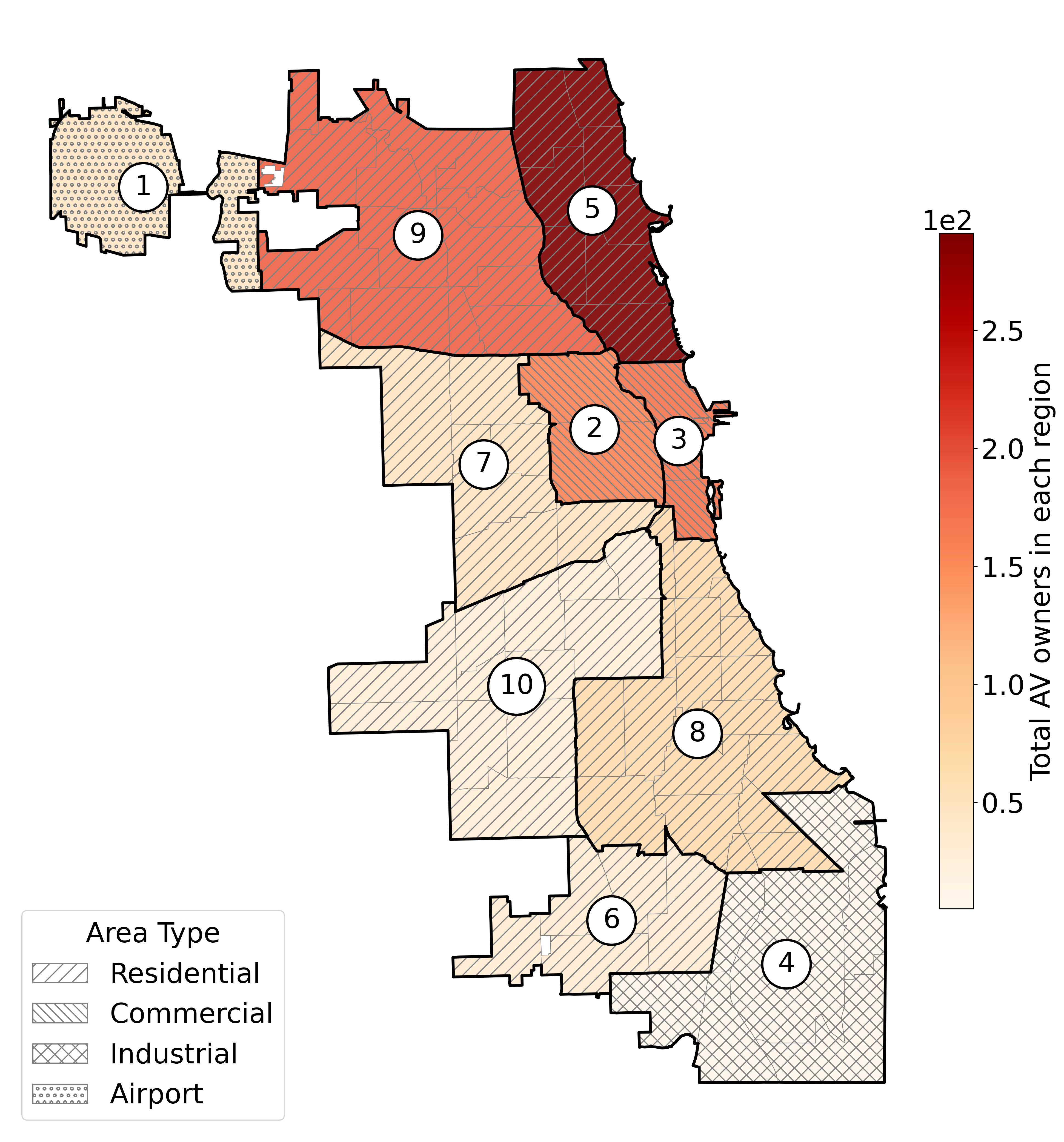}
        \caption{Potential AV owner distribution}
    \end{subfigure}
    \caption{Potential demand and supply distribution in Chicago.}
    \label{fig:chicago_comparison}
\end{figure}

We combine the land use and geographical features to partition the city of Chicago into 10 regions~\citep{ChicagoPlanningRegions}. 
The time intervals are set to $\Delta_t=15 \ \text{min}$ and $\Delta_s=1\ \text{h}$. 
Due to the lack of empirical data about private AV owners, we consider high-income individuals with private vehicles in the survey as potential AV owners and rescale it proportionally to the population. 
The spatial distribution of AV owners is shown in Figure~\ref{fig:chicago_comparison} (b).
Following \cite{valente2019sharing}, we assume the opportunity cost of AV owners follows a single uniform distribution while setting $\underline{\varepsilon} = 1$ (\$/hr) and $\overline{\varepsilon} = 25$ (\$/hr), while a further characterization of their activity patterns is detailed in Section~\ref{sec_cluster_owner}. 
The number of platform-owned vehicles is $N^\nu=100$.  

To estimate the MoD service demand pattern, we extracted the ride-hailing trip data on November 20, 2024, filtered trips that both started and ended in the study region, and mapped them to time periods and regions specified in this case study. The spatial distribution of trip origin is shown in Figure~\ref{fig:chicago_comparison} (a), while the temporal demand pattern is illustrated in Figure~\ref{fig: intraday_chicago_demand}. 
\begin{figure}[h]
    \centering
    \includegraphics[width=0.65\linewidth]{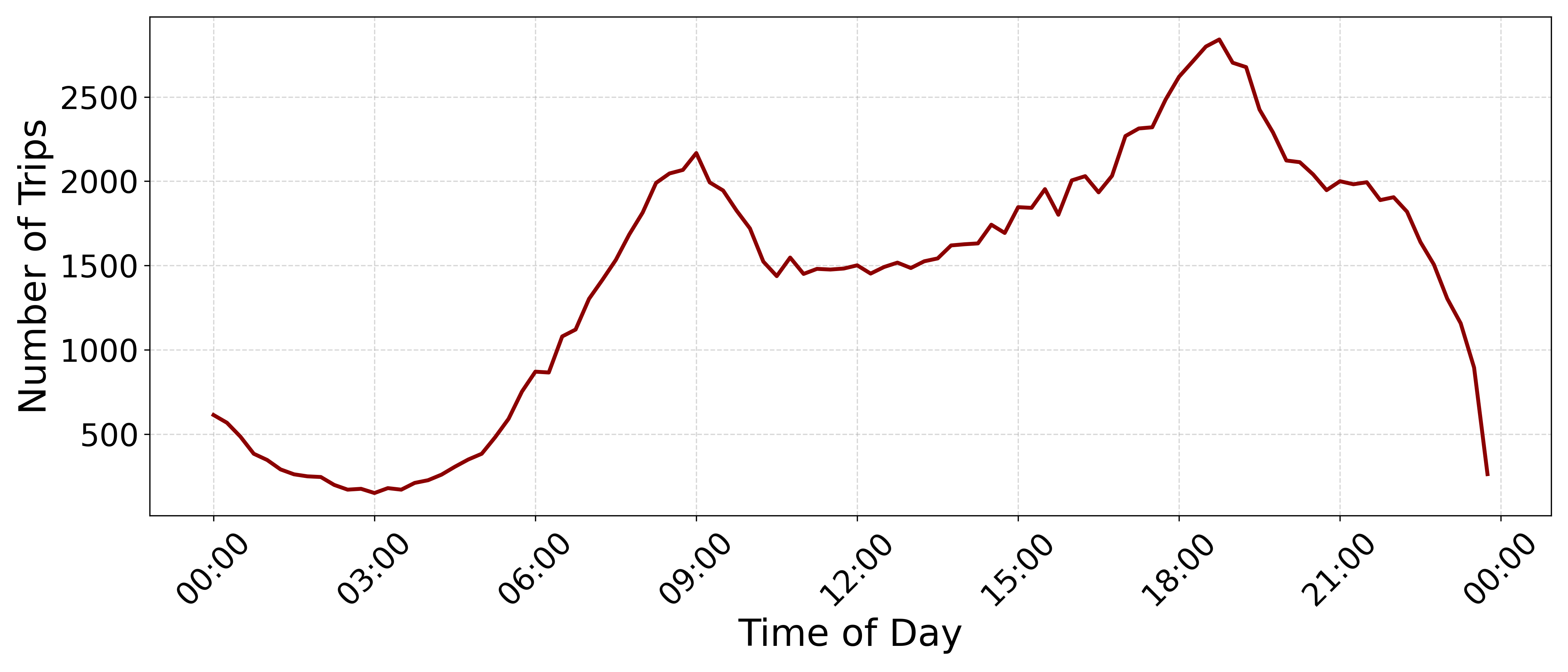}
    \caption{Variation in demand over the day}
    \label{fig: intraday_chicago_demand}
\end{figure}
\begin{figure}[h]
    \centering
    \includegraphics[width=0.45\linewidth]{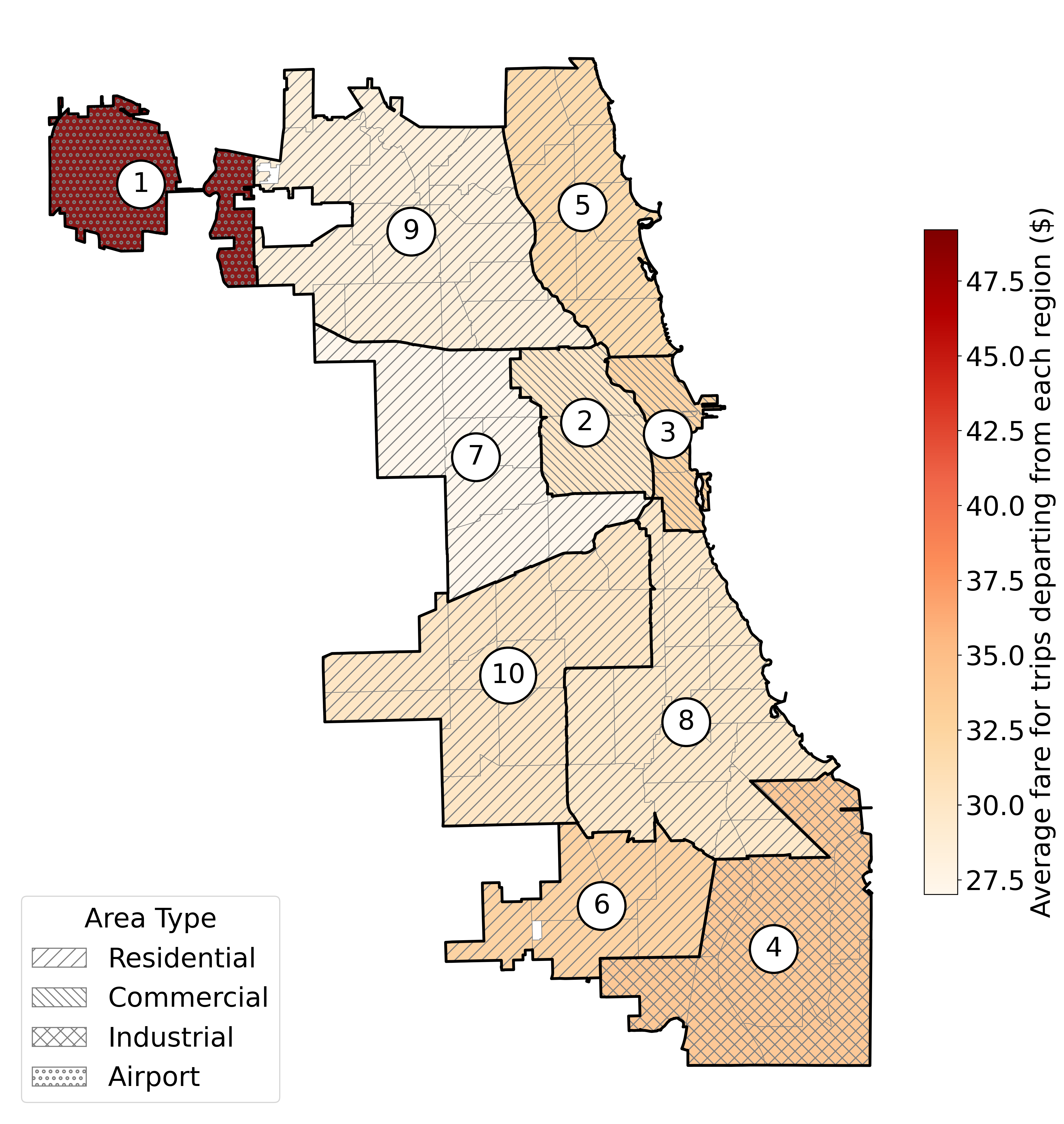}
    \caption{Average trip fare by region}
    \label{fig:fare}
\end{figure}
The travel times and demand loss penalties are also estimated from ride-hailing trip data, while the latter is set as the average trip fare presented in Figure~\ref{fig:fare}. 

\subsection{Characterization of AV owners} \label{sec_cluster_owner}

We select vehicle owners in the household survey with a yearly income higher than \$150,000 as potential AV owners, which corresponds to 699 individuals and accounts for 21.4\% of all vehicle owners in the survey. 
To generate representative daily activity patterns, we cluster the high-income vehicle owners using the eight features in Table \ref{tab:classification}, which cover the temporal, spatial, and activity-related aspects. All features are standardized using Z-scores to ensure zero mean and unit variance.
\begin{table}[h]
\centering
\caption{Indicators for clustering owners}
\label{tab:classification}
\begin{tabular}{p{4.5cm} p{11cm}}
\hline
\textbf{Indicators} & \textbf{Description} \\
\hline
Time centroid & The average beginning time of daily activities. \\
Time standard deviation & The standard deviation of daily activity start times. \\
Total displacement & The cumulative Haversine distance between consecutive activity regions. \\
Region coverage & The total number of distinct regions visited within a day. \\
Inter-activity travel speed & The average speed during inter-activity transitions. \\
Local activity ratio & The proportion of consecutive activities occurring within the same region. \\
Dominant activity ratio & The maximum time spent on a single purpose (e.g., work) relative to the total idle duration. \\
Purpose variety & The ratio of unique trip purposes by an owner to the total activity purpose categories. \\
\hline
\end{tabular}
\end{table}

\begin{table}[!h]
\centering
\caption{Representative profiles of AV owners}
\begin{tabular}{c p{2.5cm} p{8cm} c}
\hline
\textbf{Category ID} & \textbf{Type} & \textbf{Descriptions/Notes} & \textbf{Population} \\ \hline
1 & Home-oriented & Feasible except during evening peak hours. & 238 (34\%)\\
2 & Commuter & Residing outside Chicago and only feasible in the morning. & 179 (26\%) \\
3 & Commuter & Intra-city commuter.  & 98 (14\%)\\
4 & Commuter & Intra-city commuter.  & 55 (8\%)\\
5 & Remote worker & Working from home but using own vehicles during peak hours. & 42 (6\%)\\
6 & Companion & Accompanying someone and only feasible during evening peak hours. & 39 (6\%)\\
7 & Commuter & Intra-city commuter & 26 (4\%)\\
8 & Commuter &  Additional non-shopping errands (e.g., banking, post office, government) during evening peak hours. & 22 (3\%)\\ \hline
\end{tabular} \label{table_owner_type}
\end{table}

We apply the K-means clustering algorithm and select the results with eight clusters that yield a silhouette score of 0.53 and a reasonably good separation among clusters (see Figure~\ref{fig:Principal Component Analysis}). 
For each cluster, the representative AV owner is defined as the individual closest to the cluster centroid in terms of Euclidean distance. As summarized in Table \ref{table_owner_type}, the largest cluster, accounting for 34\% of the high-income vehicle owners, is characterized by a home-oriented representative pattern. Commuter-related subgroups collectively account for the majority of AV owners. Notably, non-commuters primarily remain at home but still travel during peak hours, making these time slots infeasible for sharing private vehicles.

\begin{figure}[!h]
    \centering
    \includegraphics[width=0.85\linewidth]{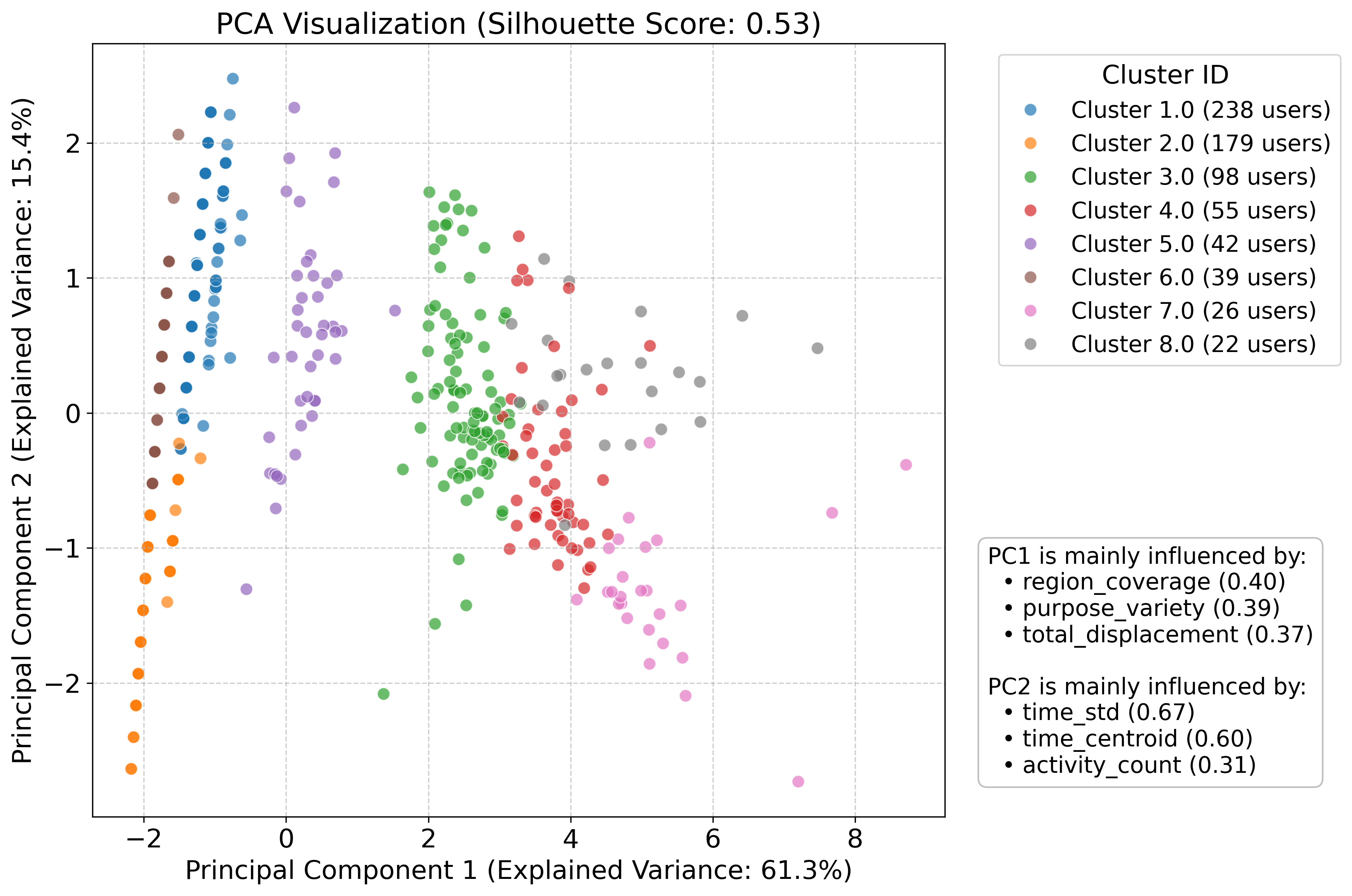}
    \caption{Principal Component Analysis}
    \label{fig:Principal Component Analysis}
\end{figure}

Figure \ref{fig:Principal Component Analysis} presents the principal component analysis (PCA) results and clustering visualization of AV owners. The first two principal components explain 76.7\% of the total variance, with PC1 accounting for 61.3\% and PC2 for 15.4\%. PC1 is primarily associated with spatial activity extent, as indicated by high loadings on region coverage, purpose variety, and total displacement. PC2 mainly reflects temporal behavior, characterized by the standard deviation of activity start times, the average activity beginning time, and the number of activities. In summary, the clustering structure shows clear separation along the spatial dimension, further refined by temporal patterns.

\subsection{Results of base scenario} \label{sec_case_result}
In Chicago case study, Gurobi 12.0.3 terminated with a final primal-dual gap below 0.2\%, indicating a near-optimal solution.
Using the data, we present the results of a case study on rental pricing strategies and dispatching strategies, as well as performance-related results, including service rate and vehicle empty hours.  
Figure~\ref{fig:price_chicago} presents the average rental price, along with variance, over the day. It can be seen that the price closely follows the passenger demand pattern in Figure~\ref{fig: intraday_chicago_demand}. 
The synchronization is due to two main factors. On the one hand, the platform-owned vehicles are not sufficient to serve the high passenger demand during peak hours, which forces the platform to increase the rental price to attract more AV owners. On the other hand, many AV owners also travel during the peaks, making it more expensive to rent private AVs. 
\begin{figure}[h] 
    \centering
    \includegraphics[width=0.6\linewidth]{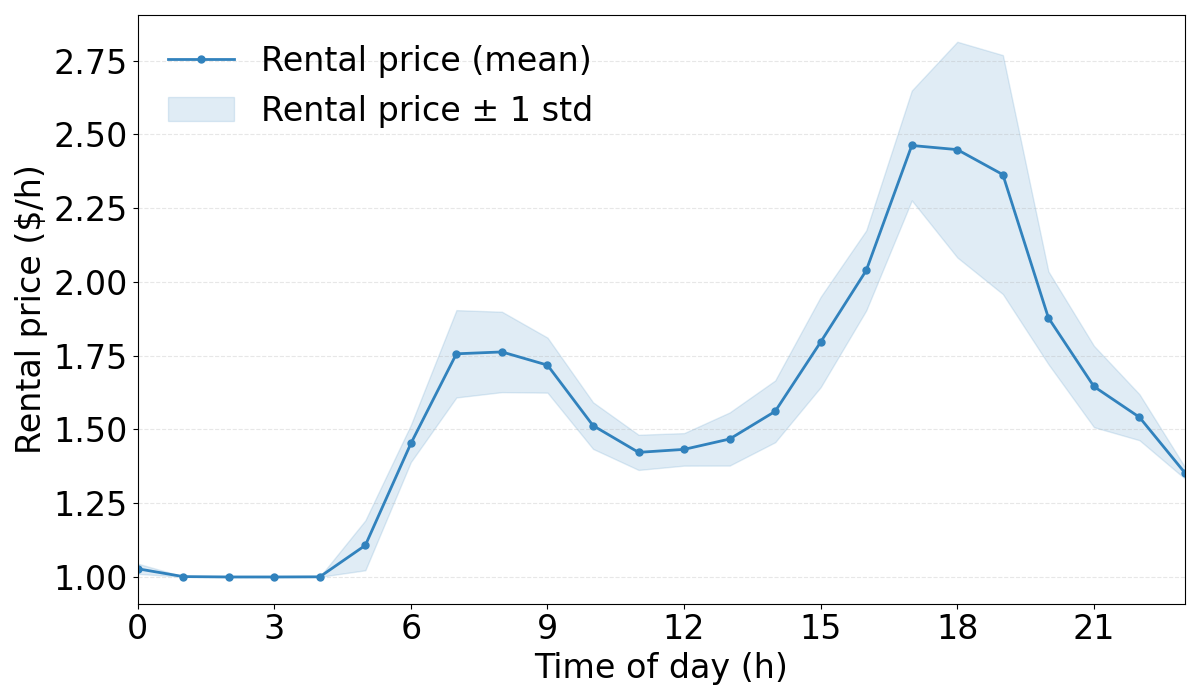}
    \caption{Average rental prices with variability}
    \label{fig:price_chicago}
\end{figure}
Figure~\ref{fig:price_chicago_boxplot} presents a box plot showing the distribution of rental participation among AV owners in each cluster, measured as the daily average proportion of owners who choose to rent. Overall, rental participation rates are relatively low in magnitude, ranging roughly from near zero to about $2.5\%$ across clusters. Most clusters concentrate within a narrow range of approximately $1\%$ to $2\%$, indicating generally low but non-negligible participation rate. Clusters 2 and 6 are located at the lower end of the distribution. When combined with the owner profiles, we observe that these two clusters have the fewest available time slots throughout the day, which in turn results in lower average participation rates.
\begin{figure}[h]
    \centering
    \includegraphics[width=0.7\linewidth]{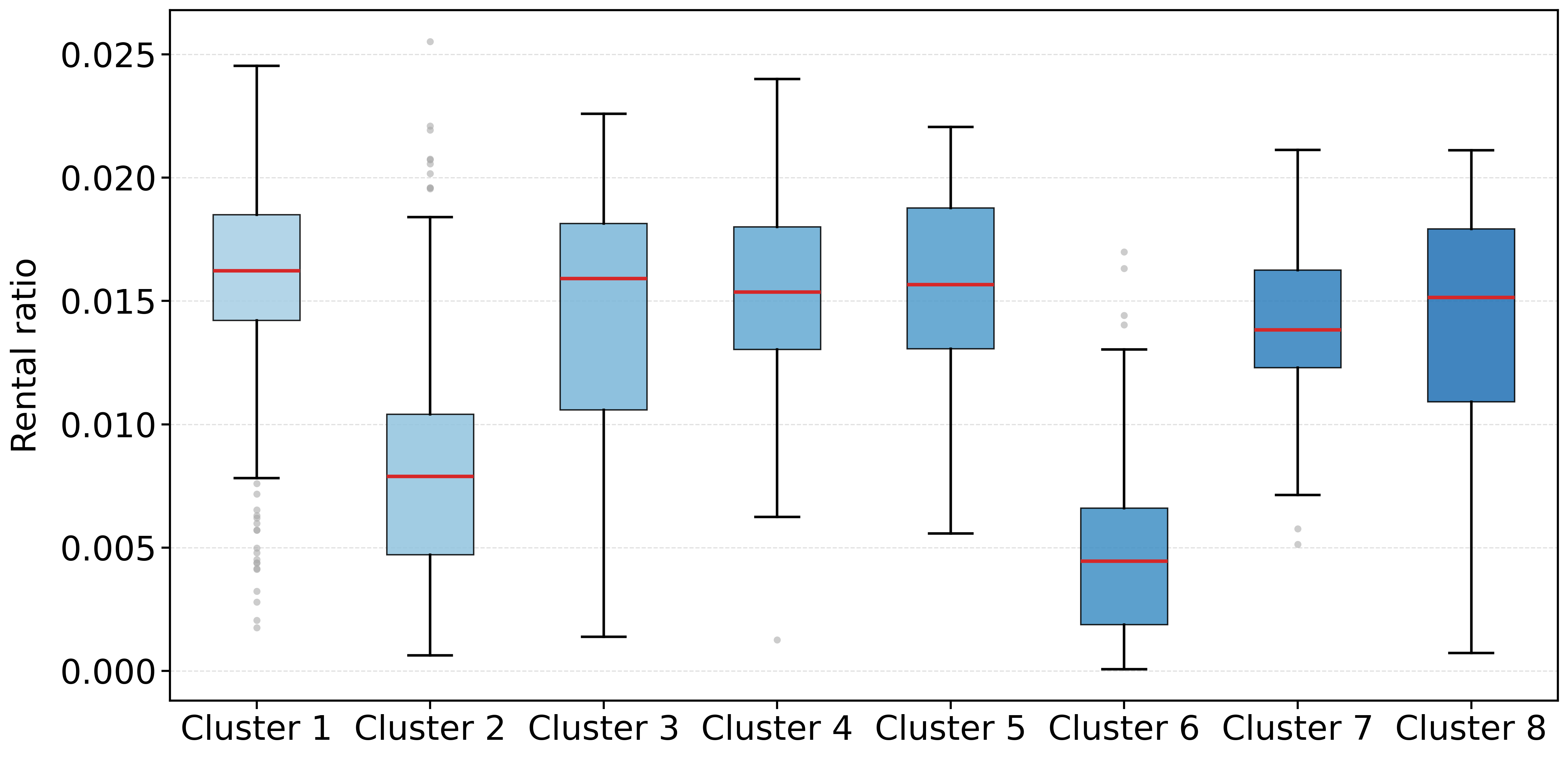}
    \caption{Box plot of rental ratios by traveler cluster}
    \label{fig:price_chicago_boxplot}
\end{figure}

\begin{figure}[!h]
    \centering
    \begin{subfigure}{0.24\textwidth}
        \includegraphics[width=\textwidth]{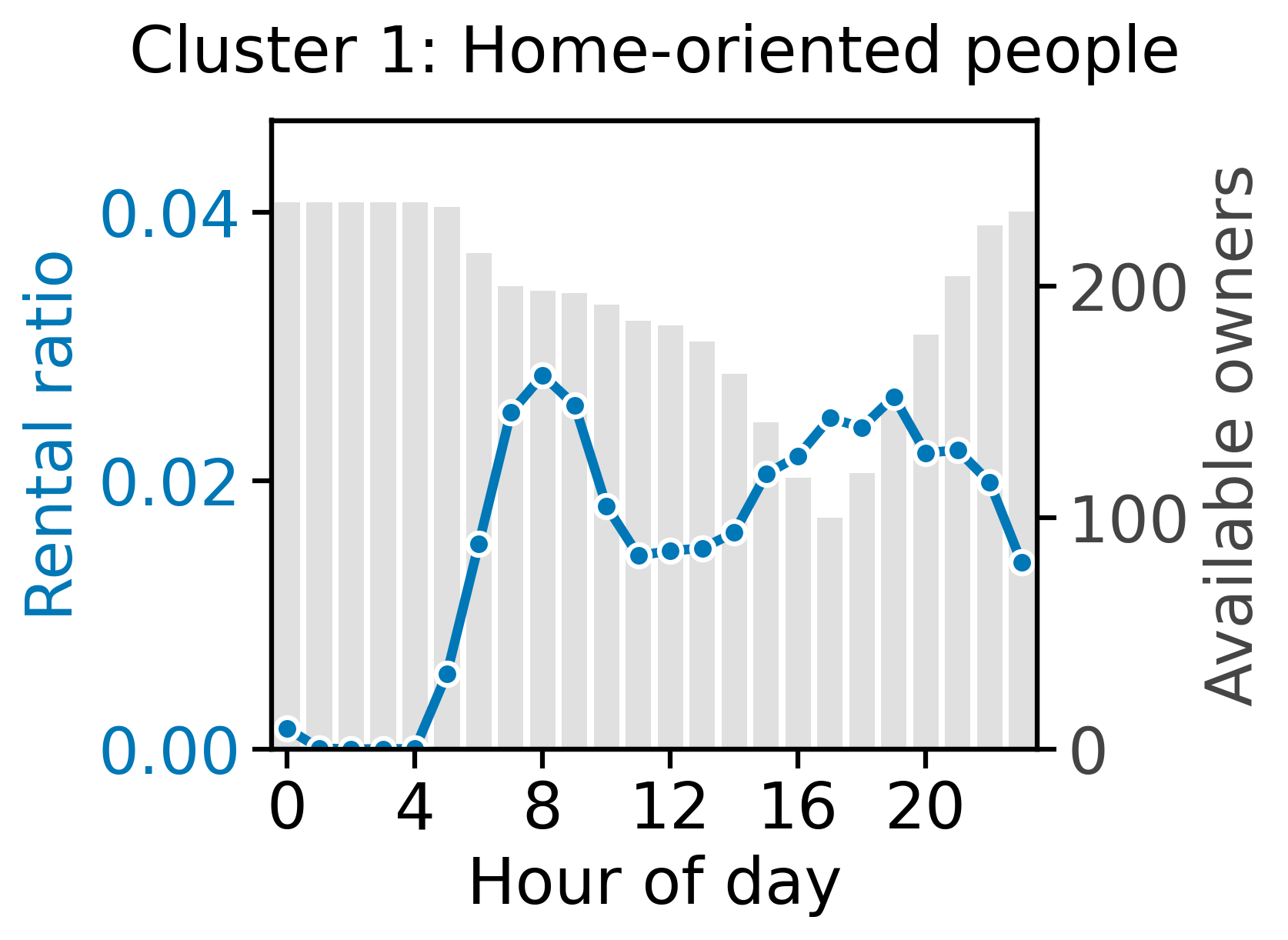}
    \end{subfigure}
    \hfill 
    \begin{subfigure}{0.24\textwidth}
        \includegraphics[width=\textwidth]{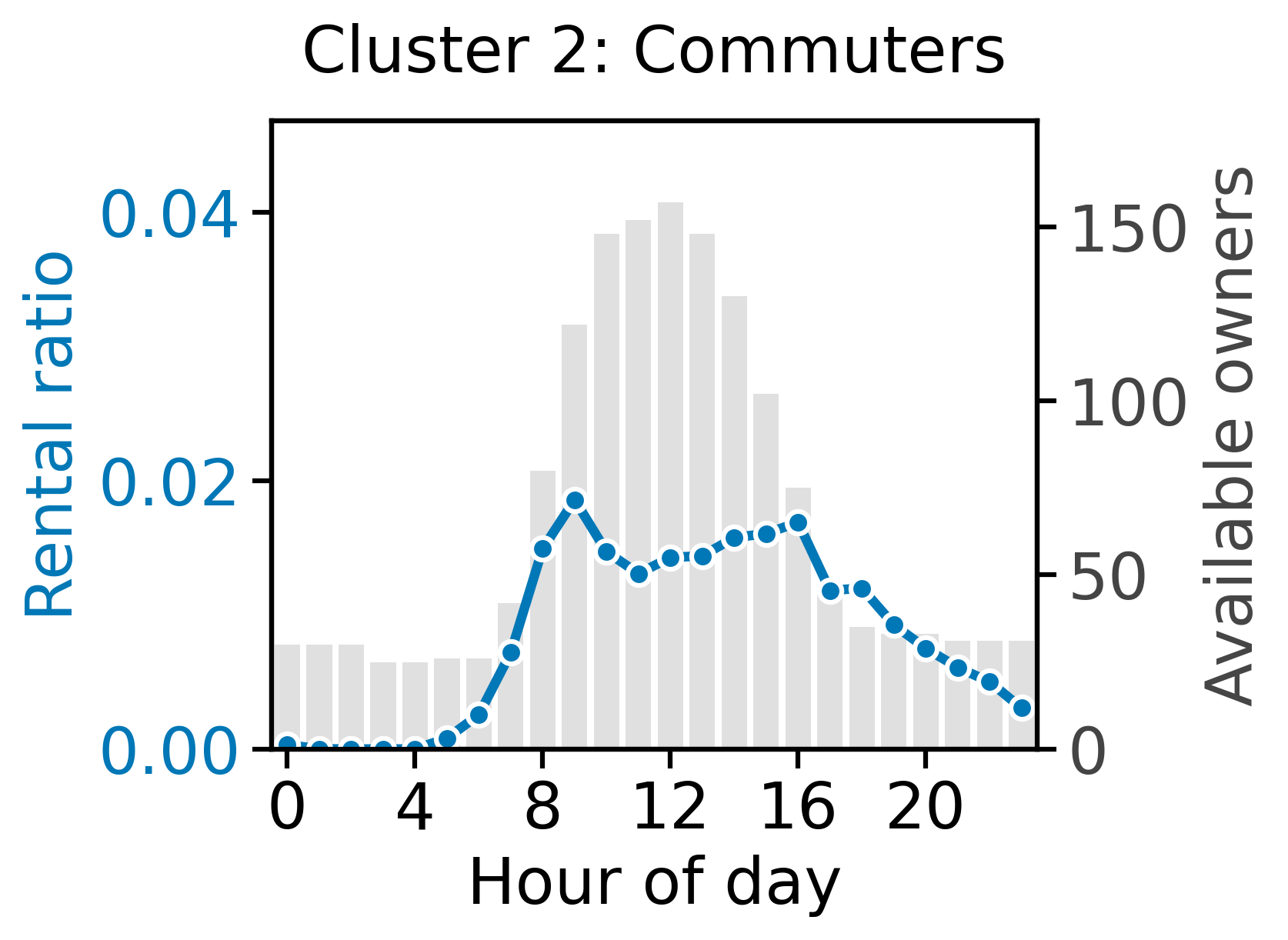}
    \end{subfigure}
    \hfill
    \begin{subfigure}{0.24\textwidth}
        \includegraphics[width=\textwidth]{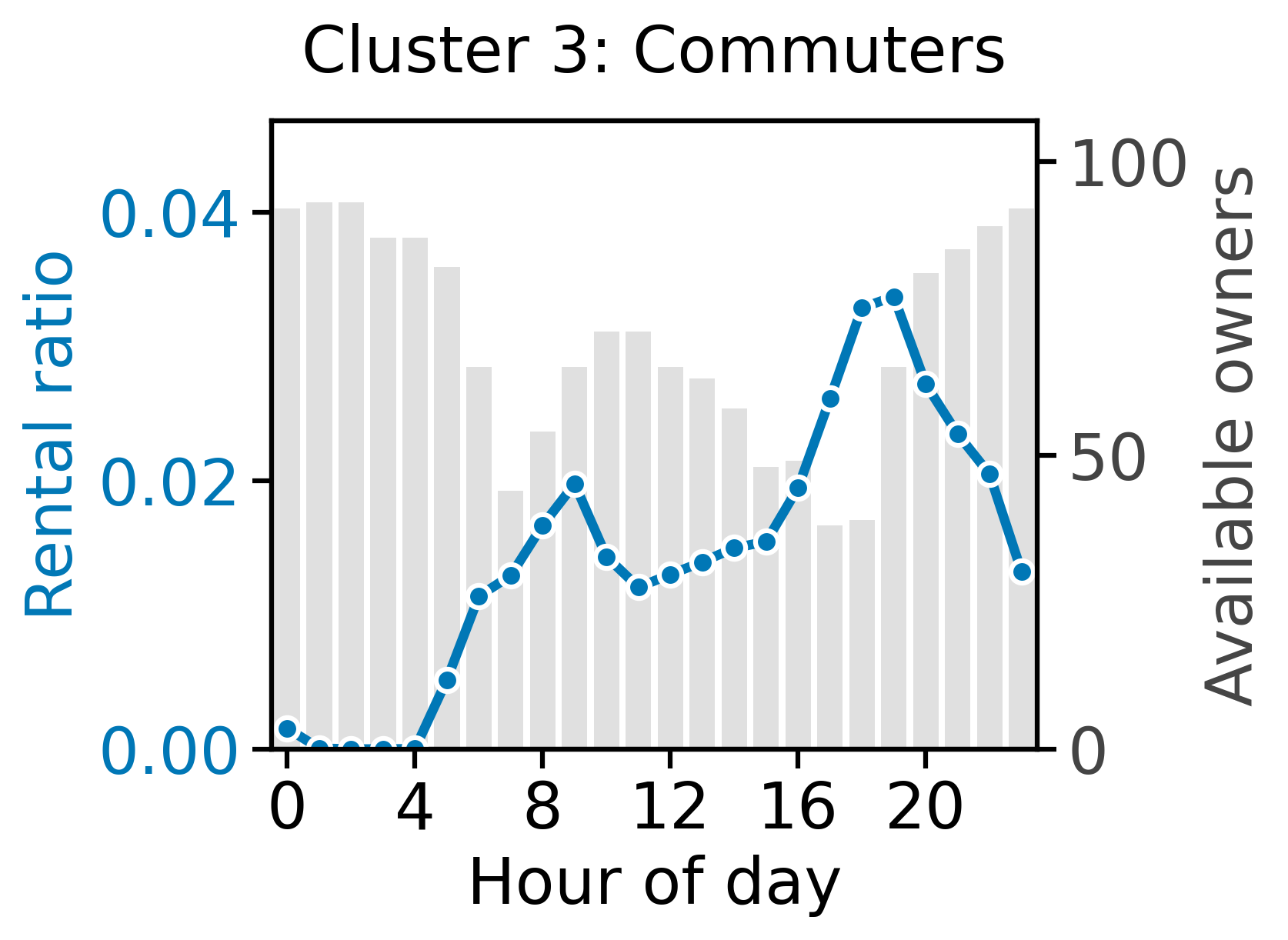}
    \end{subfigure}
    \hfill
    \begin{subfigure}{0.24\textwidth}
        \includegraphics[width=\textwidth]{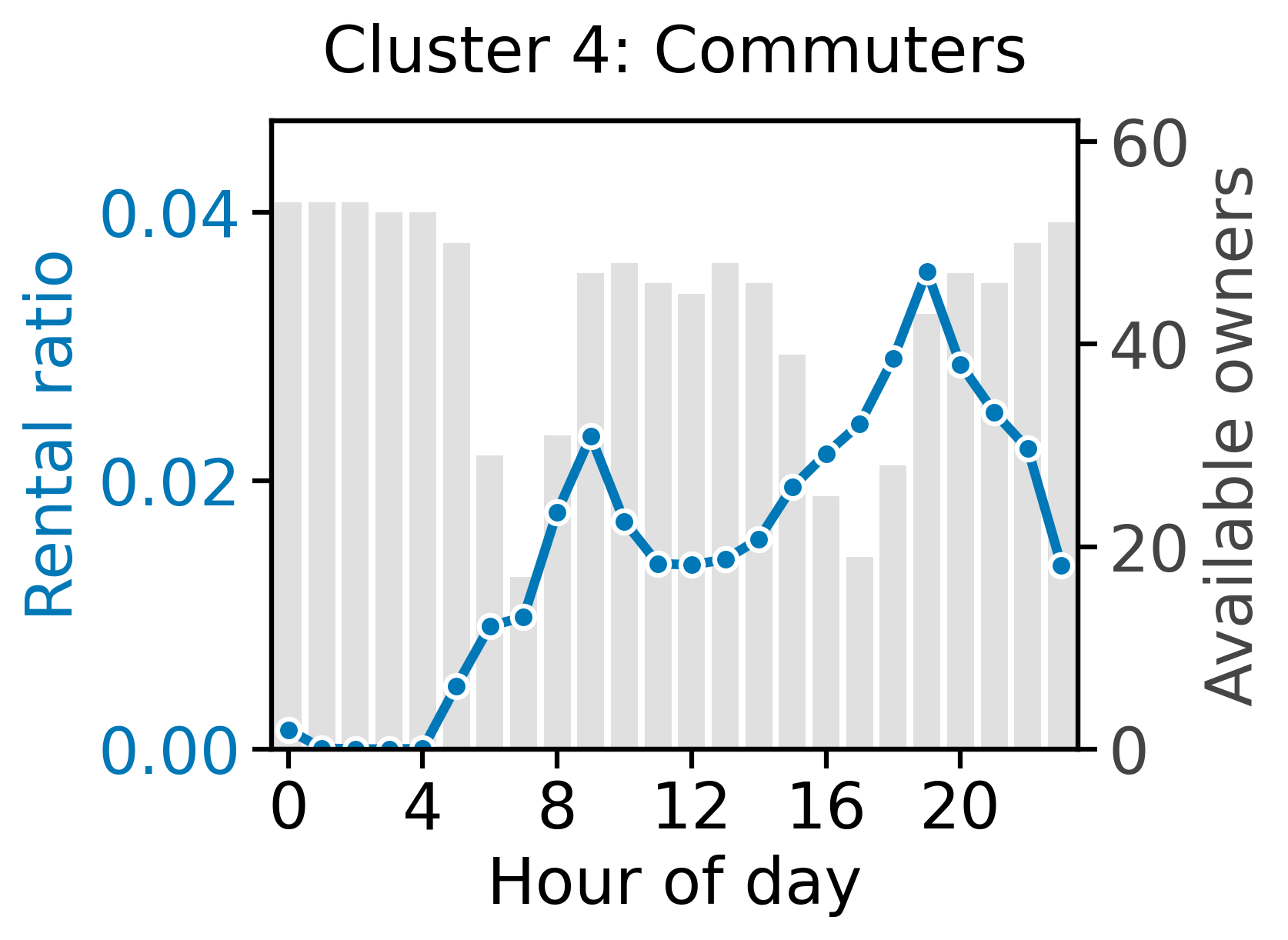}
    \end{subfigure}
    \vspace{1em} 
    \begin{subfigure}{0.24\textwidth}
        \includegraphics[width=\textwidth]{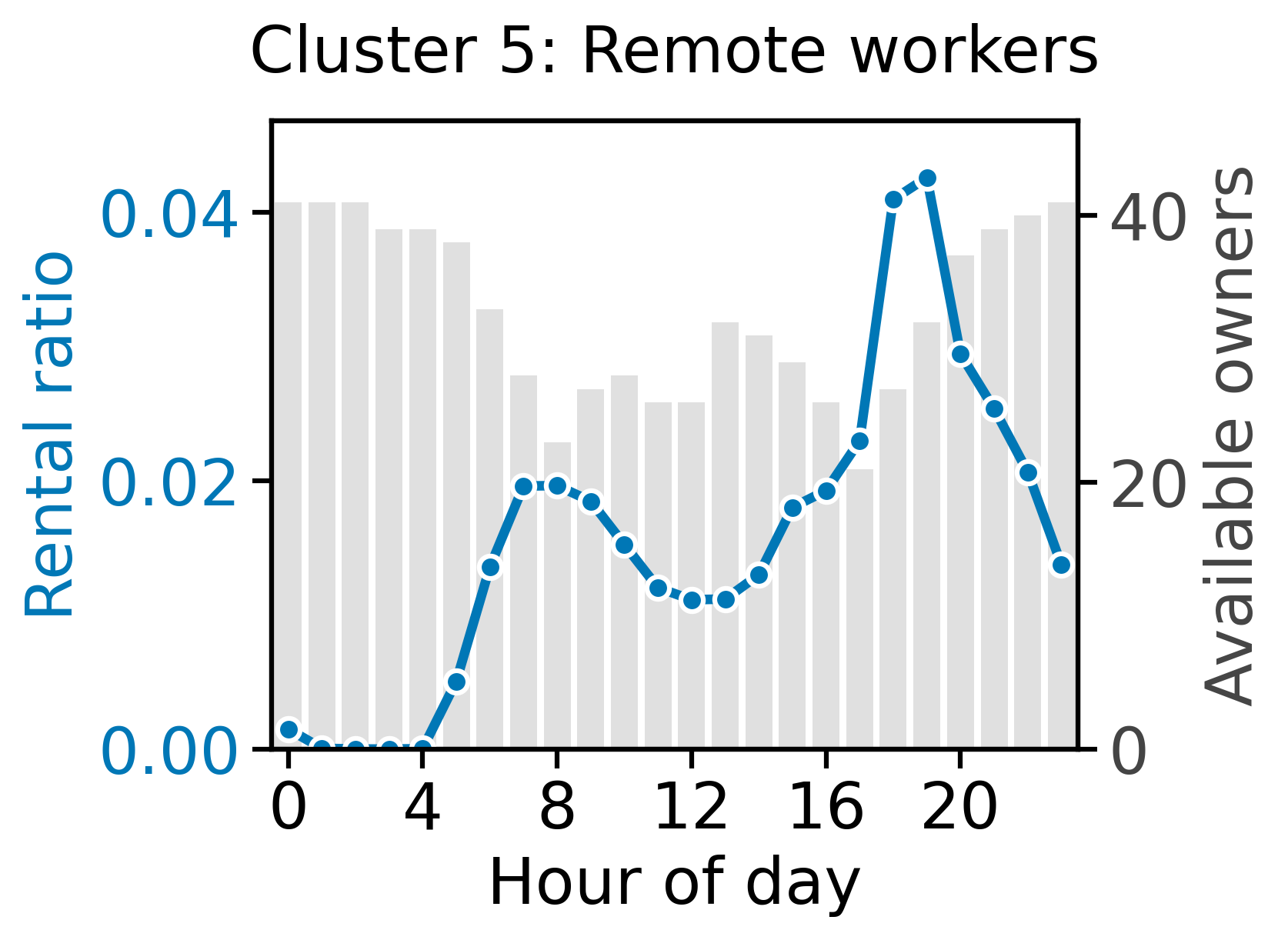}
    \end{subfigure}
    \hfill
    \begin{subfigure}{0.24\textwidth}
        \includegraphics[width=\textwidth]{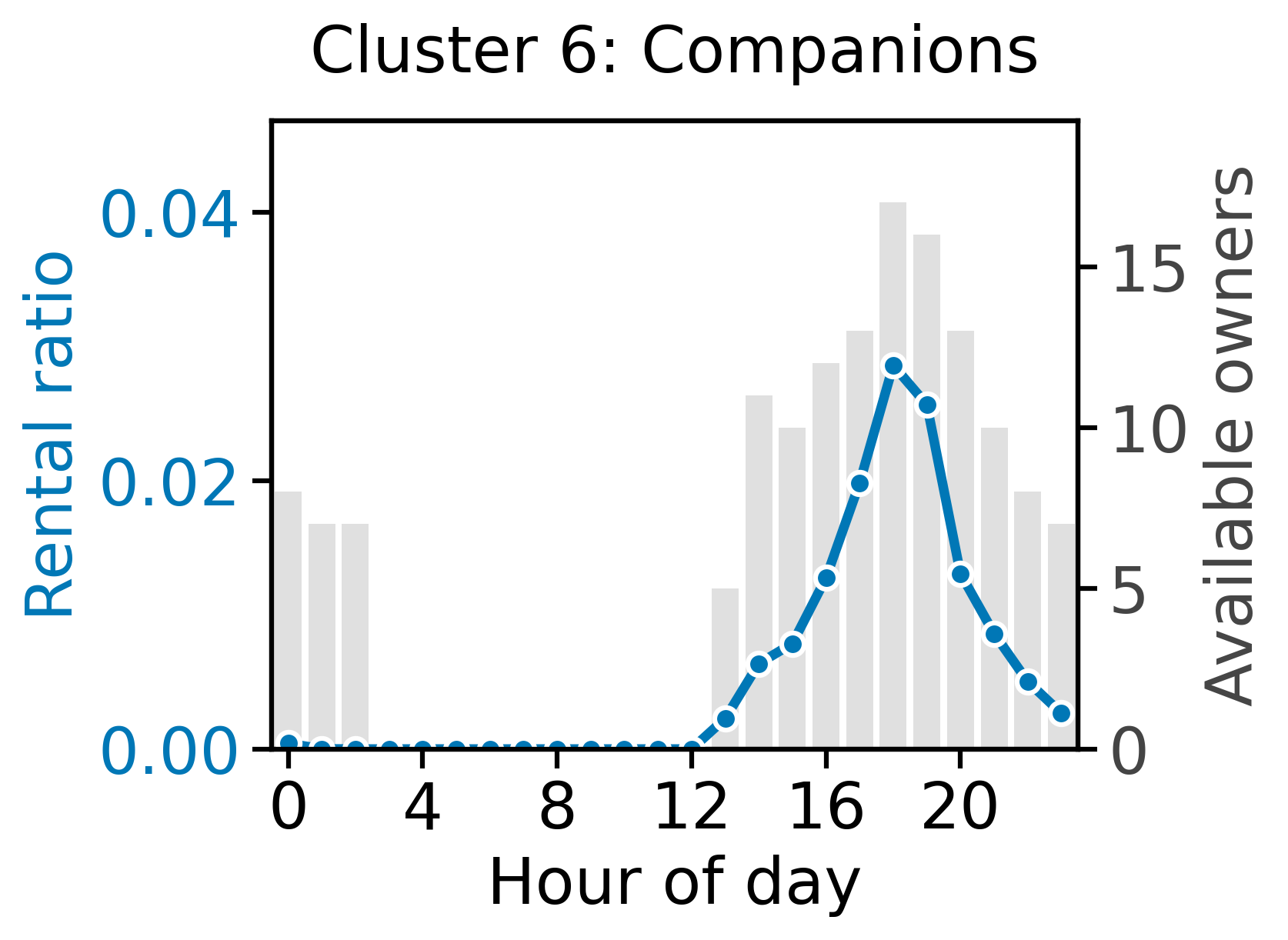}
    \end{subfigure}
    \hfill
    \begin{subfigure}{0.24\textwidth}
        \includegraphics[width=\textwidth]{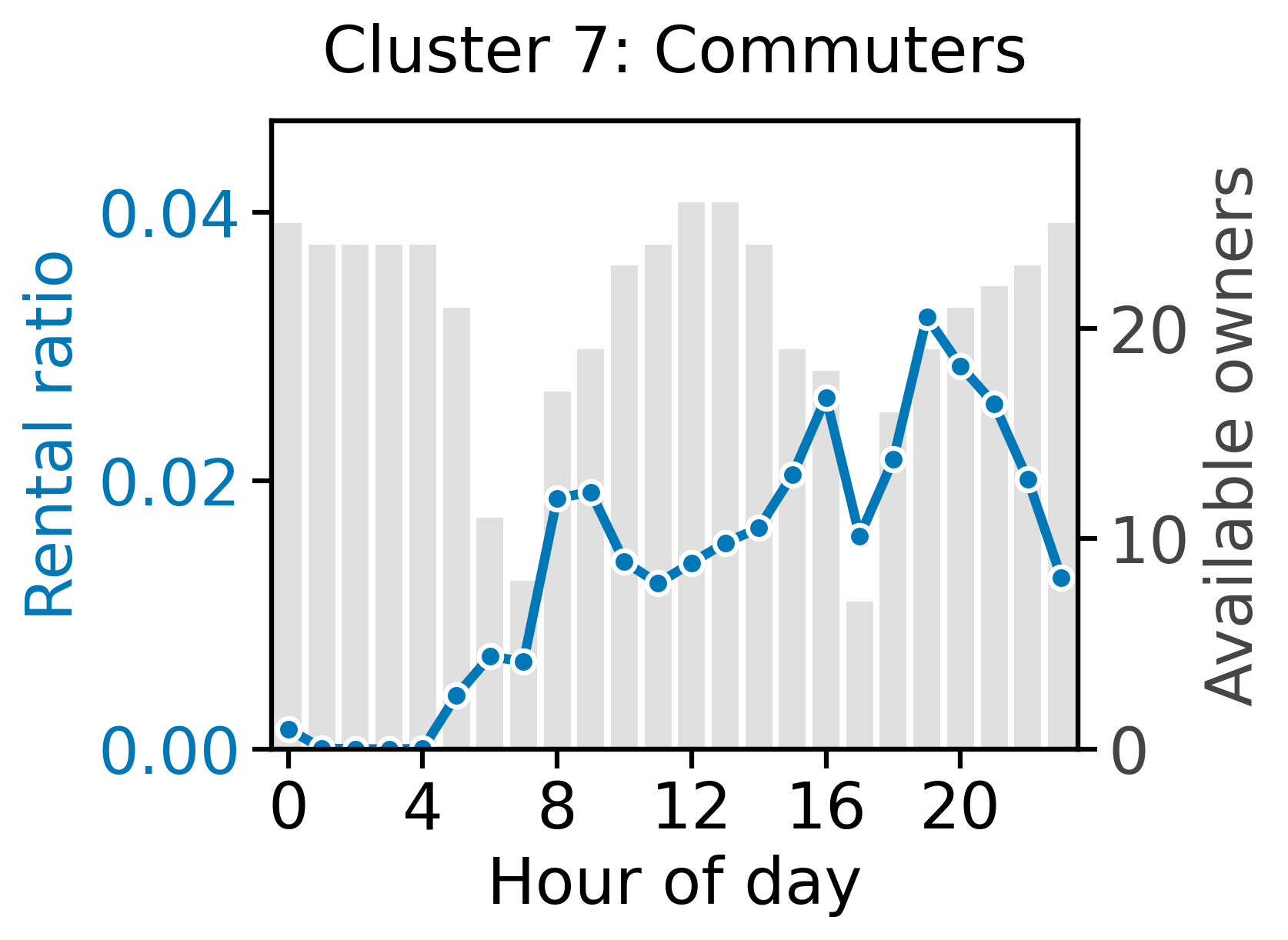}
    \end{subfigure}
    \hfill
    \begin{subfigure}{0.24\textwidth}
        \includegraphics[width=\textwidth]{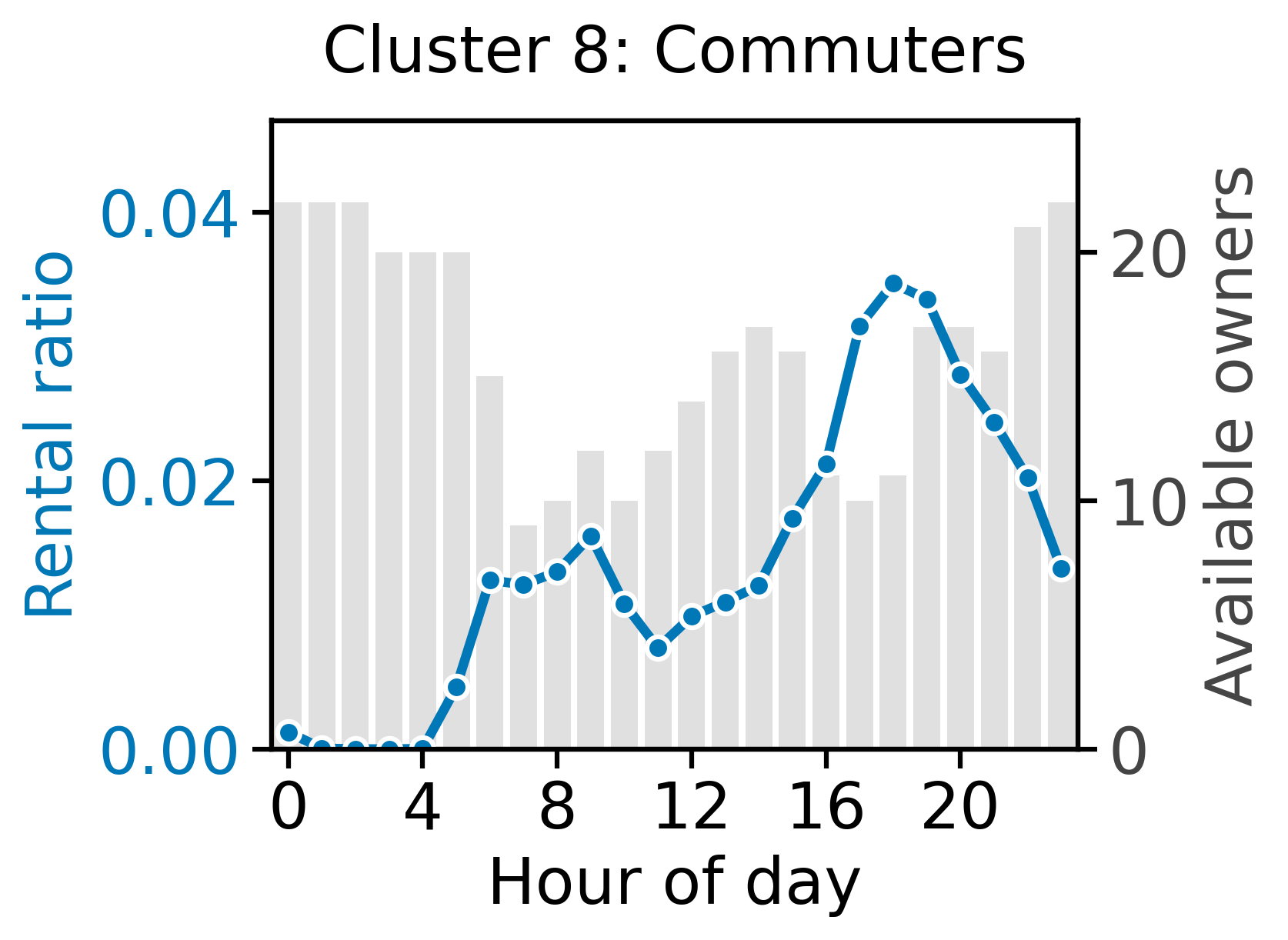}
    \end{subfigure}
    \caption{Average rental pricing for eight clusters of owners}
    \label{fig:eight_plots}
\end{figure}
In Figure~\ref{fig:eight_plots}, the gray histograms in the background show the number of available owners in each cluster across different hours of the day. The results indicate that workers exhibit clear infeasibility around $7{:}00$--$8{:}00$ and $17{:}00$--$18{:}00$. During these supply-shortage periods, the rental ratio remains high, which can be attributed to strong demand that incentivizes the platform to adopt crowdsourced AVs. Even when supply is relatively low, the rental ratio stays elevated due to higher rental payments driven by demand pressure. The blue curve in Figure~\ref{fig:eight_plots} illustrates the average participation ratio for each cluster of vehicle owners. Overall, participation is higher during the evening peak than the morning peak in most clusters, driven by stronger evening demand in Figure~\ref{fig:price_chicago}. Among the eight clusters, the home-oriented group contributes significantly to the AV crowdsourcing service. This is mainly due to two factors: it accounts for 34\% of the total AV owners, and it maintains high availability across most time periods. As a result, during the morning peak, this group exhibits the highest utilization, highlighting its key role in mitigating morning supply–demand imbalance. It also contributes substantially during the evening peak due to the larger overall pool of potential supply. Among the eight clusters of owners, commuters account for 55\% of the total, and they also contribute significantly through active participation, particularly those owners who reside within the city.
\begin{figure}[H]
    \centering
    \includegraphics[width=1\linewidth]{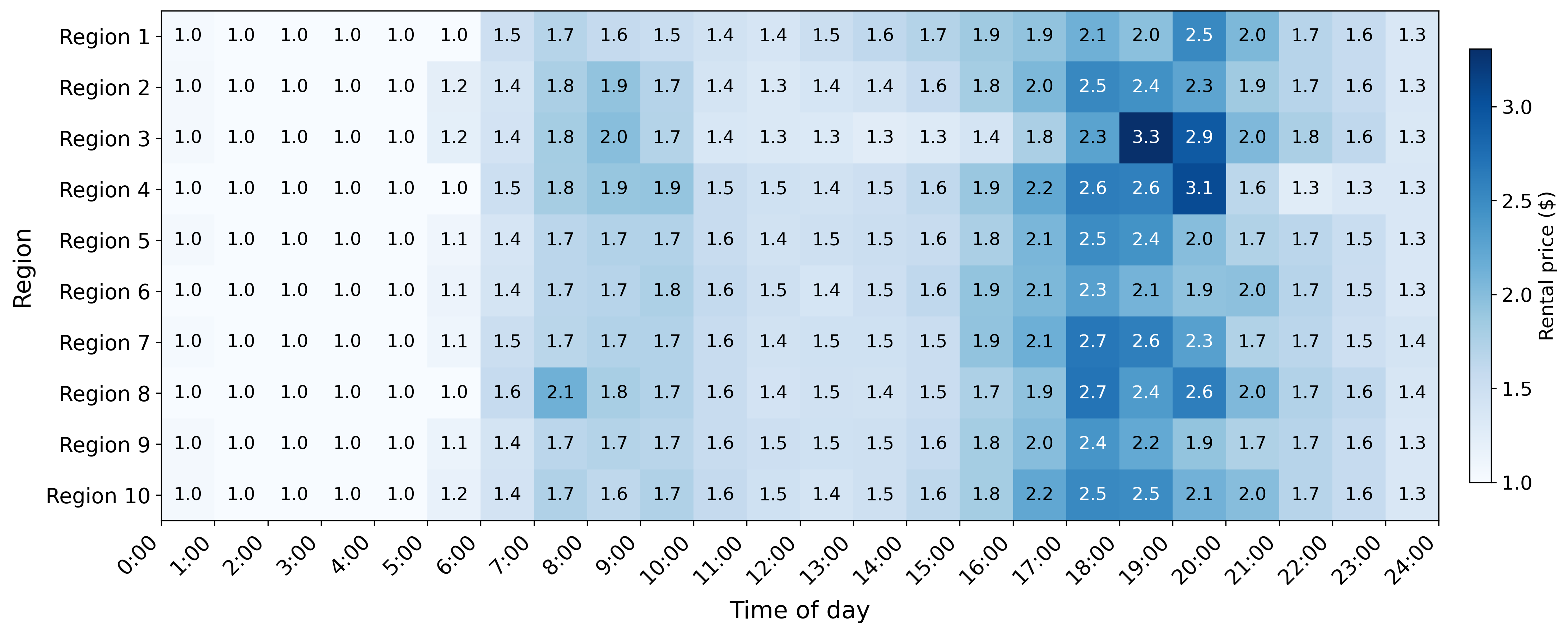}
    \caption{Spatial and temporal heatmap of rental prices}
    \label{fig:heatmap_price}
\end{figure}
Figure~\ref{fig:heatmap_price} illustrates the platform’s rental pricing strategies across different regions and time periods. It is consistent with Figure~\ref{fig:price_chicago}, showing that evening-peak rental prices are higher than morning-peak prices, with greater variability. In addition, we can notice more detailed information from this heatmap. For instance, although rental prices in the commercial region (Region 3) are high during 18:00-19:00, it is one dollar lower in the preceding period (17:00–18:00). This suggests that strategic owners interested in maximizing rental revenue should carefully consider both spatial and temporal factors when deciding where and when to rent. 
\begin{figure}[h]
    \centering
        \begin{subfigure}[b]{0.48\linewidth}
        \centering
        \includegraphics[width=\linewidth]{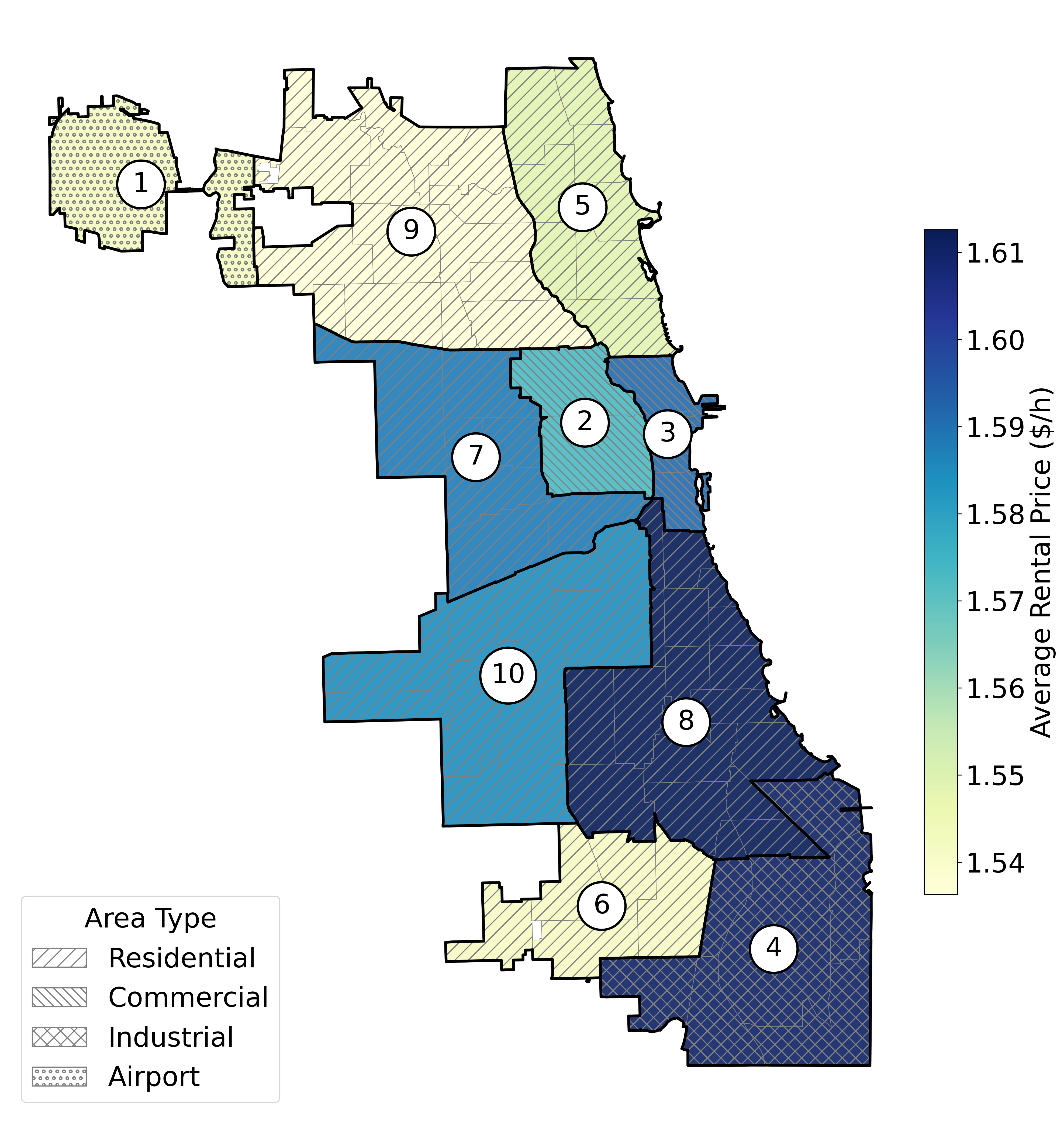}
        \caption{Average rental price across areas}
    \end{subfigure}
    \begin{subfigure}[b]{0.48\linewidth}
        \centering
        \includegraphics[width=\linewidth]{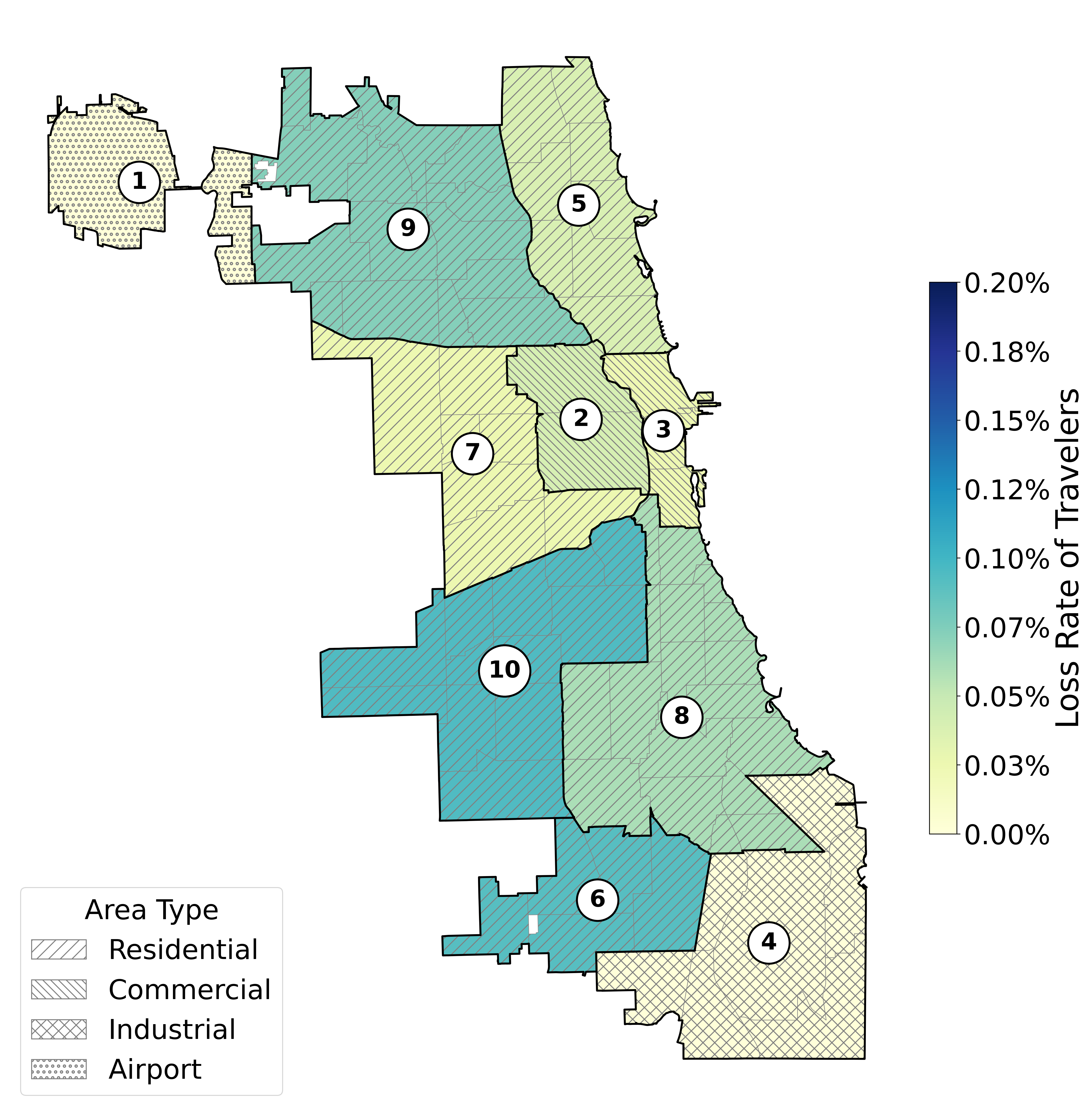}
        \caption{Percentage of lost travelers}
    \end{subfigure}
    \caption{Case study results. (a) Rental prices, (b) Lost customers.}
    \label{fig:chicago_result}
\end{figure}
Figure~\ref{fig:chicago_result}(a) shows the spatial distribution of average rental prices across the study area. Comparing Figure~\ref{fig:chicago_comparison}(b) and Figure~\ref{fig:chicago_result}(a), one can easily recognize that high rental prices are mostly associated with sparse private AV availability. 

Figures~\ref{fig:chicago_result} (b) presents the spatial distribution of lost demand. It is shown that almost all passengers are served thanks to the well-coordinated operations of crowdsourced and fleet-owned AVs, which is further demonstrated in Figure~\ref{fig:chicago_result2}.
As shown in Figure~\ref{fig:chicago_result2}, the crowdsourced and platform-owned AVs have fairly different operational patterns. 
While the crowdsourced AVs are highly concentrated in central commercial areas, the platform-owned AVs are more spatially dispersed, covering both commercial and surrounding residential areas. This may be attributed to the fact that many owners concentrate in commercial areas for work, so that limited available time windows and return-location constraints result in higher utilization of crowdsourced AVs in these areas. In contrast, platform-owned vehicles are not subject to such constraints and are therefore more evenly distributed across the city.

\begin{figure}[!h]
    \centering
        \begin{subfigure}[b]{0.48\linewidth}
        \centering
        \includegraphics[width=\linewidth]{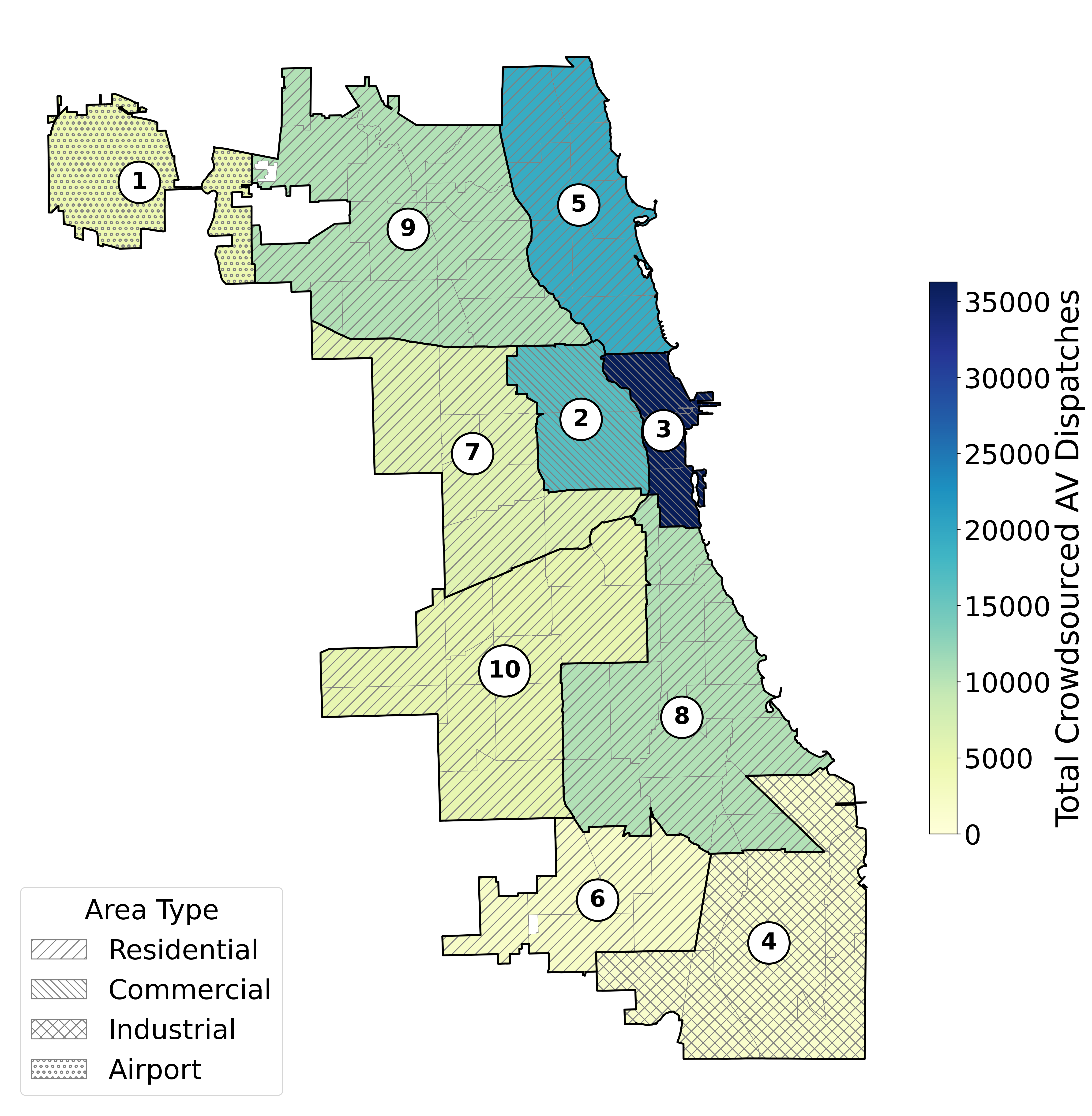}
        \caption{Total crowdsourced AV dispatches}
    \end{subfigure}
    \begin{subfigure}[b]{0.48\linewidth}
        \centering
        \includegraphics[width=\linewidth]{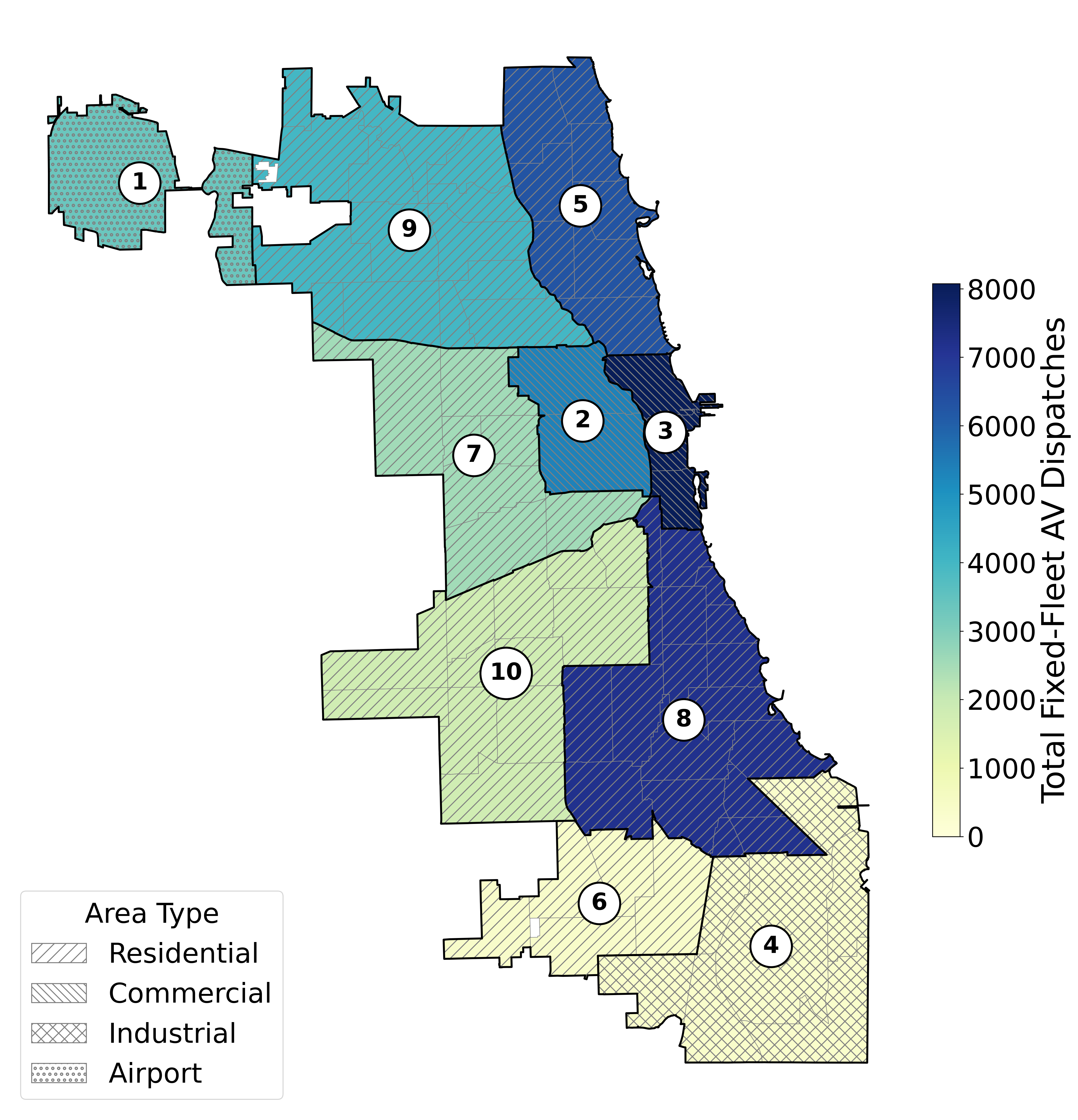}
        \caption{Total fixed-fleet AV dispatches}
    \end{subfigure}
    \caption{Case study results. (a) Crowdsourced AV dispatches, (b) Fixed fleet dispatches.}
    \label{fig:chicago_result2}
\end{figure}

\begin{figure}[!h]
    \centering
    \begin{subfigure}[b]{0.48\linewidth}
        \centering
        \includegraphics[width=\linewidth]{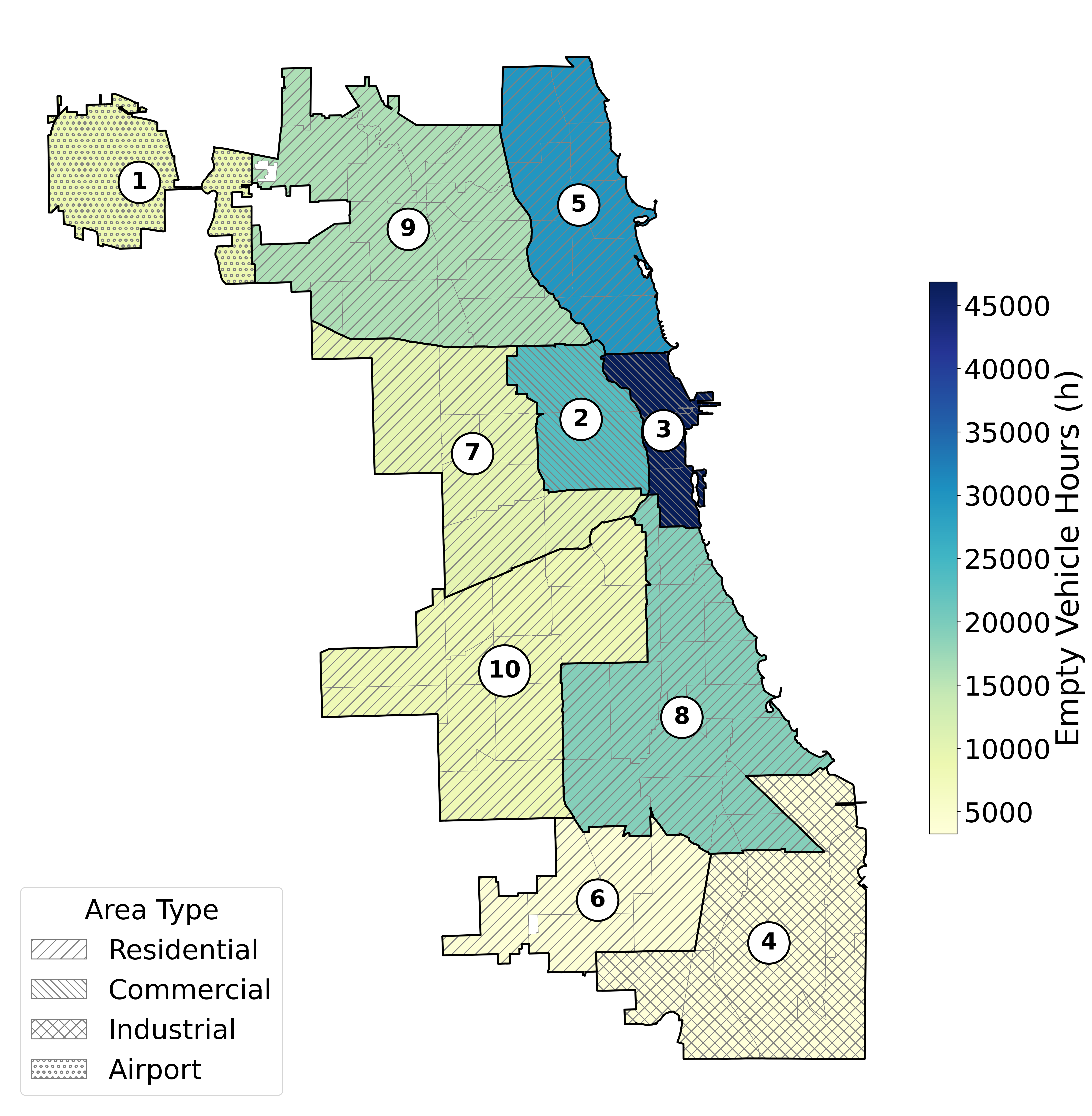}
        \caption{AV empty hours}
    \end{subfigure}
    \begin{subfigure}[b]{0.48\linewidth}
        \centering
        \includegraphics[width=\linewidth]{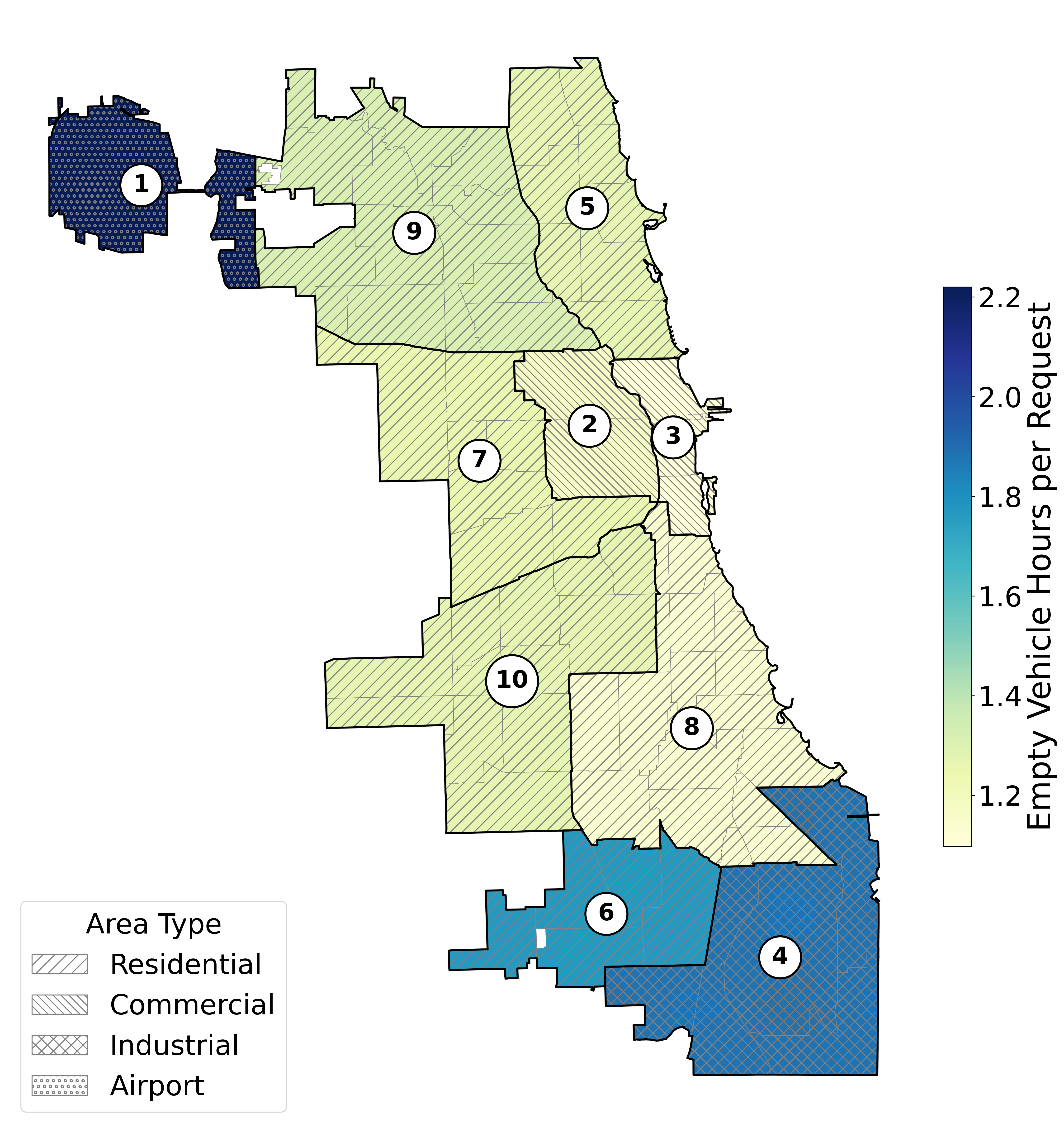}
        \caption{AV empty hours per request}
    \end{subfigure}
    \caption{Case study results. (a) AV empty hours, (b) AV empty hours per request.}
    \label{fig:chicago_result_empty_total}
\end{figure}

Figure~\ref{fig:chicago_result_empty_total}(a) illustrates the spatial distribution of empty vehicle hours. Regions with the highest demand, such as the commercial area, experience the most frequent dispatches and the longest empty vehicle hours. However, this does not necessarily imply higher service quality or shorter waiting times in those areas, as high demand in these regions can lead to intensified competition for available vehicles. 
Figure~\ref{fig:chicago_result_empty_total}(b) shows that the average amortized empty vehicle hours per request exceed one hour across all regions, suggesting that there are adequate vehicle supply for on-demand customers. In the airport area, the southern industrial area, and the southern residential area, the empty vehicle hours for each request are relatively high, reaching approximately two hours. 
Those peripheral areas face lower demand, enabling customers to secure vehicles more easily. In addition, platform-level incentives also play an important role. For example, Figure~\ref{fig:fare} shows that trip fares from the airport are substantially higher than those in other areas, incentivizing the platform to maintain service coverage even if vehicles remain idle for long periods; otherwise, substantial penalty costs would be incurred.

\subsection{Sensitivity analysis}
We next conduct a sensitivity analysis on the supply side. First, we examine how the number of potential AV owners and the fixed fleet size affect market outcomes. We then analyze how the proportion of non-commuters among AV owners influences these outcomes.

\subsubsection{Number of potential AV owners}
\begin{figure}[h]
    \centering
    \begin{subfigure}{0.48\textwidth}
        \centering
        \includegraphics[width=\linewidth]{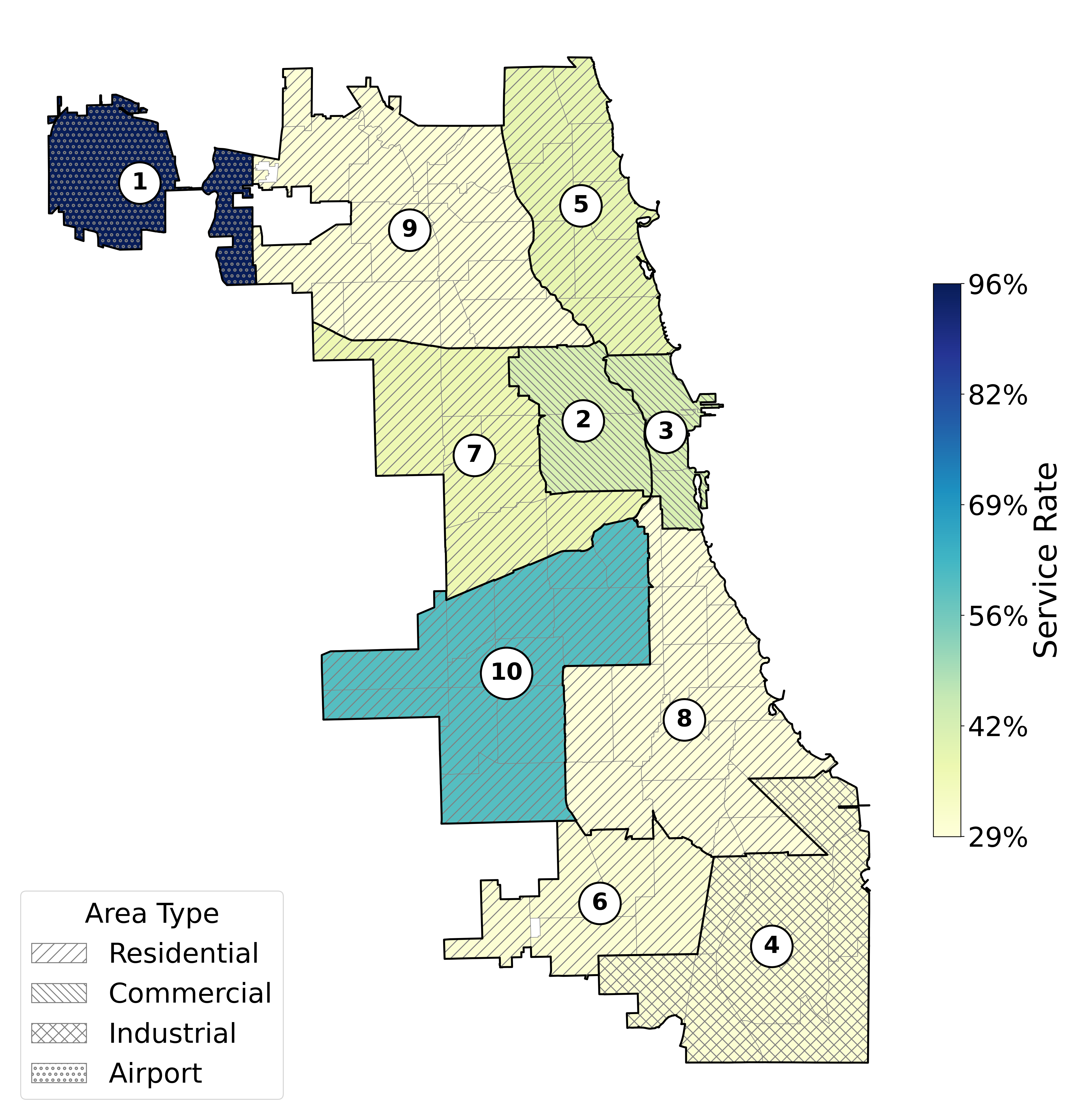}
        \caption{Service rate with $1\%$ private AV penetration}
        \label{fig:vehicles}
    \end{subfigure}
    \hfill
    \begin{subfigure}{0.48\textwidth}
        \centering
        \includegraphics[width=\linewidth]{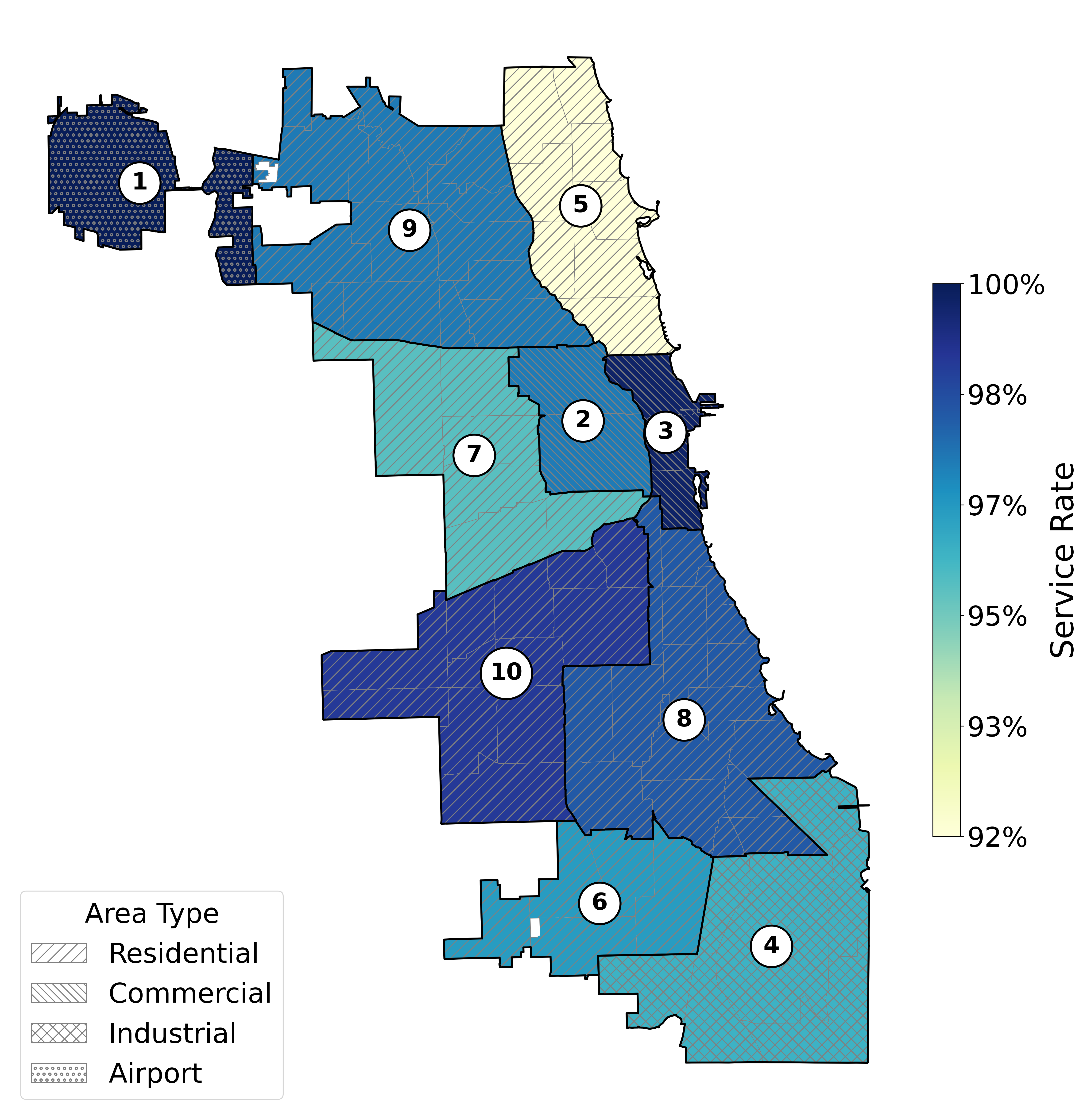}
        \caption{Service rate with $10\%$ private AV penetration}
        \label{fig:price}
    \end{subfigure}
    \caption{Impact of private AV penetration rate on service rate}
    \label{fig:service_rate_penetration}
\end{figure}
Figures~\ref{fig:service_rate_penetration}(a) and (b) show that 1\% and 10\% of the benchmark number of AV owners (as defined in subsection \ref{sec_cluster_owner}) significantly influence passenger service rates. As expected, a lower number of vehicles leads to a decline in service rates. Interestingly, when the vehicle fleet is very small (1\%), the highest passenger service rates occur at the airport and in Zone 10 (a residential area), whereas commercial areas do not exhibit high service rates. This may be attributed to the high fares generated by long-distance trips, which provide incentives for serving these areas. Given the very limited vehicle supply, trips in the city center characterized by short travel distances, tend to be given lower priority. When the available AVs increase to 10\%, the service rates at the airport, Zone 10, and the commercial area all become high. Compared to Figure \ref{fig:chicago_comparison}(b), the lowest service rate is observed in the area with the highest concentration of AV owners (Zone 5), though the difference is relatively small, only 8\%. This may be caused by the temporal overlap between when AV owners use their vehicles personally and when passengers need rides, which leads to a supply–demand imbalance. Therefore, even though Zone 5 has the largest number of AV owners, the overlapping travel times prevent them from fully serving all passengers.

\subsubsection{Number of fixed-fleet AVs}
Figure~\ref{fig:price_fixe_fleet}(a) and (b) show that as the fixed fleet size increases, rental prices decrease substantially, indicating a reduction in the participation of crowdsourced AVs. When the fleet reaches 500 vehicles, crowdsourced AVs are used only during peak hours and not during off-peak periods. As the number of fixed fleet AVs exceeds 900, the platform relies almost entirely on the fixed fleet but adopts few crowdsourced AVs. These results show that the platform-owned fleet size is important in determining the feasibility of AV crowdsourcing services. 
\begin{figure}[h]
    \centering
    \begin{subfigure}{0.46\textwidth}
        \centering
        \includegraphics[width=\linewidth]{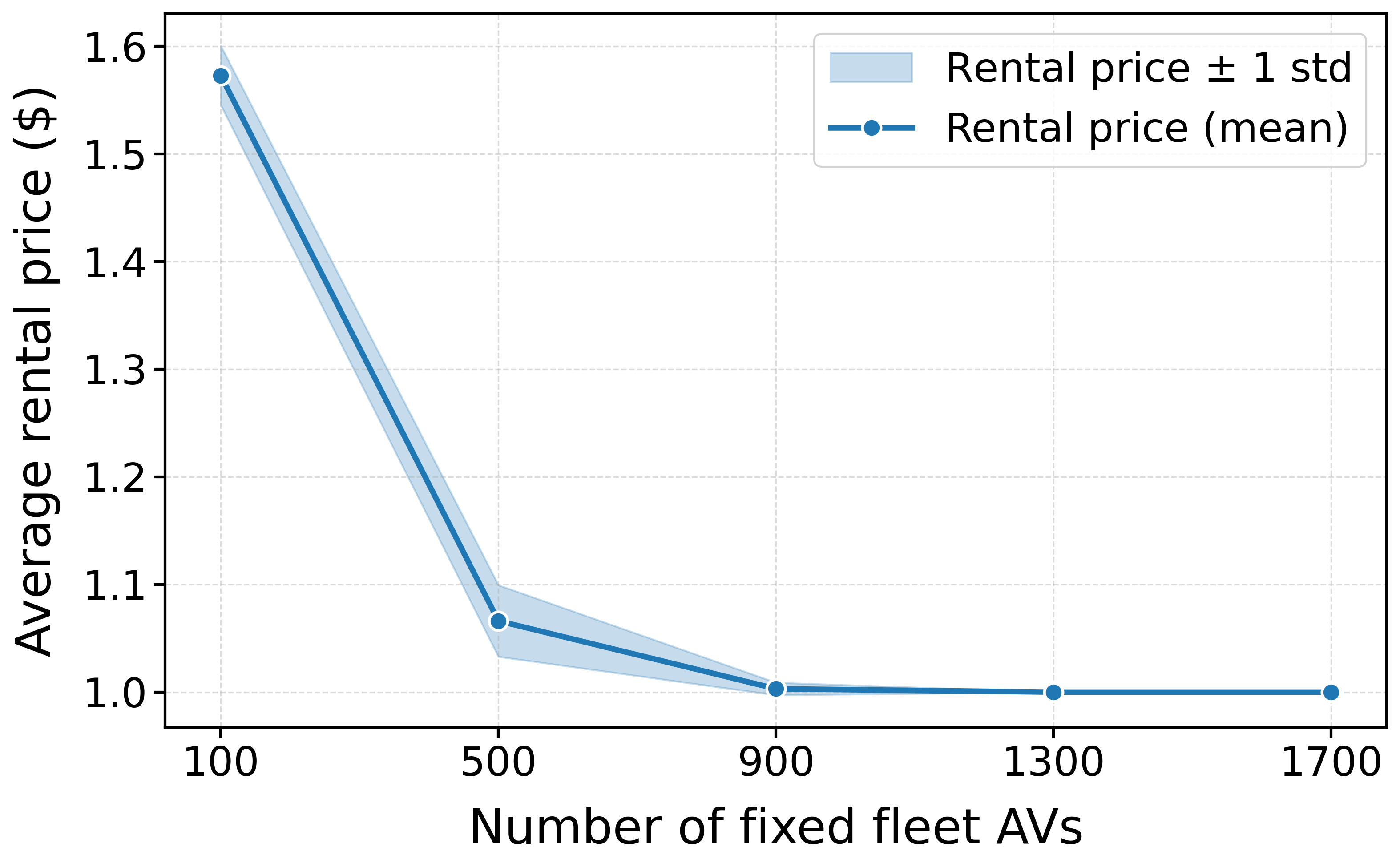}
        \caption{Average rental prices}
        \label{fig:fixed_scale_price}
    \end{subfigure}
    \hfill
    \begin{subfigure}{0.5\textwidth}
        \centering
        \includegraphics[width=\linewidth]{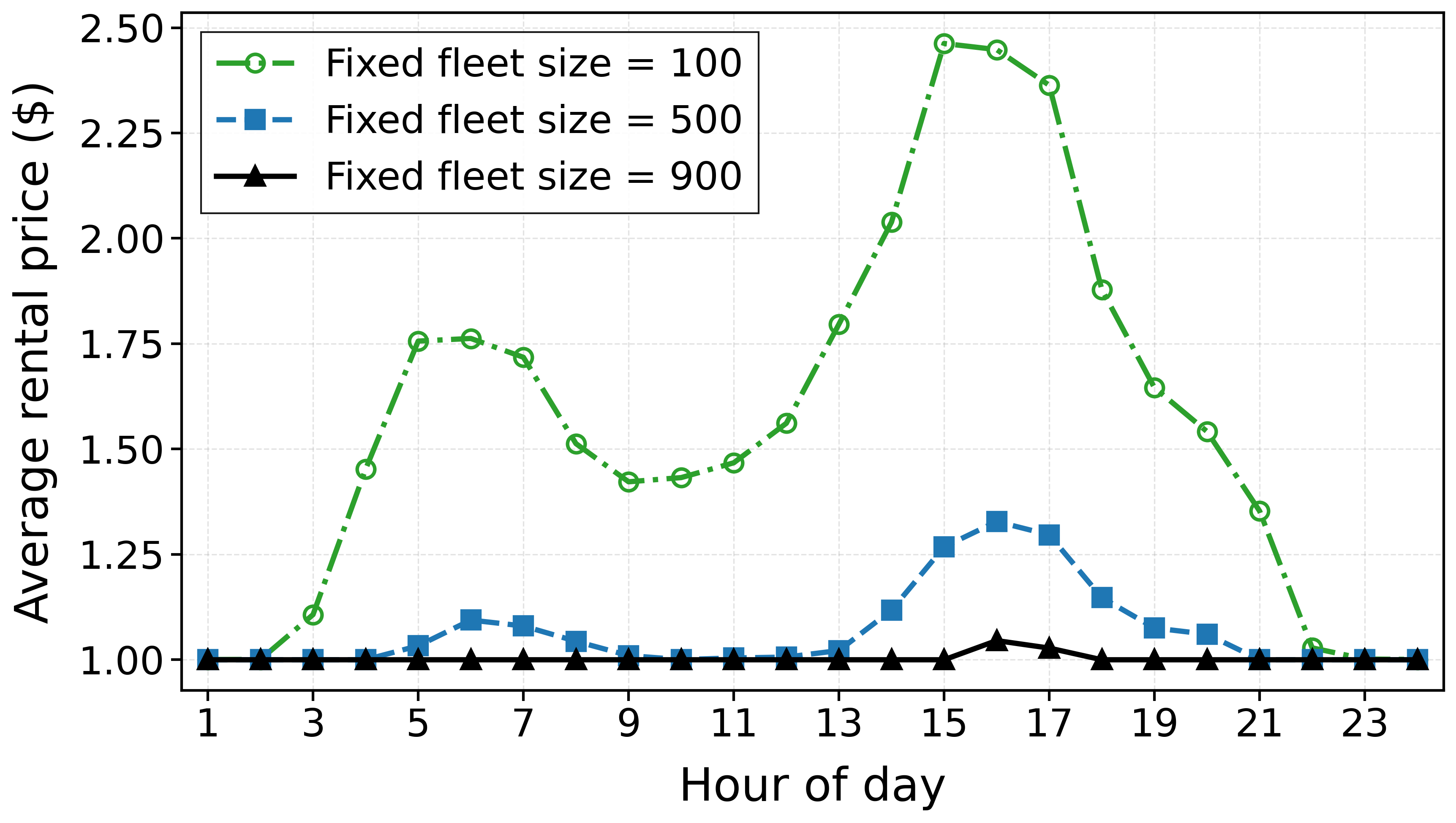}
        \caption{Rental prices along the time}
        \label{fig:fixed_scale_price}
    \end{subfigure}
    \caption{Impact of fixed fleet size on rental prices}
    \label{fig:price_fixe_fleet}
\end{figure}
\begin{figure}[h]
    \centering
    \begin{subfigure}{0.5\textwidth}
        \centering
        \includegraphics[width=\linewidth]{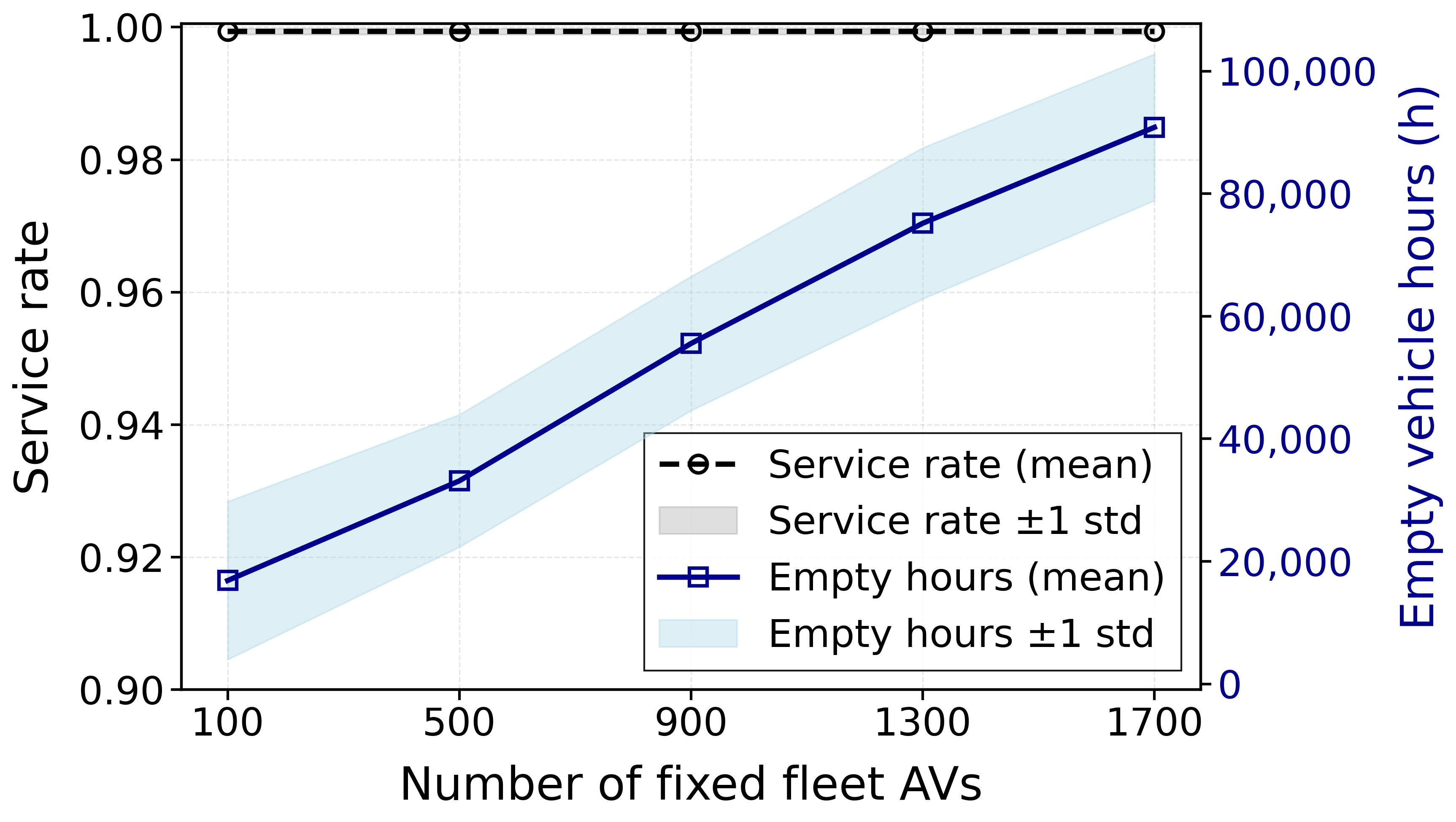}
        \caption{Empty vehicle hours}
    \end{subfigure}
    \hspace{0.2cm}
    \begin{subfigure}{0.47\textwidth}
        \centering
        \includegraphics[width=\linewidth]{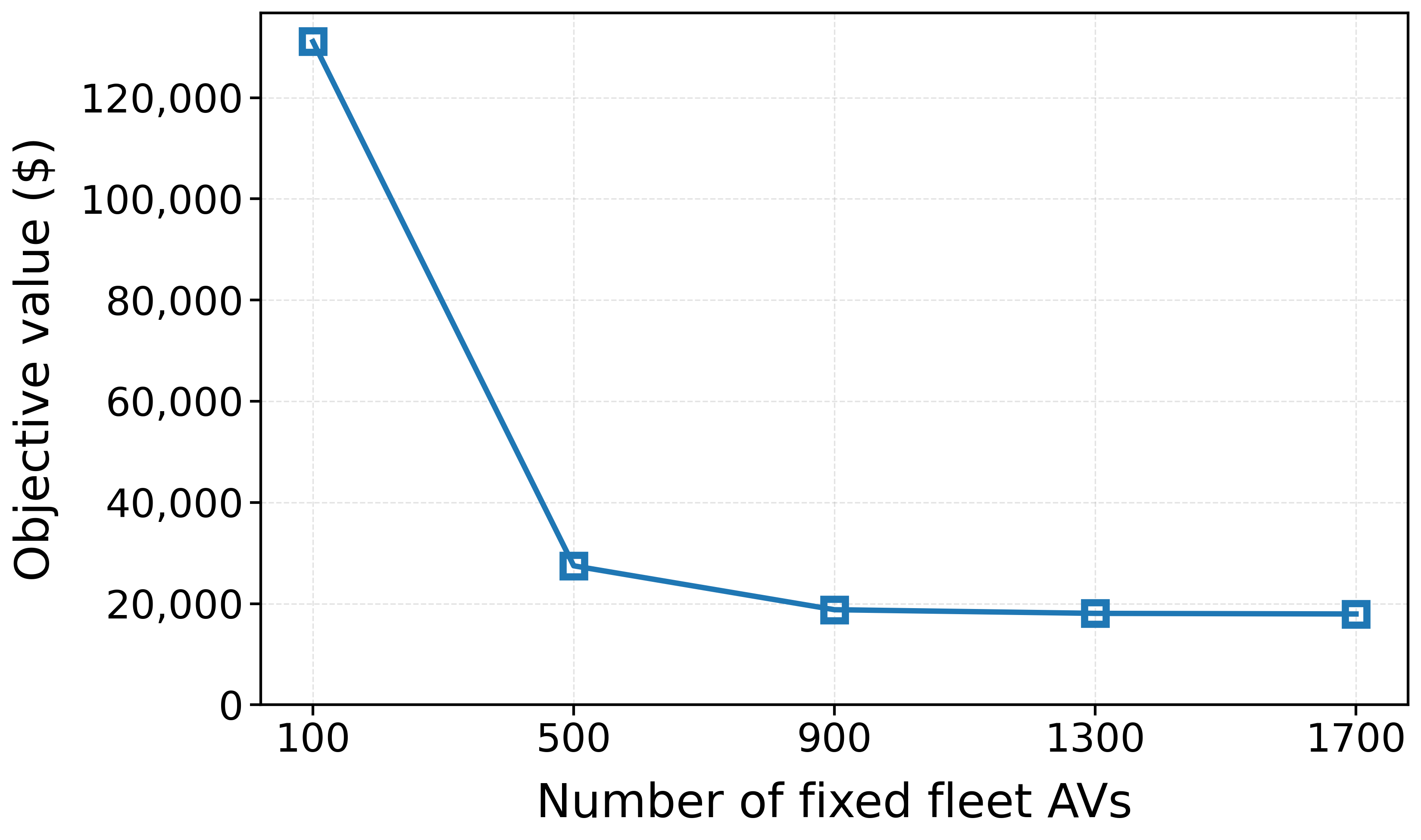}
        \caption{Platform cost (\$)}
    \end{subfigure}
    \caption{Impact of fixed fleet size on empty vehicle hours and platform cost}
    \label{fig:outcome_fixed_fleet}
\end{figure}

However, Figure \ref{fig:outcome_fixed_fleet} (a) reveals that a larger number of platform-owned AVs does not necessarily lead to better market outcomes. As the size of the fixed fleet increases, the empty vehicle hours rise significantly, while the customer service rate remains close to one. These results suggest that when the pool of potential AV owners is sufficiently large, expanding the platform-owned fleet increases total empty vehicle hours significantly but yields limited improvement in service rate. Although Figure \ref{fig:outcome_fixed_fleet}(b) shows that increasing the fixed fleet reduces the platform’s operational cost, we do not account for the capital cost of purchasing AVs. Maintaining a larger fleet may no longer be economically favorable for the platform once the acquisition cost of AVs is considered. 

\subsubsection{Proportion of non-commuting AV owners}
Motivated by the increasing prevalence of telecommuting and flexible work arrangements in the post-pandemic era, we conduct a sensitivity analysis on the ratio of non-commuters among AV owners. As suggested by the clustering result in subsection \ref{sec_cluster_owner}, commuters and non-commuters exhibit distinct spatio-temporal patterns in their daily activities. Consistent with the scale of our previous case study (699 owners), we simulate a population of 700 individuals by sampling from two distinct pools: commuters and non-commuters. 
Let the parameter $\sigma \in [0,1]$ control the proportion of commuters in the simulated population. Then, a fraction $\sigma$ of the 700 individuals is sampled from the commuter pool, while the remaining fraction $(1-\sigma)$ is sampled from the non-commuter pool.

\begin{figure}[H]
    \centering
    \begin{subfigure}{0.45\textwidth}
        \centering
        \includegraphics[width=\linewidth]{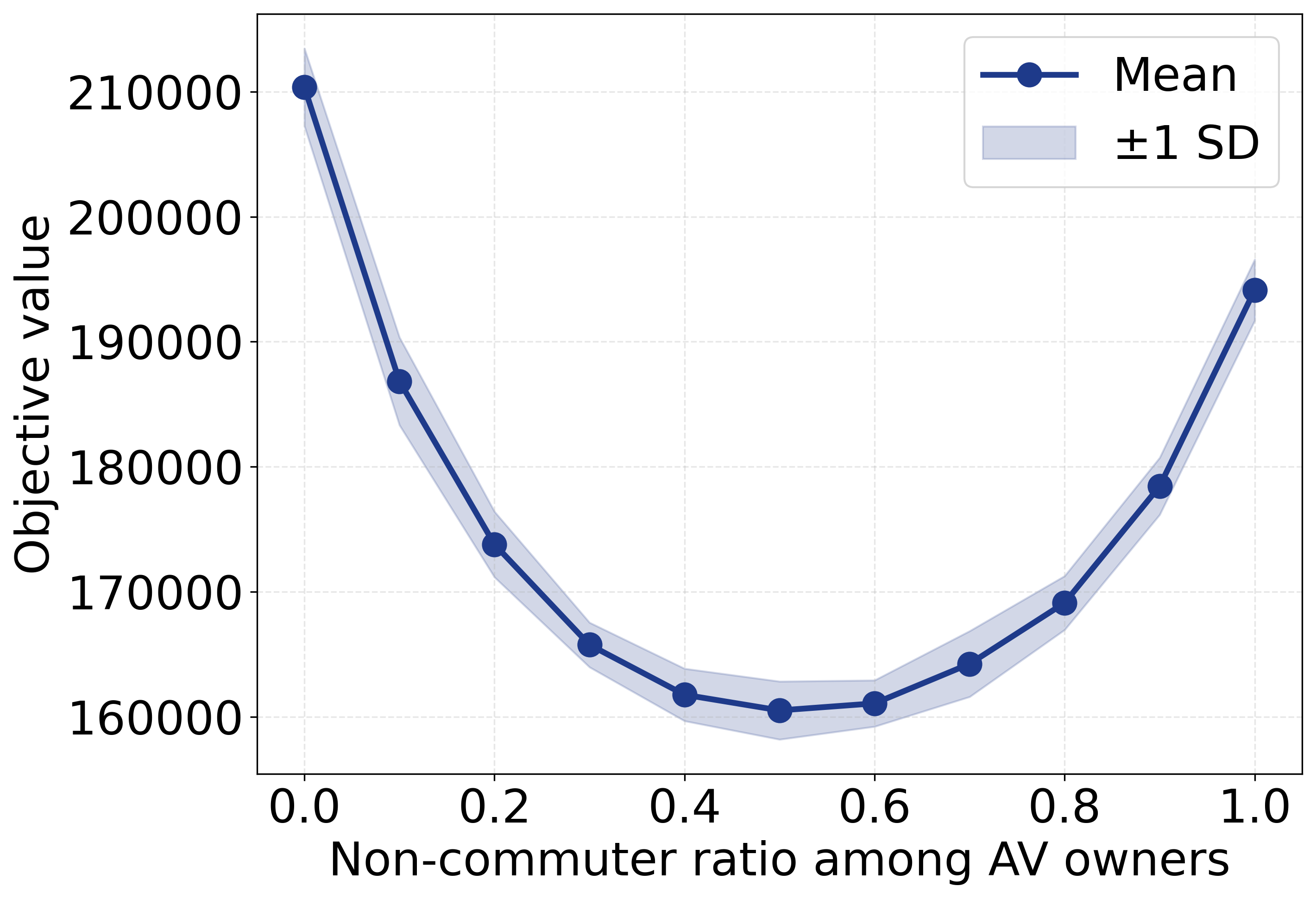}
        \caption{Platform cost (\$)}
        \label{fig:objective}
    \end{subfigure}
    \hspace{0.2cm}
    \begin{subfigure}{0.45\textwidth}
        \centering
        \includegraphics[width=\linewidth]{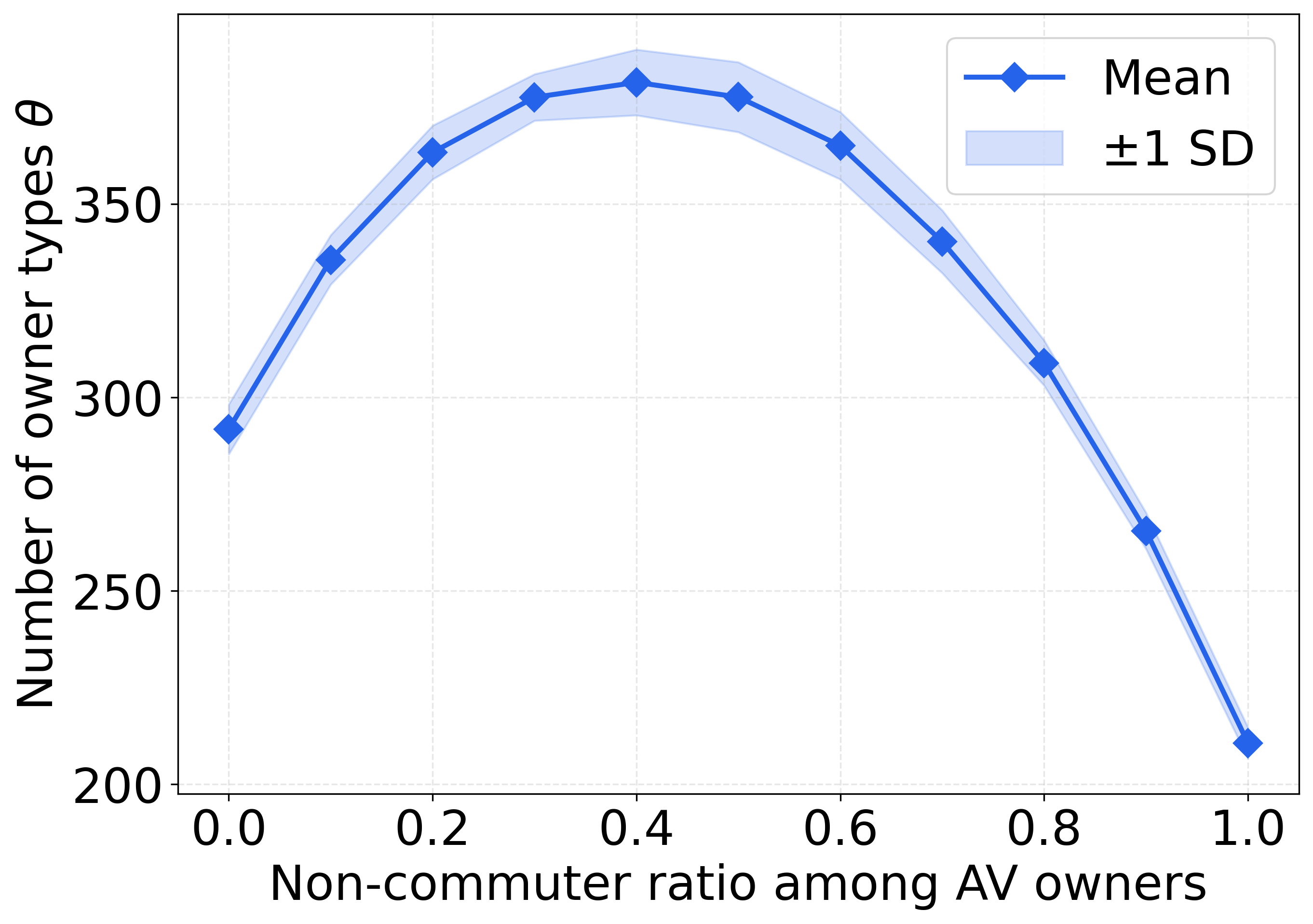}
        \caption{Number of spatiotemporal availability profiles $\theta$}
        \label{fig:owner_type}
    \end{subfigure}
    \begin{subfigure}{0.45\textwidth}
        \centering
        \includegraphics[width=\linewidth]{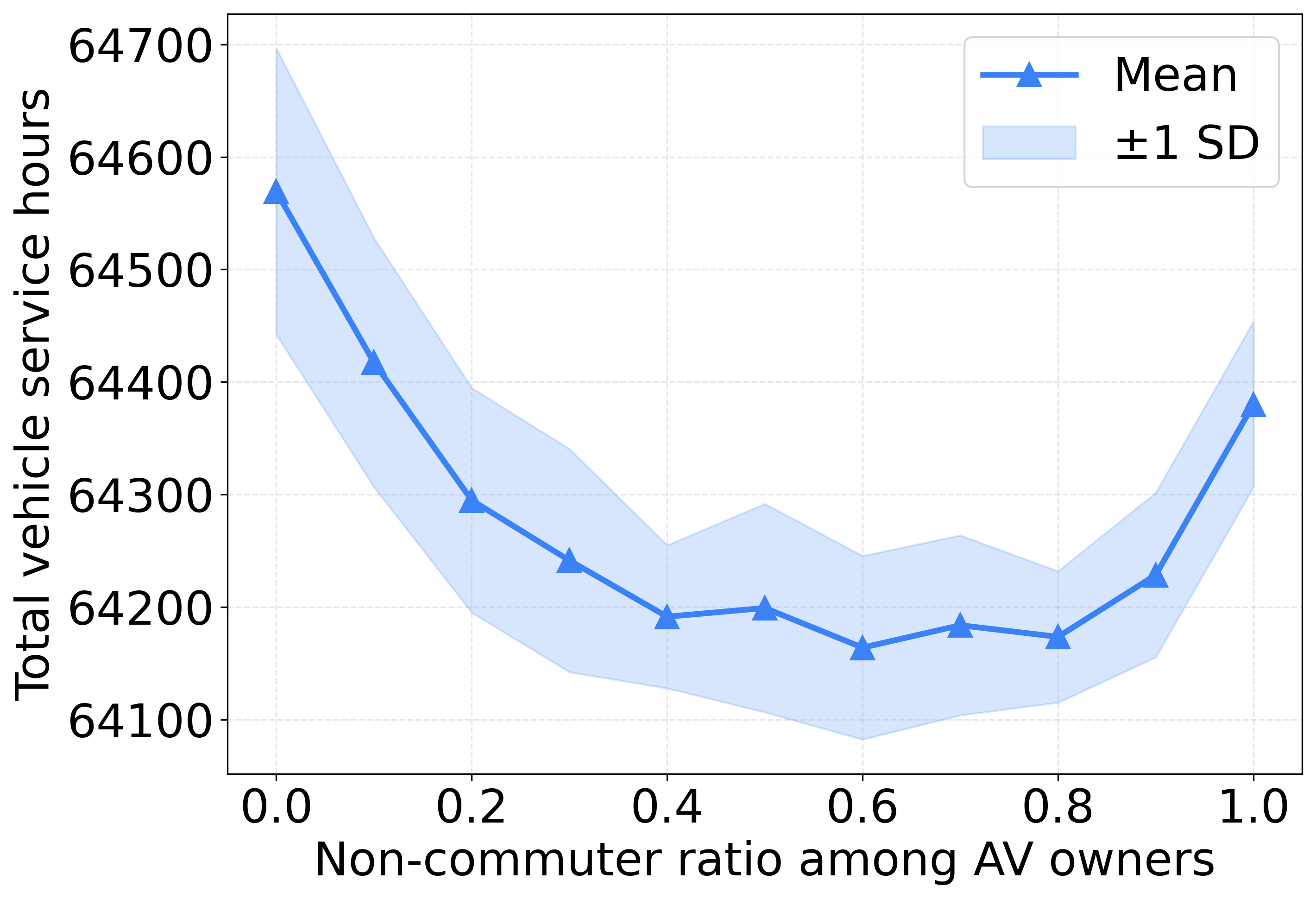}
        \caption{Total vehicle serving hours (h)}
        \label{fig:vehicles}
    \end{subfigure}
    \hspace{0.2cm}
    \begin{subfigure}{0.45\textwidth}
        \centering
        \includegraphics[width=\linewidth]{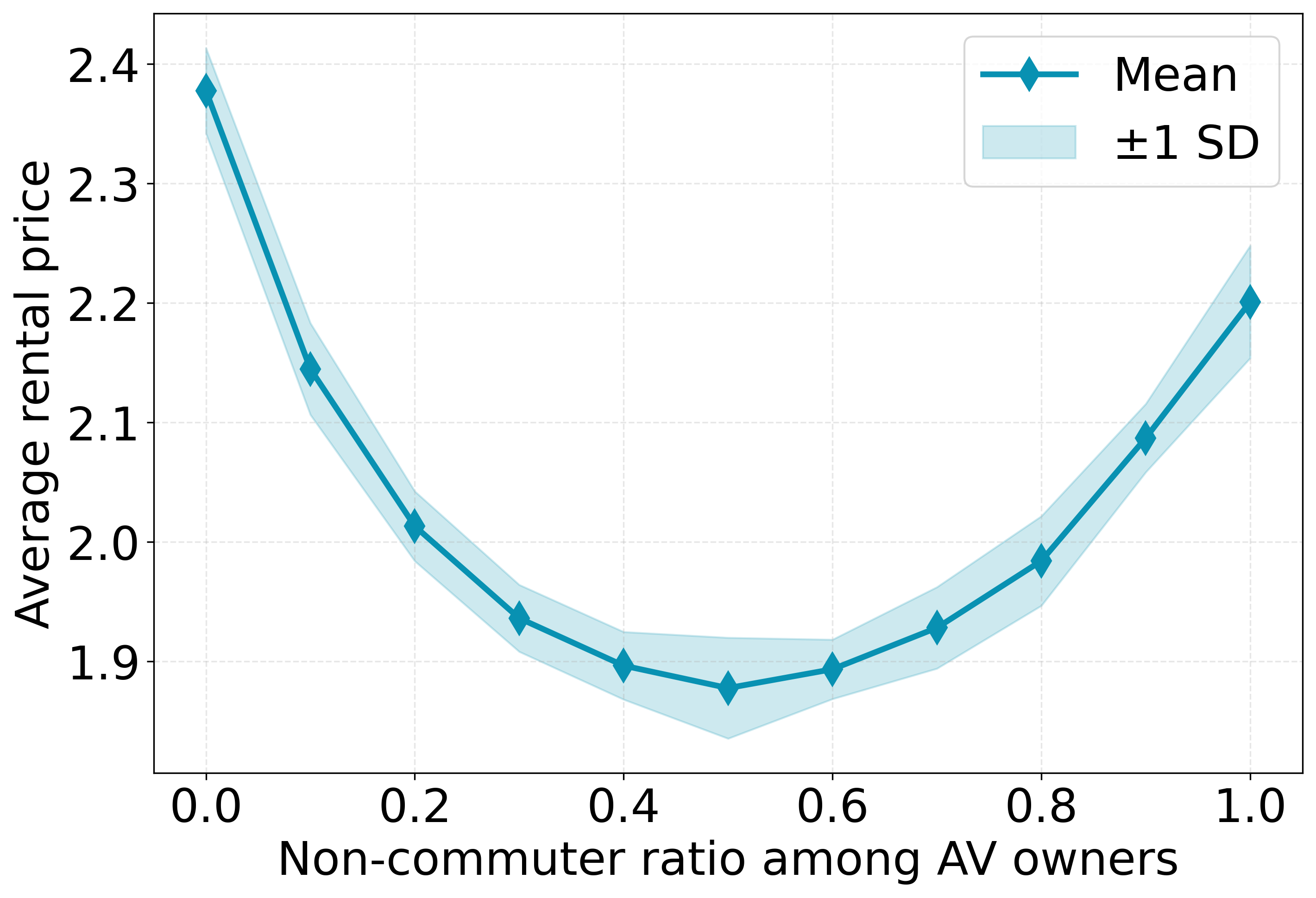}
        \caption{Average rental price (\$)}
        \label{fig:price}
    \end{subfigure}
    \caption{Impact of the non-commuting rate among AV owners}
    \label{fig:sensitivity_analysis_commuters}
\end{figure}

According to Table \ref{tab:classification}, although both commuters and non-commuters often need to use their own vehicles during peak hours, their rental locations differ significantly. By comparing the results in Figures~\ref{fig:sensitivity_analysis_commuters}(c) and (d) with Lemma \ref{lemma_cross} and Proposition \ref{prop_both_zone}, we find that the numerical results validate the theoretical analysis. Specifically, renting vehicles from a more diverse set of regions reduces the required service duration and lowers prices, ultimately leading to lower platform costs. This is illustrated in Figure~\ref{fig:sensitivity_analysis_commuters}(a), where an approximately 50\% mix of commuters and non-commuters among AV owners yields the lowest platform cost.

In addition, we plot the number of vehicle classes $\theta$, defined by spatiotemporal availability. We observe that when commuters and non-commuters are mixed, the diversity of vehicle types increases, indicating a richer spatiotemporal distribution. Comparing Figures~\ref{fig:sensitivity_analysis_commuters}(a) and (b), we find that a more diverse spatiotemporal distribution of private vehicles is associated with a reduction in platform costs. This suggests that the service operator may benefit from selecting drivers with more diverse characteristics in real practice.

\section{Conclusion} \label{sec_concl}

\subsection{Discussion}
The analysis and key components examined in Sections \ref{sec_property} and \ref{sec_num} can also be applied to other sharing economy settings, as long as the system is characterized by the misaligned usage of the same asset by owners and customers. Some sharing economy applications are not spatially relevant; for example, certain digital tasks can be performed remotely (e.g., computational power sharing). In such cases, these settings can be viewed as simplified versions of our problem.

Looking ahead, AI-enabled assets may become scarce yet highly valuable assets, characterized by high acquisition costs but substantial utility. The sharing of robots and other autonomous systems may therefore emerge as an important business model. In this sense, autonomous vehicles can be viewed as a specific class of self-driving robots, and our work provides a forward-looking perspective. The potential sharing of robotic systems would also involve operational constraints similar to those studied in this paper.

\subsection{Conclusion}
This paper investigates the operational feasibility of AV crowdsourcing services under heterogeneous travel patterns. Given private AV owners’ spatiotemporal availability throughout the day, we develop a time-expanded network model to jointly optimize AV rental pricing and vehicle operations. Assuming a uniform distribution of AV owners’ opportunity costs, the resulting capacity planning problem can be formulated as a convex quadratic program, enabling analytical tractability and scalable computation. Combining theoretical analysis, numerical experiments in a two-node network, and a Chicago case study, we summarize the main findings as follows: 
\begin{itemize}
    \item This research establishes the conditions for the feasibility of AV crowdsourcing services. The service rate increases as the travel time gap between AV owners and customers decreases, the reliability buffer for owners' private use is reduced, the customer loss penalty parameter increases, and inter-regional travel times decrease.
    \item 
    Platform will select AV owners with diverse activity patterns. In the case study of Chicago, results show that an approximately 50\% mix of commuting and non-commuting AV owners yields the lowest platform cost.
    \item Chicago data reveals that demand and owner availability exhibit significant spatial heterogeneity across regions. Although many owners tend to travel during demand peak hours, their geographic distributions differ substantially.
    \item With central dispatching, although empty vehicles tend to concentrate in high-demand areas, some peripheral areas may experience good service quality due to lower traveler competition. The airport region exhibits a markedly high service rate. 
\end{itemize}

Future research can extend this work in multiple directions. 
First and foremost, strong assumptions on the opportunity cost, reliability buffer, and other parameters shall be relaxed to further investigate the impact of user heterogeneity and generate additional insights into practice.  Moreover, the case study can be revisited once data on privately owned AVs become available. Last but not least, it has been assumed the travel pattern of AV owners is not affected by the AV crowdsourcing service, which may not be true in reality. Future research can explore how the vehicle sharing revenue could reshape the AV owners' activity patterns.

\newpage
\bibliographystyle{apalike}
\bibliography{references_management}

@article{WU201552,
title = {Electric vehicles’ energy consumption measurement and estimation},
journal = {Transportation Research Part D: Transport and Environment},
volume = {34},
pages = {52-67},
year = {2015},
issn = {1361-9209},
doi = {https://doi.org/10.1016/j.trd.2014.10.007},
url = {https://www.sciencedirect.com/science/article/pii/S1361920914001485},
author = {Xinkai Wu and David Freese and Alfredo Cabrera and William A. Kitch}
}

@article{WANG2021103362,
title = {Aggregate modeling and equilibrium analysis of the crowdsourcing market for autonomous vehicles},
journal = {Transportation Research Part C: Emerging Technologies},
volume = {132},
pages = {103362},
year = {2021},
issn = {0968-090X},
doi = {https://doi.org/10.1016/j.trc.2021.103362},
url = {https://www.sciencedirect.com/science/article/pii/S0968090X21003648},
author = {Xiaoyan Wang and Xi Lin and Meng Li}
}

@article{ZHANG2025103305,
title = {To park or to share your autonomous vehicle?},
journal = {Transportation Research Part B: Methodological},
volume = {200},
pages = {103305},
year = {2025},
issn = {0191-2615},
doi = {https://doi.org/10.1016/j.trb.2025.103305},
url = {https://www.sciencedirect.com/science/article/pii/S0191261525001547},
author = {Zhuoye Zhang and Fangni Zhang and Wei Liu}
}

@article{WANG2024104732,
title = {Competition of ride-hailing platforms in the era of autonomous vehicles: Heavy or light asset?},
journal = {Transportation Research Part C: Emerging Technologies},
volume = {165},
pages = {104732},
year = {2024},
issn = {0968-090X},
doi = {https://doi.org/10.1016/j.trc.2024.104732},
url = {https://www.sciencedirect.com/science/article/pii/S0968090X24002535},
author = {Xiaoyan Wang and Xi Lin and Meng Li and Zhengtian Xu and Ke Zhang}
}

@article{XU2025101533,
title = {Understanding the spatiotemporal dynamics of ride-hailing services: A study of demand and supply patterns using a large-scale driver activity dataset},
journal = {Case Studies on Transport Policy},
volume = {21},
pages = {101533},
year = {2025},
issn = {2213-624X},
doi = {https://doi.org/10.1016/j.cstp.2025.101533},
url = {https://www.sciencedirect.com/science/article/pii/S2213624X25001701},
author = {Shuoyan Xu and Nael Alsaleh and Timur Hamzaev and Eric J. Miller}
}

@article{LUO2024103936,
title = {Exploring competitiveness of taxis to ride-hailing services from a multidimensional spatio-temporal perspective: A case study in Beijing, China},
journal = {Journal of Transport Geography},
volume = {118},
pages = {103936},
year = {2024},
issn = {0966-6923},
doi = {https://doi.org/10.1016/j.jtrangeo.2024.103936},
url = {https://www.sciencedirect.com/science/article/pii/S0966692324001455},
author = {Yihao Luo and Ailing Huang and Zhengbing He and Jiaqi Zeng and Dianhai Wang}
}

@misc{census2020,
  author       = {U.S. Census Bureau},
  title        = {Data Tables for Counties: 2020-2023},
  year         = {2023},
  howpublished = {\url{https://www.census.gov/data/tables/time-series/demo/popest/2020s-counties-total.html}},
  note         = {Accessed: 2023-12-26}
}

@misc{ChicagoPlanningRegions,
  author       = {{City of Chicago}},
  title        = {Boundaries - Planning Regions (Map)},
  howpublished = {\url{https://data.cityofchicago.org/Community-Economic-Development/Boundaries-Planning-Regions-Map/2wek-zf5g}},
  year         = {2025},
  note         = {Accessed: 2026-01-09}
}

@misc{CMAP_TravelSurvey2019,
  author       = {{Chicago Metropolitan Agency for Planning (CMAP)}},
  title        = {2019 Travel Survey — Results and Analysis},
  howpublished = {\url{https://cmap.illinois.gov/data/transportation/travel-survey/#2019-survey-results-and-analysis}},
  year         = {2025},
  note         = {Accessed: 2026-01-09}
}

@article{WANG2023111crowddelivery,
title = {Joint optimization of parcel allocation and crowd routing for crowdsourced last-mile delivery},
journal = {Transportation Research Part B: Methodological},
volume = {171},
pages = {111-135},
year = {2023},
issn = {0191-2615},
doi = {https://doi.org/10.1016/j.trb.2023.03.007},
url = {https://www.sciencedirect.com/science/article/pii/S0191261523000504},
author = {Li Wang and Min Xu and Hu Qin}
}

@article{JIANG2026104505,
title = {An integrated optimization approach for e-order fulfillment using self-owned and crowdsourced delivery},
journal = {Transportation Research Part E: Logistics and Transportation Review},
volume = {205},
pages = {104505},
year = {2026},
issn = {1366-5545},
doi = {https://doi.org/10.1016/j.tre.2025.104505},
url = {https://www.sciencedirect.com/science/article/pii/S1366554525005332},
author = {Dapei Jiang and Xiangyong Li and Wei Yang and Yuxuan Zhao}
}

@article{hu2025multimodal,
  title={Multimodal large-language model empowering next-generation autonomous driving systems},
  author={Hu, Zhiqiang and Xu, Mingxing and Cheng, Qixiu},
  journal={Journal of Intelligent and Connected Vehicles},
  volume={8},
  number={2},
  pages={9210059--1},
  year={2025},
  publisher={TUP}
}

@article{LIN2026104762,
title = {Exploring influential factors of fleet and parking management in shared autonomous vehicle systems: An agent-based simulation framework},
journal = {Transportation Research Part A: Policy and Practice},
volume = {203},
pages = {104762},
year = {2026},
issn = {0965-8564},
doi = {https://doi.org/10.1016/j.tra.2025.104762},
url = {https://www.sciencedirect.com/science/article/pii/S0965856425003957},
author = {Yuqian Lin and Kenan Zhang and Daniel Kondor and Zhan Zhao and Carlo Ratti and Yang Xu}
}

@article{Siddiq2021,
author = {Siddiq, Auyon and Taylor, Terry A.},
title = {Ride-Hailing Platforms: Competition and Autonomous Vehicles},
journal = {Manufacturing \& Service Operations Management},
volume = {24},
number = {3},
pages = {1511-1528},
year = {2022},
doi = {10.1287/msom.2021.1013},
URL = {  https://doi.org/10.1287/msom.2021.1013},
eprint = { https://doi.org/10.1287/msom.2021.1013}
}

@article{mourad2019owning,
  title={Owning or sharing autonomous vehicles: comparing different ownership and usage scenarios},
  author={Mourad, Abood and Puchinger, Jakob and Chu, Chengbin},
  journal={European Transport Research Review},
  volume={11},
  number={1},
  pages={31},
  year={2019},
  publisher={Springer}
}

@article{LI2025103871,
title = {Fleet sizing and pricing for hybrid ownership of shared autonomous vehicles in a multimodal transportation system},
journal = {Transportation Research Part E: Logistics and Transportation Review},
volume = {193},
pages = {103871},
year = {2025},
issn = {1366-5545},
doi = {https://doi.org/10.1016/j.tre.2024.103871},
url = {https://www.sciencedirect.com/science/article/pii/S1366554524004629},
author = {Qing Li and Feixiong Liao and Wei Xu and Hai-Jun Huang}
}

@article{Senlei,
author = {Wang,  Senlei and Correia, Goncalo Homem de Almeida and Lin, Haixiang},
title = {Exploring the Performance of Different On-Demand Transit Services Provided by a Fleet of Shared Automated Vehicles: An Agent-Based Model},
journal = {Journal of Advanced Transportation},
volume = {2019},
number = {1},
pages = {7878042},
doi = {https://doi.org/10.1155/2019/7878042},
url = {https://onlinelibrary.wiley.com/doi/abs/10.1155/2019/7878042},
eprint = {https://onlinelibrary.wiley.com/doi/pdf/10.1155/2019/7878042},
year = {2019}
}

@article{ayetor2023comparing,
  title={Comparing the cost per mile of electric vehicles and internal combustion engine vehicles in Ghana},
  author={Ayetor, Godwin Kafui and Dzebre, Denis Kwame Edem and Mensah, Lena Dzifa and Boahen, Samuel and Amoabeng, Kofi Owura and Tay, Godwin Fabiola Kwaku},
  journal={Transportation Research Record},
  volume={2677},
  number={5},
  pages={682--693},
  year={2023},
  publisher={SAGE Publications Sage CA: Los Angeles, CA}
}

@misc{chicago_tnp_trips_2023,
  title        = {Transportation Network Providers - Trips (2023-2024)},
  author       = {{City of Chicago}},
  year         = {2026},
  howpublished = {\url{https://data.cityofchicago.org/Transportation/Transportation-Network-Providers-Trips-2023-2024-/n26f-ihde/about_data}},
  note         = {Accessed: 2026-01-19}
}

@article{valente2019sharing,
  title={Sharing economy: becoming an Uber driver in a developing country},
  author={Valente, Eduardo and Patrus, Roberto and C{\'o}rdova Guimar{\~a}es, Rosana},
  journal={Revista de Gest{\~a}o},
  volume={26},
  number={2},
  pages={143--160},
  year={2019},
  publisher={Emerald Publishing Limited}
}

@article{baldick2000linear,
  title={Linear supply function equilibrium: Generalizations, application, and limitations},
  author={Baldick, Ross and Grant, Ryan and Kahn, Edward Paul and others},
  year={2000},
  publisher={University of California Energy Institute Berkeley, CA}
}

@article{doi:10.1177/10591478251319683,
author = {Benjamin Legros and Johan SH van Leeuwaarden and Jan C Fransoo},
title ={Managing Reusable Resources With Usage Time Limits},
journal = {Production and Operations Management},
volume = {34},
number = {8},
pages = {2413-2429},
year = {2025},
doi = {10.1177/10591478251319683},
URL = { 
        https://doi.org/10.1177/10591478251319683
},
eprint = { 
        https://doi.org/10.1177/10591478251319683
},
}

@article{doi:10.1111/poms.13883,
author = {Huiqi Guan and Xin Geng and Haresh Gurnani},
title ={Peer‐to‐peer sharing platforms with quality differentiation: Manufacturer's strategic decision under sharing economy},
journal = {Production and Operations Management},
volume = {32},
number = {2},
pages = {485-500},
year = {2023},
doi = {10.1111/poms.13883},
URL = {https://doi.org/10.1111/poms.13883},
eprint = { https://doi.org/10.1111/poms.13883}
}

@article{doi:10.1111/poms.13491,
author = {Zenan Zhou and Xiang Wan},
title ={Does the Sharing Economy Technology Disrupt Incumbents? Exploring the Influences of Mobile Digital Freight Matching Platforms on Road Freight Logistics Firms},
journal = {Production and Operations Management},
volume = {31},
number = {1},
pages = {117-137},
year = {2022},
doi = {10.1111/poms.13491},
URL = { https://doi.org/10.1111/poms.13491},
eprint = { https://doi.org/10.1111/poms.13491}
}

@article{doi:10.1177/10591478251331407,
author = {Yu Zhang and Min Huang and Xiaohang Yue and Lin Tian},
title ={Competitive Manufacturers Product-Sharing Strategies: P2P Sharing or B2C Sharing?},
journal = {Production and Operations Management},
volume = {34},
number = {12},
pages = {4060-4078},
year = {2025},
doi = {10.1177/10591478251331407},
URL = { https://doi.org/10.1177/10591478251331407},
eprint = { https://doi.org/10.1177/10591478251331407}
}

@article{doi:10.1177/10591478251349724,
author = {Yihang Yang and Yimin Yu and Qian Wang and Junming Liu},
title ={Fleet Repositioning for Vehicle Sharing Systems: Asymptotic Optimality of the Balanced Myopic Policy},
journal = {Production and Operations Management},
volume = {35},
number = {2},
pages = {566-585},
year = {2026},
doi = {10.1177/10591478251349724},
URL = { https://doi.org/10.1177/10591478251349724},
eprint = { https://doi.org/10.1177/10591478251349724}
}

@article{doi:10.1177/10591478251404236,
author = {Xiaokun Wu and Shinyi Wu and Zhongju Zhang},
title ={An Economic Analysis of Subscription Sharing of Digital Services},
journal = {Production and Operations Management},
volume = {0},
number = {0},
pages = {10591478251404236},
year = {2025},
doi = {10.1177/10591478251404236},
URL = {  https://doi.org/10.1177/10591478251404236},
eprint = { https://doi.org/10.1177/10591478251404236}
}

@article{doi:10.1177/10591478251326390,
author = {Wenhui Zhou and Yanhong Gan and Weixiang Huang and Pengfei Guo},
title ={Group Service or Individual Service? Differentiated Pricing for a Private or Public Service Provider},
journal = {Production and Operations Management},
volume = {34},
number = {9},
pages = {2775-2792},
year = {2025},
doi = {10.1177/10591478251326390},
URL = {  https://doi.org/10.1177/10591478251326390},
eprint = {   https://doi.org/10.1177/10591478251326390}
}

@article{doi:10.1287/mnsc.1100.1203,
author = {Bassamboo, Achal and Randhawa, Ramandeep S. and Zeevi, Assaf},
title = {Capacity Sizing Under Parameter Uncertainty: Safety Staffing Principles Revisited},
journal = {Management Science},
volume = {56},
number = {10},
pages = {1668-1686},
year = {2010},
doi = {10.1287/mnsc.1100.1203},
URL = {  https://doi.org/10.1287/mnsc.1100.1203},
eprint = { https://doi.org/10.1287/mnsc.1100.1203}
}

@article{doi:10.1287/opre.2019.1916,
author = {Dong, Jing and Ibrahim, Rouba},
title = {Managing Supply in the On-Demand Economy: Flexible Workers, Full-Time Employees, or Both?},
journal = {Operations Research},
volume = {68},
number = {4},
pages = {1238-1264},
year = {2020},
doi = {10.1287/opre.2019.1916},
URL = {  https://doi.org/10.1287/opre.2019.1916},
eprint = {  https://doi.org/10.1287/opre.2019.1916}
}

@article{doi:10.1177/10591478261444827,
author = {Jianyue Wang and Ki Ling Cheung and Albert Y. Ha},
title ={Apparel retail and rental business models and their sustainability implications},
journal = {Production and Operations Management},
volume = {0},
number = {0},
pages = {10591478261444827},
year = {2026},
doi = {10.1177/10591478261444827},
URL = {  https://doi.org/10.1177/10591478261444827},
eprint = {  https://doi.org/10.1177/10591478261444827  }
}

@article{bogenberger2026role,
  title={The role of humanoid robots in future public transport systems},
  author={Bogenberger, Klaus and Niels, Tanja},
  journal={npj Sustainable Mobility and Transport},
  volume={3},
  number={1},
  pages={21},
  year={2026},
  publisher={Nature Publishing Group UK London}
}

@techreport{Stocker2017Shared,
address = {Paris},
author = {Adam Stocker and Susan Shaheen},
copyright = {http://www.econstor.eu/dspace/Nutzungsbedingungen},
language = {eng},
number = {2017-09},
publisher = {Organisation for Economic Co-operation and Development (OECD), International Transport Forum},
title = {Shared automated vehicles: Review of business models},
type = {International Transport Forum Discussion Paper},
url = {https://hdl.handle.net/10419/194044},
year = {2017}
}

@article{Altug03062022,
author = {Mehmet Sekip Altug and Oben Ceryan},
title = {Optimal dynamic allocation of rental and sales inventory for fashion apparel products},
journal = {IISE Transactions},
volume = {54},
number = {6},
pages = {603--617},
year = {2022},
publisher = {Taylor \& Francis},
doi = {10.1080/24725854.2021.1982157},
URL = {  https://www.tandfonline.com/doi/abs/10.1080/24725854.2021.1982157},
eprint = {    https://www.tandfonline.com/doi/pdf/10.1080/24725854.2021.1982157}
}

@article{VOGL2025858,
title = {Examining location factors of coworking spaces in peripheral areas: an empirical study of rural coworking in Germany},
journal = {Facilities},
volume = {43},
number = {11},
pages = {858-880},
year = {2025},
issn = {0263-2772},
doi = {https://doi.org/10.1108/F-09-2024-0125},
url = {https://www.sciencedirect.com/science/article/pii/S0263277225000157},
author = {Thomas Vogl and Grzegorz Micek}
}

@article{chen2024real,
  title={Real-time spatial--intertemporal pricing and relocation in a ride-hailing network: Near-optimal policies and the value of dynamic pricing},
  author={Chen, Qi and Lei, Yanzhe and Jasin, Stefanus},
  journal={Operations Research},
  volume={72},
  number={5},
  pages={2097--2118},
  year={2024},
  publisher={INFORMS}
}

@ARTICLE{9907878,
  author={Wang, Xing and Wang, Ling and Wang, Shengyao and Pan, Jize and Ren, Hao and Zheng, Jie},
  journal={IEEE Transactions on Intelligent Transportation Systems}, 
  title={Recommending-and-Grabbing: A Crowdsourcing-Based Order Allocation Pattern for On-Demand Food Delivery}, 
  year={2023},
  volume={24},
  number={1},
  pages={838-853},
  doi={10.1109/TITS.2022.3209722}
}

@ARTICLE{10609802,
  author={Paparella, Fabio and Hofman, Theo and Salazar, Mauro},
  journal={IEEE Transactions on Intelligent Transportation Systems}, 
  title={Electric Autonomous Mobility-on-Demand: Jointly Optimal Vehicle Design and Fleet Operation}, 
  year={2024},
  volume={25},
  number={11},
  pages={17054-17065},
  doi={10.1109/TITS.2024.3428569}
}

@article{doi:10.1287/trsc.2022.1188,
author = {Xie, Jiaohong and Liu, Yang and Chen, Nan},
title = {Two-Sided Deep Reinforcement Learning for Dynamic Mobility-on-Demand Management with Mixed Autonomy},
journal = {Transportation Science},
volume = {57},
number = {4},
pages = {1019-1046},
year = {2023},
doi = {10.1287/trsc.2022.1188},
URL = {     https://doi.org/10.1287/trsc.2022.1188},
eprint = {   https://doi.org/10.1287/trsc.2022.1188}
}

@article{HALL2026104221,
title = {Ride-hailing and urban transportation: Evidence and policy},
journal = {Regional Science and Urban Economics},
pages = {104221},
year = {2026},
issn = {0166-0462},
doi = {https://doi.org/10.1016/j.regsciurbeco.2026.104221},
url = {https://www.sciencedirect.com/science/article/pii/S0166046226000311},
author = {Jonathan D. Hall}
}

@report{Uber2024AnnualReport,
  author = {{Uber Technologies, Inc.}},
  title = {Annual Report 2024 (Form 10-K)},
  year = {2025},
  institution = {Uber Technologies, Inc.},
  url = {https://www.sec.gov/Archives/edgar/data/1543151/000155278125000101/e25100_uber-ars.pdf},
  note = {Accessed 2026-06-30}
}

@article{Yu2026Robotaxis,
  author = {Yu, Zhen and Wang, Jingyu and Zuo, Ting and Alm, James and Li, Xun and Li, Xue},
  title = {Robotaxis reduce taxi drivers’ income},
  journal = {Humanities and Social Sciences Communications},
  year = {2026},
  volume = {13},
  number = {1},
  pages = {629},
  doi = {10.1057/s41599-026-07366-x},
  url = {https://doi.org/10.1057/s41599-026-07366-x},
  issn = {2662-9992}
}

@article{Othman2022,
  author    = {Othman, Kareem},
  title     = {Exploring the implications of autonomous vehicles: a comprehensive review},
  journal   = {Innovative Infrastructure Solutions},
  year      = {2022},
  volume    = {7},
  number    = {2},
  pages     = {165},
  month     = {Mar},
  doi       = {10.1007/s41062-022-00763-6},
  url       = {https://doi.org/10.1007/s41062-022-00763-6},
  issn      = {2364-4184}
}

@article{doi:10.1111/poms.12672,
author = {Yiwei Chen and Retsef Levi and Cong Shi},
title ={Revenue Management of Reusable Resources with Advanced Reservations},

journal = {Production and Operations Management},
volume = {26},
number = {5},
pages = {836-859},
year = {2017},
doi = {10.1111/poms.12672},
URL = { 
        https://doi.org/10.1111/poms.12672
},
eprint = { 
        https://doi.org/10.1111/poms.12672
}
}

@article{doi:10.1177/00222437241255057,
author = {Ludovic Stourm and Paulo Albuquerque},
title ={Flowers and Bees: Spatial Network Effects in the Adoption of a Sharing-Economy Platform},
journal = {Journal of Marketing Research},
volume = {61},
number = {6},
pages = {1015-1040},
year = {2024},
doi = {10.1177/00222437241255057},
URL = { 
        https://doi.org/10.1177/00222437241255057},
eprint = {  https://doi.org/10.1177/00222437241255057}
}

@article{hawkins2026waymo170miles,
  author    = {Andrew J. Hawkins},
  title     = {Waymo hits 170 million miles while avoiding serious mayhem},
  journal   = {The Verge},
  year      = {2026},
  month     = {March},
  url       = {https://www.theverge.com/transportation/896837/waymo-170-million-miles-safety-crashes-injuries},
  note      = {Accessed: 2026-07-02}
}

\newpage

\appendix
\section{Proof of Proposition \ref{prop1}}\label{app:proof_prop1}

For the ease of interpretation, we keep the same notations of variables used above but redefine them with dimensions; specifically, all primary and auxiliary decision variables are reformulated within a vector space. First and foremost, the primary decision variables are the rental price $P\in \mathbb{R}^{HJ}$ and the vehicle flows $Y\in \mathbb{R}^{KTJ^2}$, where $K=|\mathcal{K}|$. We introduce a temporal expansion matrix $A^0 \in \{0, 1\}^{TJ \times HJ}$ to relate the hourly rental price $P$ to the trip-time dimension price.
Specifically, 
\begin{align*}
    A^0 = \begin{bmatrix}
    \mathbf{B} & & \\
    & \ddots & \\
    & & \mathbf{B}
    \end{bmatrix}
\end{align*}
where each block $\mathbf{B}$ is $\Delta_s / \Delta_t$ vertically stacked identity matrices $I_J$. Correspondingly, we define the private AV supply vector as $N\in \mathbb{R}^{(K-1)TJ}$ and 
rearrange Eq.~\eqref{eq:linear-private-av-supply} as 
\begin{align}
    A^1 A^0 P - N = b^1, \label{eq_A1}
\end{align}
where $A^1 \in \mathbb{R}^{(K-1)TJ \times TJ}$ is a diagonal matrix with elements $A^1_{\theta tj, \theta tj} = M^\theta \Omega^\theta_{hj} / (\overline{\varepsilon}-\underline{\varepsilon})$. The vector $b^1 \in \mathbb{R}^{(K-1)TJ}$ consists of elements $b^1_{\theta tj} = M^\theta \Omega^\theta_{hj} \underline{\varepsilon} / (\overline{\varepsilon}-\underline{\varepsilon})$.

The net supply of private AVs $Z\in \mathbb{R}^{(K-1)TJ}$ previously defined in Eqs.~\eqref{eq:net-private-av-supply-1}-\eqref{eq:net-private-av-supply-t} is then expressed as 
\begin{align}
    A^2  N - Z = 0,
\end{align}
where $A^2 \in \{-1, 0, 1\}^{(K-1)TJ \times (K-1)TJ}$ is a shift mapping with elements given by
\begin{align*}
    A^2_{\theta t j, \theta' t' j'} =
    \begin{cases}
        1, & \text{if } \theta'=\theta, t'=t, j'=j; \\
        -1, & \text{if } \theta'=\theta, t'=t-1, j'=j \text{ and } t > 1; \\
        0, & \text{otherwise.}
    \end{cases}
\end{align*}

Corresponding to the dispatching flows $Y\in\mathbb{R}^{KTJ^2}$, we redefine the arrival flows $V\in\mathbb{R}^{KTJ^2}$ and rewrite Eq.~\eqref{eq:arrival-flow} as
\begin{align}
    A^3Y - V = 0, 
\end{align}
where $A^3\in \{0,1\}^{KTJ^2 \times KTJ^2}$ is a binary mapping with each element $A^3_{k\tau(t)ij, ktij}=1$ and zero otherwise. 

Next, we define the vector form of empty vehicle flows $E\in\mathbb{R}^{KTJ}$ and reorganize their class-specific dynamics Eqs.~\eqref{eq:empty-flow-dynamics-private-av} and \eqref{eq:empty-flow-dynamics-platform-av} as 
\begin{align}
    \hat{A}^2 E - A^\text{4d}V + A^\text{4o}Y - A^5Z = 0, 
\end{align}
\noindent The shift mapping $\hat{A}^2 \in \{-1, 0, 1\}^{KTJ \times KTJ}$ is an expansion of $A^2$ to accommodate all $K=|\mathcal{K}|$ vehicle classes. Its elements $\hat{A}^2_{ktj, k't'j'}$ are $1$ if $ktj=k't'j'$, and $-1$ if $k=k', t'=t-1, j=j'$ for $t>1$. $A^\text{4o},A^\text{4d}\in \{0,1\}^{KTJ \times KTJ^2}$ aggregate flows by origin and destination, respectively. Specifically, $A^\text{4o}_{ktj, ktj'j''} = 1$ if $j=j'$ (summing flows leaving $j$), and $A^\text{4d}_{ktj, ktj'j''} = 1$ if $j=j''$ (summing flows arriving at $j$). $A^5 \in \{0, 1\}^{KTJ \times (K-1)TJ}$ is a block-structured matrix defined as:
\begin{align*}
    A^5 = 
    \begin{bmatrix} 
        I_{(K-1)TJ} \\ \mathbf{0}_{TJ \times (K-1)TJ} 
    \end{bmatrix},
\end{align*}
where $I$ is the identity matrix, effectively mapping the private AV net supply $Z$ to the first $K-1$ vehicle classes while assigning zero external supply to the platform-owned AV class (the $K$-th class).

The feasibility constraint Eq.~\eqref{eq:feasible-dispatch-flow} is rearranged as the following linear inequality constraint:
\begin{align}
    A^\text{4o}Y - E \leq 0.
\end{align}

The same-settlement for private AVs Eq.~\eqref{eq:private-av-settle} and the initial and terminal conservation constraints for private AVs Eqs.~\eqref{eq:private-av-conservation-1}-\eqref{eq:private-av-conservation-T} are written as follows:
\begin{align}
    & A^6 Y = 0, \\
    & A^7 E + A^\text{8d}V - A^\text{8o} Y - A^7 N = 0
\end{align}
where $A^6 \in \{0, 1\}^{KTJ^2 \times KTJ^2}$ is a diagonal selection matrix. Its diagonal element corresponding to the index $(k, t, i, j)$ is $1$ if $t > T - \pi_{ij}$ and $k \in \{1, \dots, K-1\}$, and $0$ otherwise.  $A^7 \in \{0, 1\}^{2(K-1)J \times (K-1)TJ}$ is a temporal boundary selector. Its elements are defined as follows: for a row indexed by a triplet $(\text{tag}, \theta, j)$ where $\text{tag} \in \{1, T\}$, the entry at column $(\theta', t', j')$ is $1$ iff $\theta' = \theta, \ j' = j, \  \text{and} \ t' = \text{tag}$; otherwise, the entry is $0$. $A^{8\text{o}}, A^{8\text{d}} \in \{0,1\}^{2(K-1)J \times (K-1)TJ^2}$ are the boundary flow mapping matrices. For any row indexed by $(\text{tag}, \theta, j)$ with $\text{tag} \in \{1, T\}$ and column indexed by $(\theta', t', j', j'')$, the element of $A^{8\text{o}}$ is $1$ iff $\text{tag} = T$, $\theta' = \theta$, $t' = T$, and $j' = j$, and the element of $A^{8\text{d}}$ is $1$ iff $\text{tag} = T$, $\theta' = \theta$, $t' = T$, and $j'' = j$. Otherwise, all elements are zero.

The fleet size constraint for platform-owned AVs Eq.~\eqref{eq:platform-av-fleet-size} are rewritten as follows:
\begin{align}
    A^{9\text{e}} E + A^{9\text{y}} Y = b^2,
\end{align}
where $A^{9\text{e}} \in \{0,1\}^{1 \times KTJ}$ and $A^{9\text{y}} \in \{0,1\}^{1 \times KTJ^2}$ extract the terminal states of platform-owned AVs ($k=K$), and $b^2 = N^\nu$. Specifically, $A^{9\text{e}}_{ktj} = 1$ iff $k=K$ and $t=T$, and $A^{9\text{y}}_{ktij} = 1$ iff $k=K$ and $T-\pi_{ij}\le t < T$, accounting for en-route platform-owned vehicles. 

With all above, we are able to express Constraints \eqref{eq:private-av-supply}-\eqref{eq:en-route-platform-AV} in a compact matrix form of linear constraints. Moreover, we can redefine the demand loss by introducing $W\in \mathbb{R}^{TJ^2}$ as an auxiliary variable. Accordingly, Eq.~\eqref{eq:demand-loss} is equivalent the following two inequality constraints:
\begin{align}
    &-W - A^{10} Y \leq - D,\\ 
    &-W \leq 0,\label{eq_W}
\end{align}
where $A^{10} \in \{0,1\}^{TJ^2 \times KTJ^2}$ is a class-aggregation matrix. Its element $(A^{10})_{tij, kt'i'j'}$ is $1$ iff $t=t', i=i', j=j'$ for all $k \in \mathcal{K}$, and $0$ otherwise; and $D\in \mathbb{R}^{TJ^2}$ is the rearranged demand vector.

We can now explicitly write out the equality constraints $A^\text{eq}x = b^\text{eq}$, which are structured as follows:
    \begin{align*}
    A^\text{eq} = 
        \begin{bmatrix}
            A^1 A^0 & 0 & -I & 0 & 0 & 0 & 0 \\ 
            0 & 0 & A^2 & -I & 0 & 0 & 0 \\ 
            0 & A^3 & 0 & 0 & -I & 0 & 0 \\ 
            0 & A^\text{4o} & 0 & -A^5 & -A^\text{4d} & \hat{A}^2 & 0 \\ 
            0 & A^6 & 0 & 0 & 0 & 0 & 0 \\ 
            0 & -A^{8\text{o}} & -A^7 & 0 & A^{8\text{d}} & A^7 & 0 \\ 
            0 & A^{9\text{y}} & 0 & 0 & 0 & A^{9\text{e}} & 0    
        \end{bmatrix}, \quad
    b^\text{eq} = \begin{bmatrix} b^1 \\ 0 \\ 0 \\ 0 \\ 0 \\ 0 \\ b^2 \end{bmatrix};
    \end{align*}
    and the inequality constraints $A^\text{neq}x \leq b^\text{neq}$ represent the capacity feasibility and demand loss:
    \begin{align*}
    A^\text{neq} = \begin{bmatrix}
    0 & A^\text{4o} & 0 & 0 & 0 & -I & 0  \\        
    0 & -A^{10} & 0 & 0 & 0 & 0 & -I \\ 
    -I & 0 & 0 & 0 & 0 & 0 & 0 \\
    0 & -I & 0 & 0 & 0 & 0 & 0 \\
    0 & 0 & -I & 0 & 0 & 0 & 0 \\
    0 & 0 & 0 & 0 & -I & 0 & 0 \\
    0 & 0 & 0 & 0 & 0 & -I & 0 \\
    0 & 0 & 0 & 0 & 0 & 0 & -I 
    \end{bmatrix}, \quad
    b^\text{neq} = \begin{bmatrix} 0 \\ -D \\ 0 \\ 0 \\ 0 \\ 0 \\ 0 \end{bmatrix}.
    \end{align*}
    where the matrices $A^1 \dots A^{10}$, vectors $b^1, b^2$ and demand $D$ are given in constraints~\eqref{eq_A1}--\eqref{eq_W}.

In the optimization problem (\ref{eq:planning-qp}), the matrix $Q$ and vector $c$ in the objective function are specified as 
\begin{align}
    Q &= \left[
    \begin{array}{llll}
         0_{d(N), d(P\wedge Y)} & I_{d(N)} & 0_{d(N), d(x) -d(P\wedge Y\wedge N)} \\
         0_{d(x)-d(N), d(P\wedge Y)} & 0_{d(x)-d(N), d(N)} & 0_{d(x)-d(N), d(x) -d(P\wedge Y\wedge N)}
    \end{array}\right],\\
    c &= \text{col}(
    \begin{array}{llll}
    0_{d(P)},& c_0\Pi,& 0_{d(x)-d(P\wedge Y\wedge W)},& c_\rho
    \end{array}),
\end{align}
where $d(\cdot)$ denotes the vector dimension, $\wedge$ refers to the concatenation of vectors, and the travel and penalty matrices are both rearranged into vectors $\Pi\in\mathbb{R}_+^{TJ^2}, c_\rho\in\mathbb{R}_+^{TJ^2}$.

\end{document}